%% file: main.tex
\documentclass[fleqn,10pt]{article}
\usepackage{jheppub} 
\usepackage{lmodern}
\usepackage{amssymb}
\usepackage{textcomp}
\usepackage{caption}
\usepackage{subcaption}
\usepackage{blindtext}
\usepackage{graphicx}
\usepackage{enumitem}
\usepackage[utf8]{inputenc}
\usepackage{verbatim}
\usepackage{listings}
\usepackage[english]{babel}
\usepackage{overpic}
\usepackage{etoc}
\usepackage{float}
\usepackage{amsfonts}
\usepackage{fancyvrb}
\usepackage{xcolor}
\usepackage{amssymb}
\usepackage{imakeidx}
\usepackage{amsmath}
\usepackage{xfrac}
\usepackage{tikz-cd}
\usepackage{pgfplots}
\usepackage{amsthm}
\usepackage{multirow}
\usepackage{hyperref}
\usepackage{titlesec}
\usepackage{csquotes}
\usepackage[ruled,vlined]{algorithm2e} 
\allowdisplaybreaks

\pgfplotsset{compat=1.18}

\usepackage[
    style=numeric,
    citestyle=numeric,
    backend=biber,
    maxcitenames=1,
    maxalphanames=1,
    uniquelist=false,
]{biblatex}

\usepackage{bm}
\usepackage{changepage}

\usepackage[font={footnotesize,color=black}]{caption}

\usepackage{mathtools} 

\newtheorem{lemma}{Lemma}[section]

\newtheorem{remark}{Remark}
\newtheorem{defin}{Definition}

\newcommand{\s}{{ \boldsymbol \sigma}}
\newcommand{\z}{{ \boldsymbol z}}

\newcommand{\SOMMA}[2]{\displaystyle\sum\limits_{#1}^{#2}}

\def\z{z^{\dag}}
\def\w{w^{\dag}}
\def\s{s^{\dag}}

\def\pz{p^{\z}}
\def\ps{p^{\s}}
\def\pw{p^{\w}}

\def\pdags{\bar{p}^{\s}}
\def\pdagw{\bar{p}^{\w}}
\def\qsigma{\bar{q}^{\sigma}}
\def\qphi{\bar{q}^{\phi}}
\def\qtau{\bar{q}^{\tau}}

\def\a{\tilde{a}}

\def\r{\rangle}
\def\l{\langle}

\def\P{\bar{P}}
\def\q{\bar{q}}

\def\b{\beta^{'}}
\def\P{\mathbb{P}}

\DeclareMathOperator{\sign}{sign}
\DeclareMathOperator{\erf}{erf}

\def\a{g_{\sigma\tau}}
\def\b{g_{\sigma\phi}}
\def\c{g_{\tau\phi}}

\title{\boldmath Supervised and Unsupervised protocols for hetero-associative neural networks}

\author[a,b,e]{Andrea Alessandrelli,}
\author[c,d,e]{Adriano Barra,}
\author[f]{Andrea Ladiana,}
\author[f]{Andrea Lepre,}
\author[g,h,e]{Federico Ricci-Tersenghi.}

\affiliation[a]{Dipartimento di Informatica, Università di Pisa, Pisa Italy.}

\affiliation[b]{Istituto Nazionale di Fisica Nucleare, Sezione di Lecce, Italy.}

\affiliation[c]{Dipartimento di Scienze di Base Applicate all'Ingegneria, Sapienza Universit\`a di Roma, Rome, Italy.}

\affiliation[d]{Istituto Nazionale di Fisica Nucleare, Sezione di Roma1, Italy}

\affiliation[e]{Istituto Nazionale d'Alta Matematica, GNFM, Roma, Italy.}

\affiliation[f]{Dipartimento di Matematica e Fisica, Università del Salento, Lecce, Italy.}

\affiliation[g]{Dipartimento di Fisica, Sapienza Universit\`a di Roma, Roma, Italy.}

\affiliation[h]{CNR-Nanotec, Rome unit, 00185 Roma, Italy.}

\abstract{This paper introduces a learning framework for Three-Directional Associative Memory (TAM)  models, extending the classical Hebbian paradigm to both supervised and unsupervised protocols within an hetero-associative setting. These neural networks consist of three interconnected layers of binary neurons interacting via generalized Hebbian synaptic couplings that allow learning, storage and retrieval of structured triplets of patterns. By relying upon glassy statistical mechanical techniques (mainly  replica theory and Guerra interpolation), we analyze the emergent computational properties of these networks, at work with random (Rademacher) datasets and at the replica-symmetric level of description: we obtain a set of self-consistency equations for the order parameters that quantify the critical dataset sizes (i.e. their thresholds for learning) and describe the retrieval performance of these networks, highlighting the differences between supervised and unsupervised protocols. Numerical simulations validate our theoretical findings and demonstrate the robustness of the captured picture about TAMs also at work with structured datasets. In particular, this study provides insights into the cooperative interplay of layers, beyond that of the neurons within the layers, with potential implications for optimal design of artificial neural network architectures.
}

\begin{document}
\maketitle
\flushbottom

\section{Introduction} \label{sec:intro}

\input{sections/introduction.tex}

\section{Hebbian Learning for Three-Directional Associative Memories}
\label{sec:hebbLearning}
\input{sections/hebbLearning.tex}

\section{Theoretical Findings} \label{sec:analyticFindings}
\input{sections/aFindings.tex}

\section{Applications}
\label{sec:numFinding}
\input{sections/nFindings.tex}


\section{Conclusions and Future Outlooks} \label{sec:conclusion}
\input{sections/conclusions.tex}
\subsection*{Acknowledgments}
This research was supported by PNRR MUR project no. PE0000013-FAIR and by the “National Centre for HPC, Big Data and Quantum Computing - HPC”, Project CN\_00000013, CUP B83C22002940006, NRP Mission 4 Component 2 Investment 1.5,  Funded by the European Union - NextGenerationEU.
\newline
AA and AB acknowledge GNFM-INdAM (Gruppo Nazionale per la Fisica Matematica, Istituto Nazionale d’Alta Matematica), AA further acknowledges UniSalento for financial support via PhD-AI and AB further acknowledges the PRIN-2022 Project Statistical Mechanics of Learning Machines: from algorithmic and information-theoretical limits to new biologically inspired paradigms.
\newline
AB and FRT acknowledge Sapienza Università di Roma for support via the internal grant {\em Statistical learning theory for generalized Hopfield models}, prot. num. RM12419112BF7119.
\newline
AL and AL acknowledge Università degli Studi di Bari Aldo Moro for support via post-Lauream research fellowships under the ADAPTIVE‑AI project (CUP F83C24001470001, D.R. n. 123/16-01-2024) and via the Future Artificial Intelligence Research (FAIR) project (project code PE00000013, CUP H97G22000210007, Spoke 6 “Symbiotic AI”), funded by the European Union---NextGenerationEU.

\printbibliography

\appendix

\newpage
\section{Proofs in RS Assumption}
\subsection{Replica symmetric Guerra Interpolation} \label{appsec:proofGuerra}
\input{appendices/proofGuerra.tex}

\subsection{Large dataset limit \texorpdfstring{$M \gg 1$}{ }}\label{appsec:approx}
\input{appendices/largeSample.tex}

\newpage
\section{Evaluation of Momenta of the Effective Post-synaptic Potential} \label{appsec:evaluation}
\input{appendices/momentaEval.tex}

\end{document}

%% file: sections/introduction.tex
\textcolor{black}{The attractor dynamics inherent in Hopfield models serves as a cornerstone for understanding how neural networks store and retrieve patterns of information. In classical Hopfield systems, if the load and the noise are not too high, each attractor represents a stable state toward which the network converges despite the presence of perturbations. This concept is not only central to artificial neural computation but also finds significant parallels in biological processes
\cite{piermarocchi2,piermarocchi1}.}
While Hopfield networks excel at recalling individual patterns through energy minimization \cite{Amit}, real-world cognition often requires integrating distinct data streams—a capability better captured by hetero-associative models \cite{kosko1988bidirectional}.

Hetero-associative neural networks, indeed, extend the principles of auto-associative memory models \cite{AGS2,AGS1,Amit,CKS}, such as the Hopfield network \cite{Hopfield}, by enabling associations between distinct patterns rather than reinforcing individual ones.  Traditional hetero-associative networks, namely the Bidirectional Associative Memories  (BAM) introduced by Kosko in the eighties, operate by establishing pairwise connections between two neuron layers, encoding input-output mappings via Hebbian-like synaptic updates \cite{BM1}: once one of its two layers is provided with a noisy pattern, not only that layer performs pattern recognition returning the denoised input but, also, the other layer retrieves the conjugate pattern of the presented one. While effective in simple pattern association tasks, these architectures lack the capacity to process higher-order relational structures, which are crucial in domains such as cognitive modeling, language processing and hierarchical memory retrieval \cite{MPV}. A natural extension, the Three-directional Associative Memory (TAM),  overcomes this limitation by introducing a tripartite network where three layers interact via generalized Hebbian couplings, enabling the storage and retrieval of triplets of patterns, that may give rise to complex structures\footnote{We stress that several composite signals are built by triplets of elementary inputs as e.g. in audio signal processing (where chords are built of by three notes in their simplest setting) or in video signal processing (where the primary colors are three).} \cite{albanese2022replica}. This tripartite architecture outperforms w.r.t.  bipartite systems in learning, storing and retrieving structured triplets of patterns, thus facilitating advanced tasks such as pattern disentanglement and frequency modulation \cite{TAMstoring,agliari2025networks}. Recent advancements in the statistical physics of neural networks further justify the need for extending associative memory models beyond bidirectional constraints, as such extensions improve robustness and versatility \cite{BarbierCamilli2,BarbierCamilli1,FachechiICLR25,alessandrelli2025beyond}.
\par\vspace{3mm}
From a neurobiological standpoint, the functional properties of the TAM model can be understood in light of two complementary mechanisms observed in the hippocampus: pattern completion and pattern separation \cite{kandel2000principles}. The CA3 subregion relies on strong recurrent connectivity to reactivate full stored patterns from minimal cues, mirroring the TAM’s capability to reconstruct entire configurations from incomplete inputs. Conversely, the dentate gyrus, via a process often described as “expansion recoding”, is believed to map partially overlapping or noisy inputs onto sufficiently distinct neural representations, thus reducing interference and improving discriminability—an ability that the TAM architecture leverages to disentangle similar patterns. By unifying these two processes, the TAM model aligns with core principles of hippocampal function, thereby enabling robust and biologically inspired pattern recognition and pattern disentanglement.
\par\vspace{3mm}
While in \cite{TAMstoring,agliari2025networks} we introduced the TAM and inspected its retrieval capabilities as pattern recognition and pattern disentanglement, the core contribution of the present study is the systematic examination of \emph{supervised and unsupervised hetero-associative Hebbian learning protocols}, i.e. the exploration of weight's  dynamics within the TAM networks and the consequent  collective acquired retrieval capabilities by their neurons. Supervised learning —characterized by explicit labeling of input patterns— allows the network to disambiguate overlapping information through externally provided guidance \cite{alemanno2023supervised,franz2001exact,agliariRedundantRepresentation,talagrand}. In contrast, unsupervised learning relies solely on statistical correlations among observed patterns, necessitating the emergence of self-organized feature representations \cite{agliariEmergence,FachechiCoolen,FachechiEarly}. Prior works on associative memories have demonstrated that learning efficacy and retrieval quality strongly depend on the interplay between dataset structure and synaptic plasticity \cite{LenkaNature,LenkaStructure,FedeMarino,barra2023thermodynamics,randomfeature,YasserBravi}. We extend these insights to TAM by deriving self-consistency equations for the order parameters of the network that capture the model’s macroscopic behavior under varying degrees of supervision as the control parameters are made to vary.

Methodologically, our approach integrates analytical techniques from statistical mechanics of spin glasses, including replica theory \cite{MPV,Amit,Fachechi1,barra2010replica,huang2017statistical} and Guerra’s interpolation method \cite{sumrule,Guerra1,Barra-Bipartiti,Fachechi1}, to achieve an exact expression for the free energy -and, thus, self-consistencies and computational phase transitions- of TAM networks at the fairly standard replica symmetric level of description \cite{Amit,CKS}. These theoretical results are complemented by Monte Carlo simulations, which validate the predicted performance across different training scenarios. Notably, our findings highlight the emergence of cooperative dynamics between layers, wherein dataset quality disparities influence retrieval capabilities in a non-trivial manner. Specifically, layers trained with low-entropy datasets benefit from interactions with layers receiving higher-quality inputs, leading to an overall enhancement of network-wide retrieval performance. This phenomenon, termed ``\emph{cooperativeness}", challenges conventional assumptions regarding the independence of memory layers in associative networks \cite{shwartz2017opening,alessandrelli2025beyond}.

Beyond theoretical implications, our work has practical relevance for artificial neural network architectures: the ability to model structured dependencies among multiple input sources is crucial for applications such as multimodal learning, hierarchical knowledge representation and neurobiologically inspired memory systems \cite{ramsauer2020hopfield,agliariRedundantRepresentation,krotov2020large}. By rigorously quantifying the conditions (e.g. noise level, maximal storage, dataset quantity and quality) under which TAM networks exhibit optimal learning —and thus retrieval— behavior, we provide a principled foundation for designing next-generation hetero-associative memories with enhanced generalization capabilities.
\par\vspace{3mm}
The remainder of this paper is organized as follows. In Section~\ref{sec:hebbLearning}, we formally introduce the TAM model and its Hebbian learning mechanisms, detailing the differences between supervised and unsupervised formulations. Section~\ref{sec:analyticFindings} presents the analytical results, including phase diagrams and critical dataset size estimates. In Section~\ref{sec:numFinding}, we describe numerical simulations that corroborate our theoretical predictions. Finally, Section~\ref{sec:cooperativeness} deepens the spontaneous cooperation that the various layers develop when provided with different amounts of information during training: in particular, at work with pattern recognition tasks, layers that experienced less information are helped by more expert layers (that have been provided with more information) in order to preserve overall robustness of the recall.

%% file: sections/hebbLearning.tex
In the following, we provide a concise introduction to the Three-Directional Associative Memory, mirroring —and generalizing— the work \cite{TAMstoring}.
Consider three sets of binary neurons, denoted as \(\boldsymbol{\sigma}\equiv\{\sigma_{i}\}_{i=1,\ldots,N_1}\), \(\boldsymbol{\tau}\equiv\{\tau_{i}\}_{i=1,\ldots,N_2}\) and \(\boldsymbol{\phi}\equiv\{\phi_{i}\}_{i=1,\ldots,N_3}\), which interact in pairs through generalized Hebbian synaptic couplings ({\em vide infra}), allowing the storage of \(K\) triplets of patterns, i.e., \(\{\boldsymbol{\xi}^{\mu},\boldsymbol{\eta}^{\mu},\boldsymbol{\chi}^{\mu}\}_{ \mu=1,\ldots,K}\), where 
$$\boldsymbol \xi^{\mu} \in \{-1, +1\}^{N_1},\;\;\boldsymbol \eta^{\mu} \in \{-1, +1\}^{N_2},\;\;\boldsymbol \chi^{\mu} \in \{-1, +1\}^{N_3}$$ 
and their entries are Rademacher random variables drawn as
\begin{equation}
	\mathbb{P}(x_i^\mu)= \dfrac{1}{2}\delta(x_i^\mu-1)+\dfrac{1}{2}\delta(x_i^\mu+1)
\end{equation}
for $\mu = 1, \hdots, K$ and where $\bm x^\mu \in \{\bm\xi^\mu, \bm\eta^\mu,\bm\chi^\mu\}$.

The cost function of our model admits a Hamiltonian representation where the three layers of neurons mutually interact by the following prescription 
\begin{equation}\label{eq:hamTAM}
    \mathcal{H}_{\bm{N}}(\boldsymbol{\sigma},\boldsymbol{\tau},\boldsymbol {\phi}|\boldsymbol{\xi},\boldsymbol{\eta},\boldsymbol{\chi})=-\left(g_{\sigma\tau}\sum_{i ,j=1}^{N_1,N_2}J_{ij}^{\sigma\tau}\sigma_{i}\tau_{j}+g_{\sigma\phi}\sum_{i,j=1}^{N_1,N_3}J_{ij}^{ \sigma\phi}\sigma_{i}\phi_{j}+g_{\tau\phi}\sum_{i,j=1}^{N_2,N_3}J_{ij}^{\tau\phi}\tau_{i}\phi _{j}\right),
\end{equation}
    where \(\bm{g}=(g_{\sigma\tau},g_{\sigma\phi},g_{\tau\phi})\in\mathbb{R}^3\) modulates the relative strength of intra-layer interactions while the synaptic matrices are defined as follows:

    \begin{equation} \label{sinapsigmatau}
        J_{ij}^{\sigma\tau}=\frac{1}{\sqrt{N_1N_2}}\sum_{\mu=1}^{K}\xi_{i}^{\mu}\eta_{j}^{\mu};
    \end{equation}
    \begin{equation} \label{sinapsigmaphi}
        J_{ij}^{\sigma\phi}=\frac{1}{\sqrt{N_1N_3}}\sum_{\mu=1}^{K}\xi_{i}^{\mu}\chi_{j}^{\mu};
    \end{equation}
    \begin{equation} \label{sinapsitauphi}
        J_{ij}^{\tau\phi}=\frac{1}{\sqrt{N_2N_3}}\sum_{\mu=1}^{K}\eta_{i}^{\mu}\chi_{j}^{\mu}.
    \end{equation}
    These equations define the so-called  \emph{generalized}\footnote{The generalized hetero-associative Hebbian couplings defined in Eqs.~\eqref{sinapsigmatau}-\eqref{sinapsitauphi} also bear a notable resemblance to the cyclic coupling matrices \cite{personnaz1986collective,piermarocchi1}.} hetero-associative Hebbian couplings \cite{kosko1988bidirectional}. It is important to notice that the presence of the factors \(\bm{N}=(N_1,N_2,N_3)\) in the denominators ensures that the Hamiltonian remains linearly extensive in the number of neurons and that it scales correctly in the thermodynamic limit \(N_1, N_2, N_3\to\infty\), where our calculations will be carried out.  

\vspace{3mm}
\noindent
Conventional artificial learning methods—such as those employed in restricted Boltzmann machines \cite{Decelle21} (which utilize the {\em contrastive divergence} class of algorithms \cite{carreira2005contrastive,YasserBravi}), feed-forward networks \cite{hopfield1987learning} (which are trained via the {\em backpropagation} chain rule \cite{lecun1988theoretical}), and other similar techniques) generally necessitate the computation of the gradients of the loss functions with respect to the weights, as a result these approaches fundamentally rely on the evaluation of derivatives.
Conversely, \emph{Hebbian learning}—a simplified abstraction of biological learning—adopts an integral formulation: rather than depending on derivatives, it explicitly performs summation over training examples, a process inherently embedded in the corresponding cost functions. This summation-based approach leverages concentration of measure principles, adhering to the classical convergence properties established by the Central Limit Theorem (CLT).

\vspace{3mm}
\noindent
Building upon the framework established in \cite{agliariEmergence, alemanno2023supervised} for the standard Hopfield model, hereafter we introduce a theoretical formulation for Three-directional Hebbian learning, addressing both supervised and unsupervised settings (i.e., in the presence or absence of a teacher). Broadly speaking, this implies that the network is never directly exposed to the actual patterns (that we call {\em archetypes} for the sake of clearness) —which it never encounters— but is instead presented solely with noisy instances of these patterns. Consequently, the network must correctly infer the underlying archetypes hidden within these corrupted versions once {\em sufficient information} has been accumulated.

In this context the terms \(\{\boldsymbol{\xi}^{\mu},\boldsymbol{\eta}^{\mu},\boldsymbol{\chi}^{\mu}\}_{ \mu=1,\ldots,K}\) play the role of archetype and we introduce respectively $M_1, M_2$ and $M_3$ examples for each archetype of each dataset\footnote{
	Hereafter we will refer to $\bm\xi,\bm\eta,\bm\chi$ as dataset.} and we denote them as 
\begin{equation}
    \{\bm\Xi^{\mu,a_1}\}^{\mu=1,\hdots,K}_{a_1=1,\hdots,M_1},\;\; \{\bm\Theta^{\mu,a_2}\}^{\mu=1,\hdots,K}_{a_2=1,\hdots,M_2},\;\;\{\bm\Upsilon^{\mu,a_3}\}^{\mu=1,\hdots,K}_{a_3=1,\hdots,M_3};
\end{equation}
these are corrupted versions of the related archetypes, and their entries are drawn as
\begin{equation}
	\mathbb{P}(X_i^{\mu,a})= \dfrac{1+r}{2}\delta(X_i^{\mu,a}-x_i^\mu)+\dfrac{1-r}{2}\delta(X_i^{\mu,a}+x_i^\mu)
\end{equation}
where $\bm X^{\mu,a}\in \{\bm\Xi^{\mu,a_1}, \bm\Theta^{\mu,a_2}, \bm\Upsilon^{\mu,a_3}\}$, $\bm x^\mu \in \{\bm\xi^\mu, \bm\eta^\mu,\bm\chi^\mu\}$,  in such a way that $r \in [0, 1]$ assesses the training-set quality\footnote{We specify that each dataset could have its own quality, denoted by $r_k$ with $k=1,2,3$.}, that is, as $r \to 1$ the example matches perfectly the archetype, whereas for $r \to 0$ an example is, on average, orthogonal to the related
archetype. 

\noindent
The network is supplied with the three datasets $\bm\Xi$, $\bm\Theta$, and $\bm\Upsilon$ and our goal is to reconstruct the three unknown archetypes $\bm\xi$, $\bm\eta$, and $\bm\chi$.

\vspace{3mm}
\noindent
A direct implication of working with examples rather than patterns is that the Hebbian storage rule for reverberation \cite{centonze2024statistical} must be extended into a generalized Hebbian learning rule for reverberation, leading to two distinct scenarios: the supervised and the unsupervised settings.
\newline
In the supervised case, the presence of a teacher enables differentiation among various categories (i.e., different archetypes) before the gathered information is supplied to the network. Consequently, the natural choice for the synaptic matrix\footnote{For simplicity, we consider only the synaptic matrix between the layer $\bm\sigma$ and the layer $\bm\tau$ but the same generalization applies to all of them.} is thus
\begin{equation}
        J_{ij}^{\textit{(sup)}}  \sim \SOMMA{\mu=1}{K}\SOMMA{a=1}{M_1}\Xi_i^{\mu,a}\SOMMA{b=1}{M_2}\Theta_j^{\mu,b}. 
    \end{equation}
    
In the case of unsupervised Hebbian learning, the examples are directly summed up in the Hamiltonian of the model, resulting in the following form of the related synaptic matrix:

\begin{equation}
    J_{ij}^{\textit{(unsup)}} \sim \SOMMA{\mu=1}{K}\SOMMA{a=1}{M}\Xi_i^{\mu,a}\Theta_j^{\mu,a}.
\end{equation}

This expression implements the idea that a teacher — able to distinguish a priori the different categories (i.e., the different archetypes, in our jargon) to which the examples belong — is not available. Hence, the best we can do is to supply the network with all the available information directly, in the form of a linear superposition of all the examples.

In the following, we investigate the information-processing capabilities of TAM in both cases, tackling the problem from a statistical mechanics perspective.

%% file: sections/aFindings.tex
In the following, we present the cost functions (i.e. the Hamiltonians) for both the supervised and unsupervised settings and analyze the two scenarios these cost functions give rise to.

\vspace{3mm}
\noindent
The Hamiltonian (or {\em cost function}) of the TAM model in the supervised setting is given by
\begin{equation}\label{eq:H-TAM-sup}
\begin{array}{lll}
     \mathcal{H}_{\bm N, \bm M, \bm r}^{(sup)}(\bm\sigma,\bm\tau,\bm\phi|\bm\Xi,\bm\Theta,\bm\Upsilon)=-\SOMMA{\mu=1}{K}&\left(\dfrac{\a\sqrt{(1+\rho_1)(1+\rho_2)}}{\sqrt{N_1 N_2}}\SOMMA{i,j=1}{N_1,N_2}\hat\xi_i^\mu\hat\eta_j^\mu\sigma_i\tau_j\right.
     \\\\ \noalign{\vspace{-10pt}}
     &\left. + \dfrac{\b\sqrt{(1+\rho_1)(1+\rho_3)}}{\sqrt{N_1 N_3}}\SOMMA{i,j=1}{N_1,N_3}\hat\xi_i^\mu\hat\chi_j^\mu\sigma_i\phi_j\right.
     \\\\ \noalign{\vspace{-10pt}}
     &\left. + \dfrac{\c\sqrt{(1+\rho_2)(1+\rho_3)}}{\sqrt{N_2 N_3}}\SOMMA{i,j=1}{N_2,N_3}\hat\eta_i^\mu\hat\chi_j^\mu\tau_i\phi_j\right),
\end{array}
\end{equation}
where $\bm M=(M_1,M_2,M_3)$, $\bm r=(r_1,r_2,r_3)$ and
we assume $N_1 \geq N_2\geq N_3$. 
\\
\noindent
Furthermore, for compactness and practicality, we have defined
\begin{equation}
\label{eq:rho_def}
\begin{array}{llll}
        \rho_k = \dfrac{1-r_k^2}{M_k r_k^2}\;\;\;\;\mathrm{with}\;k=1,2,3
\end{array}
\end{equation}
which represents the dataset entropy (see \cite{Aquaro-Pavlov}), and

\begin{equation}
\label{eq:medie}
\begin{array}{llll}
\hat{\xi}_i^\mu= \dfrac{1}{M_1 r_1 (1+\rho_1)}\SOMMA{a=1}{M_1}\Xi_i^{\mu,a}\,;
    \\\\ \noalign{\vspace{-10pt}}
         \hat{\eta}_i^\mu= \dfrac{1}{M_2 r_2 (1+\rho_2)}\SOMMA{a=1}{M_2}\Theta_i^{\mu,a}\,;
     \\\\ \noalign{\vspace{-10pt}}
         \hat{\chi}_i^\mu= \dfrac{1}{M_3 r_3 (1+\rho_3)}\SOMMA{a=1}{M_3}\Upsilon_i^{\mu,a}\,.
\end{array}
\end{equation}
Note that this network is made for the retrieval of triplets of patterns, one per layer, but —during training— each layer may experience a different amount of information, namely the one contained in the dataset it is exposed to (and there is one dataset per layer): by accounting for this potential disparity during training, ultimately the network gives rise to cooperativity among its layers —once it is used to perform pattern recognition— such that layers trained with huge information help layers trained with poor information in retrieving the corresponding pattern, we deepen this phenomenon in Sec. \ref{sec:cooperativeness} ( see also \cite{alessandrelli2025beyond}).
\newline
The expressions in Equation \(\eqref{eq:medie}\) represent the empirically computed averages of the examples \(\Xi_i^{\mu,a}\), \(\Theta_i^{\mu,a}\), and \(\Upsilon_i^{\mu,a}\), rescaled to properly account for the factors \(M_k\), \(r_k\), and \(\rho_k\) (with \(k = 1,2,3\)). 

\vspace{3mm}
\noindent
The Hamiltonian (or {\em cost function}) of the TAM model in the unsupervised setting reads as
\begin{equation}\label{eq:H-TAM-unsup}
	\begin{array}{lll}
		& \mathcal{H}_{\bm N, M, \bm r}^{(unsup)}(\bm\sigma,\bm\tau,\bm\phi|\bm\Xi,\bm\Theta,\bm\Upsilon)=
		\\\\ \noalign{\vspace{-10pt}}
		& -\SOMMA{\mu=1}{K} \left( \dfrac{\a}{Mr_1 r_2\sqrt{(1+\rho_1)(1+\rho_2)}\sqrt{N_1 N_2}}\SOMMA{a=1}{M}\SOMMA{i,j=1}{N_1,N_2}\Xi_i^{\mu,a}\Theta_j^{\mu,a}\sigma_i\tau_j+\right.
		\\\\ \noalign{\vspace{-10pt}}
		&\quad\quad\;\; +\dfrac{\b}{Mr_1 r_3\sqrt{(1+\rho_1)(1+\rho_3)}\sqrt{N_1 N_3}}\SOMMA{a=1}{M}\SOMMA{i,j=1}{N_1,N_3}\Xi_i^{\mu,a}\Upsilon_j^{\mu,a}\sigma_i\phi_j+
		\\\\ \noalign{\vspace{-10pt}}
		& \quad\quad\; \left. +\dfrac{\c}{Mr_2 r_3\sqrt{(1+\rho_2)(1+\rho_3)}\sqrt{N_2 N_3}}\SOMMA{a=1}{M}\SOMMA{i,j=1}{N_2,N_3}\Theta_i^{\mu,a}\Upsilon_j^{\mu,a}\tau_i\phi_j\right),
	\end{array}
\end{equation}
Where, for the sake of simplicity, we set\footnote{
	To ensure comparable informational content across the three datasets, we can fix the number of samples \( M \) for each set, with \( M_1 = M_2 = M_3 \). In this context, each dataset \( i \) has examples that come from a common archetype but vary in quality, represented by the parameter \( r_i \), where \( 0 \leq r \leq 1 \). The entropy \( \rho_i = \frac{1 - r_i^2}{M r_i^2} \) effectively captures the residual uncertainty or deviation from the archetype within each data set. Holding \( M \) constant across datasets stabilizes the entropy contribution due to sample count, meaning that any variability in informational content between datasets depends solely on \( r_i \). Consequently, adjusting the quality \( r \) is sufficient to model different levels of information alignment with the archetype, so standardizing the sample count allows direct comparisons of information content based on \( r \) alone.} $M=M_1=M_2=M_3$. 
\\\\
Moreover, we have defined
\begin{equation}\label{eq:rhoUnsup_def}
	\begin{array}{llll}
		{\rho_{i j}}=\dfrac{1-r_{i}^{2}r_{j}^{2}}{M r_{i}^{2}r_{j}^{2}}\;\;\;\;\mathrm{with}\;i,j=1,2,3.
	\end{array}
\end{equation}
which will be central in the analysis of the unsupervised model.

Note that the size of the three layers and, accordingly, the length of the related patterns can differ, that is $N_1 \geq N_2 \geq N_3$, nevertheless, their number must be the same, that is $K$ patterns per each layer. Furthermore, the ratio between the amount of patterns and their length must stay finite in the asymptotic (thermodynamic) limit, namely, we choose $K,N_1,N_2,N_3$ such that
\begin{equation}
\begin{array}{lllllll}
      \lim\limits_{N_1,N_3\to\infty}\sqrt{\dfrac{N_1}{N_3}}=\alpha\,,\; && \lim\limits_{N_2,N_1\to\infty}\sqrt{\dfrac{N_1}{N_2}}=\theta \,,
      \\\\ \lim\limits_{N_2,N_3\to\infty}\sqrt{\dfrac{N_2}{N_3}}=\dfrac{\alpha}{\theta} \,,\; &&\lim\limits_{N_1\to\infty}\dfrac{K}{N_1}=\gamma\,,
\end{array}
\end{equation}
with $\alpha,\theta\in [\ 1, + \infty[\ $ and $\gamma \in \mathbb{R}^{+}$, that defines the storage of the network. 
\\
Further, denoted by $\mathcal{P}_{t}(\bm\sigma,\bm\tau,\bm\phi|\bm\xi,\bm\eta,\bm\chi)$ the probability of finding the system in a configuration $(\bm\sigma,\bm\tau,\bm\phi)$ at time $t$, the following master equation
\begin{equation}\label{dinamical}
\mathcal{P}_{t+1}(\bm{s}\vert\bm{\zeta})=\sum_{\bm{s}'\in\{-1,+1\}^{N_1+N_2+N_3}}  W_{\beta}(\bm{s}'\to\bm{s}\vert\bm{\zeta})\,\mathcal{P}_t(\bm{s}'\vert\bm{\zeta}),
\end{equation}
where \(\bm{s} = (\bm{\sigma},\bm{\tau},\bm{\phi})\), \(\bm{s}' = (\bm{\sigma}',\bm{\tau}',\bm{\phi}')\) and \(\bm{\zeta} = (\bm{\xi},\bm{\eta},\bm{\chi})\),
governs a Markov process in the space of neural configurations, where $W_{\beta}$  represents the transition rate from a state $(\bm\sigma',\bm\tau',\bm\phi')$ to a state $(\bm\sigma,\bm\tau,\bm\phi)$ and it is chosen in such a way that the system is likely to lower the value of the cost function (respectively, \eqref{eq:H-TAM-sup} in the supervised setting, \eqref{eq:H-TAM-unsup} for the unsupervised one) along its evolution: this likelihood is tuned by the parameter $\beta \in \mathbb R^+$ such that for $\beta \to 0^+$ the dynamics is a pure  random walk in the neural configuration space (and any configuration is equally likely to occur), while for $\beta \to +\infty$ the dynamics steepest descends toward the minima of the Hamiltonian \cite{CKS}. Notably, the symmetry of the pairwise couplings in the TAM cost function ensures detailed balance, thus, the long-time limit of the stochastic process \eqref{dinamical} relaxes to the Boltzmann-Gibbs distribution 
\begin{equation}\label{BGmeasure}
\lim_{t \to \infty}\mathbb{P}_t(\bm\sigma,\bm\tau,\bm\phi|\bm\Xi,\bm\Theta,\bm\Upsilon) = \mathbb{P}(\bm\sigma,\bm\tau,\bm\phi|\bm\Xi,\bm\Theta,\bm\Upsilon) = \frac{e^{-\beta \mathcal{H}_{\bm N, \bm M, \bm r}(\bm\sigma,\bm\tau,\bm\phi|\bm\Xi,\bm\Theta,\bm\Upsilon)}}{\mathcal Z^{\bm g}_{\bm N, \bm M, \bm r}(\beta,\bm J)}
\end{equation} 
where the normalization factor $\mathcal Z^{\bm g}_{\bm N, \bm M, \bm r}(\beta, \bm J)$ is given by  
 \begin{equation}\label{partition-function}
    \mathcal Z_{\bm N,\bm M,\bm r}^{\bm{g}}(\beta, \bm J) = \sum_{(\bm \sigma,\bm \tau, \bm\phi) \in \{-1, +1\}^{N_1+N_2+N_3} } e^{-\beta \mathcal{H}_{\bm N, \bm M, \bm r}(\bm\sigma,\bm\tau,\bm\phi|\bm\Xi,\bm\Theta,\bm\Upsilon)},
\end{equation}   
and is also referred to as {\em partition function}.
\\ Pivotal for a statistical mechanical analysis is the study of  the expectation of its intensive logarithm in the thermodynamic limit, namely the asymptotic quenched free energy,
defined as
\begin{equation}\label{eq:Free-Definition}
\mathcal A_{\alpha,\theta,\gamma}^{\bm{g}}(\beta) = \lim_{\bm N \to \infty}\mathbb{E}\frac{1}{L}\ln \mathcal Z^{\bm g}_{\bm N, \bm M, \bm r}(\beta, \bm J), 
\end{equation}
where $\mathbb{E}$ averages over the pattern distributions and 
\begin{equation*}
    \dfrac{1}{L} = \dfrac{1}{3}\biggl(\dfrac{1}{\sqrt{N_1N_2}} + \dfrac{1}{\sqrt{N_1N_3}} + \dfrac{1}{\sqrt{N_2N_3}}\biggr).
\end{equation*}

\noindent
Our goal is to derive an expression for $\mathcal A_{\alpha,\theta,\gamma}^{\bm{g}}(\beta)$, depending on a suitable set of macroscopic observables —i.e. the order parameters— able to quantify the emerging behavior of the system as the control parameters are made to vary. As standard for these Hopfield-like networks, we need two sets of observables that assess, respectively, the quality of retrieval and the extent of frustration leading to glassy phenomena.

Specifically, we define the archetype (ground truth) Mattis magnetizations
\begin{eqnarray} 
         m_{\xi^\mu}^\sigma = \dfrac{1}{N_1}\SOMMA{i=1}{N_1}\xi_i^\mu\sigma_i,\label{eq:ma}
         \\  
         m_{\eta^\mu}^\tau = \dfrac{1}{N_2}\SOMMA{i=1}{N_2}\eta_i^\mu\tau_i,\label{eq:ma_t}
         \\
         m_{\chi^\mu}^\phi = \dfrac{1}{N_3}\SOMMA{i=1}{N_3}\chi_i^\mu\phi_i,\label{eq:ma_p}
\end{eqnarray}
the examples Mattis magnetizations 
\begin{eqnarray} 
         n_{\Xi^{\mu,a}}^\sigma = \dfrac{1}{r_1(1+\rho_1)N_1}\SOMMA{i=1}{N_1}\Xi_i^{\mu,a}\sigma_i,\label{eq:N_singolo_a_s}
         \\  
         n_{\Theta^{\mu,a}}^\tau = \dfrac{1}{r_2(1+\rho_2)N_2}\SOMMA{i=1}{N_2}\Theta_i^{\mu,a}\tau_i,\label{eq:N_singolo_a_t}
         \\
         n_{\Upsilon^{\mu,a}}^\phi = \dfrac{1}{r_3(1+\rho_3)N_3}\SOMMA{i=1}{N_3}\Upsilon_i^{\mu,a}\phi_i,\label{eq:N_singolo_a_p}
\end{eqnarray}
and the replica overlaps for each layer
\begin{eqnarray}
         q^{\sigma}_{ab} = \dfrac{1}{N_1}\SOMMA{i=1}{N_1}\sigma_i^{(a)}\sigma_i^{(b)},
         \\ 
         q^\tau_{ab} = \dfrac{1}{N_2}\SOMMA{i=1}{N_2}\tau_i^{(a)}\tau_i^{(b)}, 
         \\
         \label{eq:qc}
         q^\phi_{ab} = \dfrac{1}{N_3}\SOMMA{i=1}{N_3}\phi_i^{(a)}\phi_i^{(b)},
\end{eqnarray}       
where $a$ and $b$ label two different replicas. 
We highlight that the Mattis magnetization $m_{x^\mu}^y$ quantifies the alignment of the layer configuration $y$ with the
related archetype $x^\mu$; $n_{x^{\mu,a}}^y$ compares the alignment of the layer configuration $y$ with one of the related archetype-examples $x^{\mu,a}$ and $q_{ab}^y$ is the standard two-replica overlap between the replicas $y^{(a)}$  and $y^{(b)}$.
For practical computations, it is also useful to define
\begin{eqnarray} 
         n_{\xi^\mu}^\sigma &=& \dfrac{1}{N_1}\SOMMA{i=1}{N_1}\hat{\xi}_i^\mu\sigma_i,\label{eq:na}
         \\  
         n_{\eta^\mu}^\tau &=& \dfrac{1}{N_2}\SOMMA{i=1}{N_2}\hat{\eta}_i^\mu\tau_i,\label{eq:na_t}
         \\
         n_{\chi^\mu}^\phi &=& \dfrac{1}{N_3}\SOMMA{i=1}{N_3}\hat{\chi}_i^\mu\phi_i,\label{eq:na_p}
\end{eqnarray}
these additional order parameters quantify the alignment of the network configuration of one layer with the average of all the corresponding examples labeled with the same $\mu$ (i.e., all examples pertaining to the $\xi^\mu$-th archetype).
\\In the following two Subsections, we investigate supervised and unsupervised learning protocols (one per Subsection) achieved by the TAM network. 

\noindent
As is standard in the bulk of statistical mechanics works on neural networks \cite{CKS}, we develop the theory at the replica-symmetric (RS) level of description: roughly speaking, we assume that all the order parameters concentrate on their means in the infinite volume limit, that is 
\begin{equation*}
    \lim_{\bm N \to \infty} \mathbb{P}(x) = \delta(x - \bar{x} ).
\end{equation*}

\subsection{Supervised setting}

\begin{figure}[t]
	\centering
	\includegraphics[width=1\textwidth]{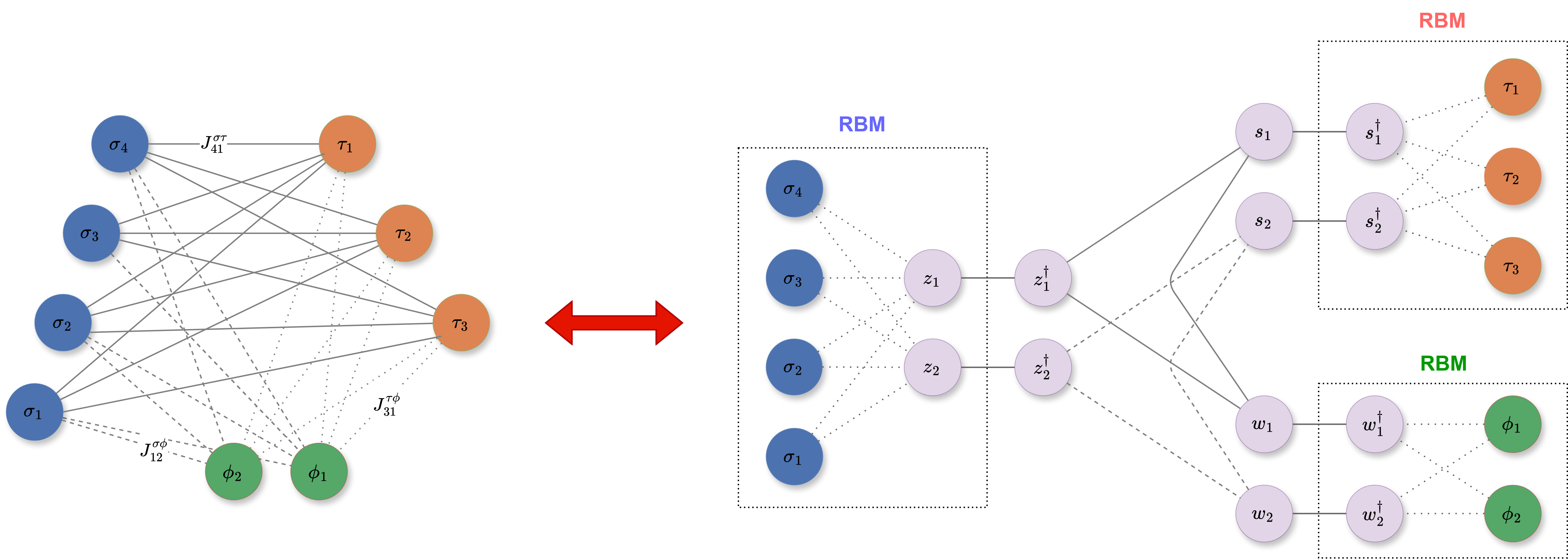}
	\caption{
    (left) Representation of the TAM neural network as given by Eq. \(\eqref{eq:hamTAM}\). In this depiction, each layer is composed of binary (Ising) units interacting pairwise in a generalized Hebbian fashion (see Eqs. \(\eqref{sinapsigmatau}\)–\(\eqref{sinapsitauphi}\) and the cost function \(\eqref{eq:hamTAM}\)). In the illustrated example, the layer \(\bm\sigma\) comprises \(N_1=4\) neurons, the layer \(\bm\tau\) \(N_2=3\) neurons, and the layer \(\bm\phi\) \(N_3=2\) neurons.  
    (right) Integral representation of the TAM neural network \cite{TAMstoring} as provided by Eq. \(\eqref{eq:rappresentazioneintegrale_sup}\). In this formulation, the three visible layers (depicted on the left) are no longer directly interconnected; instead, each is coupled with a corresponding hidden layer that governs further interactions. Notably, while the visible layers consist of standard binary neurons, the hidden layers are composed of highly selective “grandmother” units \cite{grandmother,agliariEmergence} that activate solely when the pattern they encode for is either presented or reconstructed on the associated visible layer, thereby facilitating the retrieval of pattern triples.}
\end{figure}
\vspace{3mm}
\noindent
Before proceeding with the calculation of an explicit expression of the quenched statistical pressure, we are going to use the following Gaussian approximations:
    \begin{equation}
    \label{GaussApprox}
    \begin{array}{lll}
         \sqrt{(1+\rho_1)}\hat{\xi}_i^\mu\sim J_i^\mu
         \\\\ 
         \sqrt{(1+\rho_2)}\hat{\eta}_i^\mu\sim Y_i^\mu
         \\\\ 
         \sqrt{(1+\rho_3)}\hat{\chi}_i^\mu \sim V_i^\mu
    \end{array}
\end{equation}
where $J_i^\mu, Y_i^\mu, V_i^\mu \sim\mathcal{N}(0,1)$.
\\ In fact, considering only the case of $\hat{\bm\xi}$ (the calculations for the other datasets are analogous), using the explicit expression of $\hat{\xi}_i^\mu$ \eqref{eq:rho_def}
\begin{equation}
\sqrt{(1+\rho_1)}\hat{\xi}_i^\mu = \dfrac{1}{M_1r_1\sqrt{(1+\rho_1)}}\SOMMA{a=1}{M_1}\Xi_i^{\mu,a}
\end{equation}
and applying the Central Limit Theorem (CLT) to the sum over the examples, in the limit of a large number of examples (i.e., $M_1 \gg 1$), we obtain 
\begin{equation}
    \sqrt{(1+\rho_1)}\hat{\xi}_i^\mu \sim \mu_1+J_i^\mu \sqrt{\mu_2-\mu_1^2}
\end{equation}
where $J_i^\mu\sim\mathcal{N}(0,1)$ and 
\begin{equation}
    \begin{array}{lll}
         \mu_1=\mathbb{E}_{\bm\xi}\mathbb{E}_{(\bm\Xi|\bm\xi)}\left[\sqrt{(1+\rho_1)}\hat{\xi}_i^\mu\right]=\mathbb{E}_{\bm\xi}\mathbb{E}_{(\bm\Xi|\bm\xi)}\left[\dfrac{1}{M_1r_1\sqrt{(1+\rho_1)}}\SOMMA{a=1}{M_1}\Xi_i^{\mu,a}\right]=0,
         \\\\ \noalign{\vspace{-10pt}}
         \mu_2=\mathbb{E}_{\bm\xi}\mathbb{E}_{(\bm\Xi|\bm\xi)}\left[(1+\rho_1)(\hat{\xi}_i^\mu)^2\right]=\mathbb{E}_{\bm\xi}\mathbb{E}_{(\bm\Xi|\bm\xi)}\left[\dfrac{1}{M_1^2r_1^2(1+\rho_1)}\SOMMA{a,b=1}{M_1,M_1}\Xi_i^{\mu,a}\Xi_i^{\mu,b}\right]=1.
    \end{array}
\end{equation}


\noindent
Next, exploiting the fact that the Hamiltonian is a quadratic form in the magnetizations and the Boltzmann-Gibbs measure is (proportional to) the exponential of the Hamiltonian, the partition function \eqref{partition-function} admits an  integral representation that reads as\footnote{We add the terms 
\begin{equation*}
    J_1\SOMMA{i=1}{N_1}\xi_i^1\sigma_i+J_2\SOMMA{i=1}{N_2}\eta_i^1\tau_i+J_3\SOMMA{i=1}{N_3}\chi_i^1\phi_i
\end{equation*}
in the exponent of the Boltzmann factor to ``generate'' the expectation of the Mattis magnetizations
$m^\sigma_{\xi^\mu}, m^\tau_{\eta^\mu}, m^\phi_{\chi^\mu}$. To do so, we evaluate the derivative w.r.t. $\bm J$ of the quenched statistical pressure and then set $\bm J = \bm 0$.
We will use the same technique in the unsupervised context.}  

\vspace{-\baselineskip}
\enlargethispage{\baselineskip}
\begin{equation}
\label{eq:rappresentazioneintegrale_sup}
\begin{aligned}
\mathcal{Z}^{\bm g}_{\bm N, \bm M, \bm r}(\beta, \bm J) = 
& \sum_{\{\sigma\},\{\tau\},\{\phi\}} \exp\Bigg[\beta\Big(
    \a \sqrt{N_1N_2(1+\rho_1)(1+\rho_2)} n^\sigma_{\xi^1} n^\tau_{\eta^1} \\
& + \b \sqrt{N_1N_3(1+\rho_1)(1+\rho_3)} n^\sigma_{\xi^1} n^\phi_{\chi^1} \\
& + \c \sqrt{N_2N_3(1+\rho_2)(1+\rho_3)} n^\tau_{\eta^1} n^\phi_{\chi^1} \Big) \\
& + J_1 \sum_{i=1}^{N_1} \xi_i^1 \sigma_i 
+ J_2 \sum_{i=1}^{N_2} \eta_i^1 \tau_i 
+ J_3 \sum_{i=1}^{N_3} \chi_i^1 \phi_i \Bigg] \\
& \times \int \mathcal{D}(z \z\, s \s\, w \w)\,
\exp\Bigg[\sqrt{\frac{\beta}{2N_1}} \sum_{\mu>1}^{K} \sum_{i=1}^{N_1} J_i^\mu \sigma_i z_\mu \\
& + \sqrt{\frac{\beta}{2N_2}} \sum_{\mu>1}^{K} \sum_{j=1}^{N_2} Y_j^\mu \tau_j \s_\mu
+ \sqrt{\frac{\beta}{2N_3}} \sum_{\mu>1}^{K} \sum_{k=1}^{N_3} V_k^\mu \phi_k \w_\mu \Bigg]
\end{aligned}
\end{equation}

where\footnote{Note that we used the relation:
\begin{equation} \label{relationTam}
\exp\left[\beta\sum_{\mu>1}^{K}A_{\mu}B_{\mu}\right]=\prod_{\mu>1}^{K}\int \left(\frac{dx_{\mu}dx_{\mu}^{\dagger}}{(\sqrt{2\pi})^{2}}\right)\exp\left[- \frac{1}{2}x_{\mu}x_{\mu}^{\dagger}+\sqrt{\frac{\beta}{2}}A_{\mu}x_{\mu}+ \sqrt{\frac{\beta}{2}}B_{\mu}x_{\mu}^{\dagger}\right]\,.
\end{equation}}

{
\begin{equation}
\label{transformInt}
    \begin{array}{lllll}
        \displaystyle\int\mathcal{D}(z\z s\s w\w)=\prod\limits_{\mu>1}^K&\displaystyle\int \left(\dfrac{dz_\mu d\z_\mu \, ds_\mu d\s_\mu \, dw_\mu d\w_\mu}{(\sqrt{2\pi})^6}\right)\exp\left[-\dfrac{1}{2}\left(z_\mu \z_\mu +s_\mu \s_\mu+w_\mu \w_\mu \right.\right.
        \\ \noalign{\vspace{-10pt}}
        \\
        &\left.\left.- \a\z_\mu s_\mu- \b\z_\mu w_\mu- \c s_\mu w_\mu\right)\right]
    \end{array}
\end{equation}
}
\noindent
and we have used the Gaussian approximations \eqref{GaussApprox}.
    With the introduction of analog neurons in \eqref{eq:rappresentazioneintegrale_sup}, we can define the following overlaps \(\{p^{\sigma}_{ab},p^{\tau}_{ab},p^{\phi}_{ab}\}\)  and \(\{p^{\sigma^{\dagger}}_{ab},p^{\tau^{\dagger}}_{ab},p^{\phi^{\dagger}}_{ab}\}\).
 
\begin{equation}
\begin{array}{ll}
p^{z}_{ab}=\dfrac{1}{K-1}\displaystyle\sum_{\mu>1}^{K}z_{\mu}^{(a)}z_{\mu}^{(b)} \,,\quad &
p^{z^{\dagger}}_{ab}=\dfrac{1}{K-1}\displaystyle\sum_{\mu>1}^{K}{z_{\mu}^{\dagger}}^{(a)}{z_{\mu}^{\dagger}}^{(b)} \,,
\\[2ex]
p^{s}_{ab}=\dfrac{1}{K-1}\displaystyle\sum_{\mu>1}^{K}s_{\mu}^{(a)}s_{\mu}^{(b)} \,,\quad &
p^{s^{\dagger}}_{ab}=\dfrac{1}{K-1}\displaystyle\sum_{\mu>1}^{K}{s_{\mu}^{\dagger}}^{(a)}{s_{\mu}^{\dagger}}^{(b)} \,,
\\[2ex]
p^{w}_{ab}=\dfrac{1}{K-1}\displaystyle\sum_{\mu>1}^{K}w_{\mu}^{(a)}w_{\mu}^{(b)} \,,\quad &
p^{w^{\dagger}}_{ab}=\dfrac{1}{K-1}\displaystyle\sum_{\mu>1}^{K}{w_{\mu}^{\dagger}}^{(a)}{w_{\mu}^{\dagger}}^{(b)}\,.
\end{array}
\end{equation}

As detailed in App.~\ref{appsec:proofGuerra}, the linearized expression of the partition function allows us to apply Guerra's interpolation method and reach an explicit expression for the quenched free energy. 
Once achieved such an expression in terms of control and order parameters, we can extremize it w.r.t. the order parameters, obtaining a set of self-consistent equations, whose solution provides the behavior of the order parameters as a function of the control parameters. 
\begin{figure}[!t]
    \centering
    \includegraphics[width=16cm]{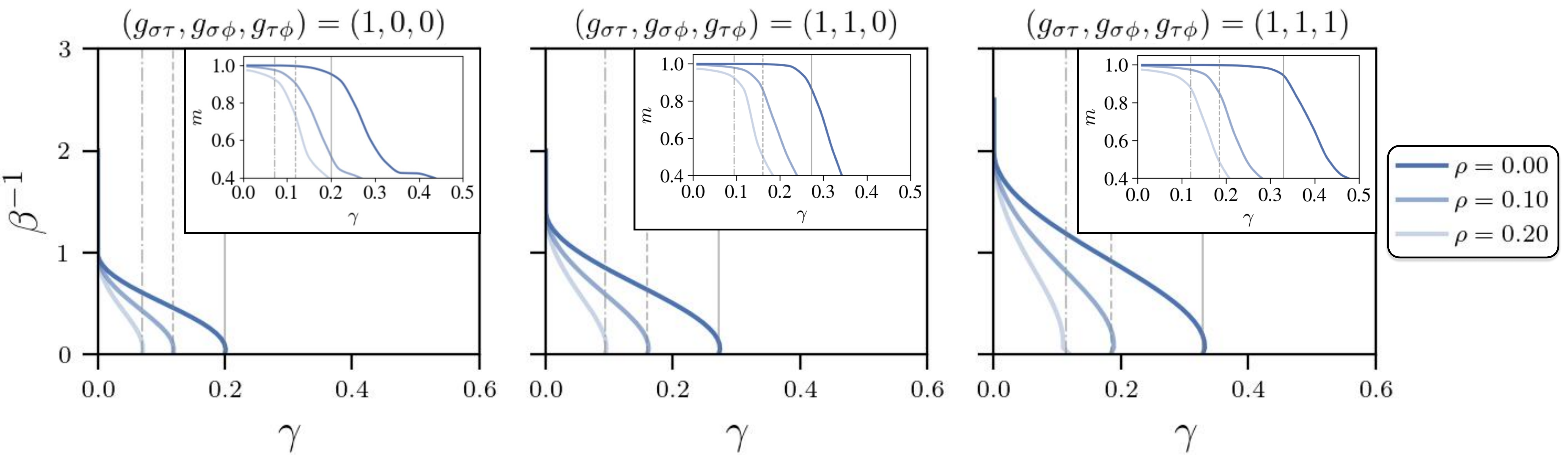}
    \caption{
    Phase diagrams of the TAM network in the supervised setting in the noise versus storage plane at $\alpha = \theta = 1$, obtained by solving numerically equations \eqref{seldefinitive} and \eqref{seldefinitive1} (only for the $\sigma$ layer, since the network is symmetric). The analysis includes different inter-layer interaction forces and various values of the entropy of the $\rho$ dataset, as indicated in the titles and legends. Each blue solid line represents the phase transition for the entire network, dividing the working region (bottom left)—where archetypes are learned and thus can be recovered and generalized—from the blackout region (top right), where spin glass effects prevail, for a specific value of the entropy of the dataset, i.e., \(\rho_1=\rho_2=\rho_3=\rho\). The retrieval region is determined by the conditions $|m^\sigma_{\xi_1^1}|, |m^\tau_{\xi_1^1}|, |m^\phi_{\xi_1^1}| \sim 1$: these constraints are all simultaneously satisfied in the region below the solid line, while above it all magnetization vanishes. The influence of $\rho$ is clearly visible: as $\rho$ increases, the recovery region gets progressively narrower in all diagrams. For $\rho = 0$, we recover the results of the standard BAM case of Kosko \cite{kosko1988bidirectional} (first panel) and the new ones related to TAM \cite{TAMstoring} (second and third panels). For simplicity and symmetry considerations, only the $\bm \sigma$ layer is shown.
    In the insets of each plot: MC simulation at zero-fast noise ($\beta^{-1} = 0$) with a symmetric network ($N_1 = N_2 = N_3 = 400$), showing the evolution of the Mattis magnetizations $m_{\xi_1}$ across the layers as a function of network load ($\gamma$) for different $\rho$. The simulations agree with theoretical predictions, correctly depicting the maximum load beyond which the network stops functioning.
    }
    \label{VariABC}
\end{figure}
Fixing the value of the control parameters and focusing on the retrieval of the triplet of patterns with label \(\mu = 1\), without loss of generality, we obtain the self-consistency equations for the order parameters by evaluating the saddle points of \eqref{eq:fRS}. However, for practical and numerical reasons, we focus on the set of self-consistent equations in the large $\bm M$ limit, as shown in Appendix~\ref{appsubsec:largeSampleSup}, where we exploit the generating function to obtain the self-consistency equations for the three archetypes-Mattis magnetizations, that read
{\small 
\begin{equation} \hspace{-0.8cm}
\label{seldefinitive}
\begin{array}{lllll}
    \bar{m}^\sigma_{\xi^1} = \mathbb{E}\left\{\tanh\left[\dfrac{\beta}{1+\rho_1}\xi^1\biggl(\tilde A_2  \bar{n}^\tau_{\eta^1}+ \tilde A_3 \bar{n}^\phi_{\chi^1}\biggr)+ z\sqrt{\dfrac{\beta^2\rho_1}{(1+\rho_1)^2}\biggl(\tilde A_2  \bar{n}^\tau_{\eta^1}+ \tilde A_3 \bar{n}^\phi_{\chi^1}\biggr)^2 +\dfrac{\beta\gamma \bar{p}^z}{2}}\right]\right\},
    \\\\ \noalign{\vspace{-10pt}}
    \bar{m}^\tau_{\eta^1} = \mathbb{E}\left\{\tanh\left[\dfrac{\beta}{1+\rho_2}\eta^1\theta\biggl(\tilde B_2 \bar{n}^\sigma_{\xi^1}+ \tilde B_3 \bar{n}^\phi_{\chi^1}\biggr)+ z\theta\sqrt{\dfrac{\beta^2\rho_2}{(1+\rho_2)^2}\biggl(\tilde B_2 \bar{n}^\sigma_{\xi^1}+ \tilde B_3 \bar{n}^\phi_{\chi^1}\biggr)^2 +\dfrac{\beta\gamma \bar{p}^{\s}}{2}}\right]\right\},
    \\\\ \noalign{\vspace{-10pt}}
    \bar{m}^\phi_{\chi^1} = \mathbb{E}\left\{\tanh\left[\dfrac{\beta}{1+\rho_3}\chi^1\alpha\biggl(\tilde C_2 \bar{n}^\sigma_{\xi^1}+ \tilde C_3\bar{n}^\tau_{\eta^1}\biggr) + z\alpha\sqrt{\dfrac{\beta^2\rho_3}{(1+\rho_3)^2}\biggl(\tilde C_2 \bar{n}^\sigma_{\xi^1}+ \tilde C_3\bar{n}^\tau_{\eta^1}\biggr)^2 +\dfrac{\beta\gamma \bar{p}^{\w}}{2}}\right]\right\},
    \\\\
    \end{array}
\end{equation}
}
and
{\small
\begin{equation} \hspace{-0.8cm}
\label{seldefinitive1}
\begin{array}{lllll}
    \qsigma = \mathbb{E}\left\{\tanh^2\left[\dfrac{\beta}{1+\rho_1}\xi^1\biggl(\tilde A_2  \bar{n}^\tau_{\eta^1}+ \tilde A_3 \bar{n}^\phi_{\chi^1}\biggr)+ z\sqrt{\dfrac{\beta^2\rho_1}{(1+\rho_1)^2}\biggl(\tilde A_2  \bar{n}^\tau_{\eta^1}+ \tilde A_3 \bar{n}^\phi_{\chi^1}\biggr)^2 +\dfrac{\beta\gamma \bar{p}^z}{2}}\right]\right\},
    \\\\ \noalign{\vspace{-10pt}}
    \qtau = \mathbb{E}\left\{\tanh^2\left[\dfrac{\beta}{1+\rho_2}\eta^1\theta\biggl(\tilde B_2 \bar{n}^\sigma_{\xi^1}+ \tilde B_3 \bar{n}^\phi_{\chi^1}\biggr)+ z\theta\sqrt{\dfrac{\beta^2\rho_2}{(1+\rho_2)^2}\biggl(\tilde B_2 \bar{n}^\sigma_{\xi^1}+ \tilde B_3 \bar{n}^\phi_{\chi^1}\biggr)^2 +\dfrac{\beta\gamma \bar{p}^{\s}}{2}}\right]\right\},
    \\\\ \noalign{\vspace{-10pt}}
    \qphi =\mathbb{E}\left\{\tanh^2\left[\dfrac{\beta}{1+\rho_3}\chi^1\alpha\biggl(\tilde C_2 \bar{n}^\sigma_{\xi^1}+ \tilde C_3\bar{n}^\tau_{\eta^1}\biggr)+ z\alpha\sqrt{\dfrac{\beta^2\rho_3}{(1+\rho_3)^2}\biggl(\tilde C_2 \bar{n}^\sigma_{\xi^1}+ \tilde C_3\bar{n}^\tau_{\eta^1}\biggr)^2 +\dfrac{\beta\gamma \bar{p}^{\w}}{2}}\right]\right\}.
\end{array}
\end{equation}
}

\noindent
where the $n^{\sigma}_{\xi^1}, n^{\tau}_{\eta^1}$, and $n^{\phi}_{\chi^1}$ variables are ruled by \eqref{eq:n_large_M} and $\bar{p}^z$, $\bar{p}^{s^\dag},$ and $\bar{p}^{w^\dag}$ obey \eqref{eq:p_equations}.
Furthermore, we set:
\begin{equation}
    \begin{array}{lllll}
         \tilde A_2=\dfrac{\a}{\theta} \sqrt{(1 + \rho_1) (1 + \rho_2)} ,&\tilde A_3=\dfrac{\b}{\alpha} \sqrt{(1 + \rho_1) (1 + \rho_3)} 
         \\\\ \noalign{\vspace{-10pt}}
         \tilde B_2=\a \sqrt{(1 + \rho_1) (1 + \rho_2)}  ,&\tilde B_3=\dfrac{ \c}{\alpha} \sqrt{(1 + \rho_2) (1 + \rho_3)}  
         \\\\ \noalign{\vspace{-10pt}}
         \tilde C_2=\b \sqrt{(1 + \rho_1) (1 + \rho_3)}  ,&\tilde C_3=\dfrac{\c}{\theta} \sqrt{(1 + \rho_2) (1 + \rho_3)}
    \end{array}
\end{equation}

\begin{algorithm}[tb]
\caption{Numerical solution of the self consistency equations. For the unsupervised case we fixed the threshold for $\bm M_{start}$ at 0.1.\label{alg:SC}}
\KwIn{Starting load of the network $\gamma_{max}$, activation values of the layers $\bm G=(g_{\sigma\tau}, g_{\sigma\phi}, g_{\tau\phi})$, size values of the layers $\bm \Gamma=(\alpha, \theta)$, vector of temperature values $T_{\text{vett}}$, datasets entropies $\bm \rho$}
\KwOut{Vector of maximum load values of the network for each temperature values}

\For{$t$ in $T_{vett}$}{
    $\Delta \gamma = \gamma_{max}/2$\;
    \While{$\Delta\gamma > 10^{-4}$}{
        $\bm M=\bm 1, \bm Q=\bm 1$\;
        \For{$iter$ in \( (1, \ldots, 1000) \)}{
        Compute the r.h.s. of the self consistency equations: 
        \\
        $\bm M_{new}= \mathbb{E}_x\tanh{f^{\bm G,\bm \Gamma, \bm \rho}_{t,\gamma}(x,\bm M,\bm Q)}$ \\ $\bm Q_{new}= \mathbb{E}_x\tanh^2{f^{\bm G,\bm \Gamma, \bm \rho}_{t,\gamma}(x, \bm M,\bm Q)}$\;
        Evaluate $\delta = \sqrt{|\bm{M}_{new}-\bm{M}|^2+|\bm{Q}_{new}-\bm{Q}|^2}$\;
            \If{$\delta< 10^{-10}$}{
        \textbf{break}}
        \Else{
        Compute the fixed point equations for the order parameters
        $\bm M = \frac{\bm M +\bm M_{new}}{2}$, $\bm Q = \frac{\bm Q +\bm Q_{new}}{2}$}
        }
        \If{$\bm M_{start}<10^{-5}$}{
        $\gamma_{max}=\gamma_{max}-\Delta\gamma$}
        \Else{$\gamma_{max}=\gamma_{max}+\Delta\gamma$}
        
        $\Delta\gamma = \Delta\gamma/2$\;
    }
}
\end{algorithm}

Equations \eqref{seldefinitive} and \eqref{seldefinitive1} have been numerically investigated, and the results are presented in Fig.~\ref{VariABC}. In this figure, we display the retrieval region in the $(\gamma, \beta^{-1})$ plane. Specifically, we illustrate the phase diagram corresponding to the first layer, allowing the parameters \( g_{\sigma\tau}, \ g_{\sigma\phi}, \ g_{\tau\phi} \) and \( \rho=\rho_1=\rho_2=\rho_3 \) to vary. In the particular case where \(\b=\c=0\) and \(\rho=0\), the resulting phase diagram reduces to that presented in \cite{centonze2024statistical} for the BAM model, showing only the boundary separating the retrieval region from the spin glass phase.

\vspace{3mm}

Once the retrieval region is identified, we define \(\gamma_{\textit{max}}\)\footnote{The following statements also apply in the unsupervised setting, with appropriate modifications to the value of \(\gamma_{\textit{max}}\).} as the maximum load that the network can sustain while remaining within this regime. In this context, we investigate how the total number of synapses in such a multilayer neural network varies with its degree of asymmetry, with the objective of determining an optimal architecture in terms of storage capacity.

\begin{figure}[t]
    \centering
    \begin{minipage}{0.31\textwidth}
        \centering
        \includegraphics[width=\textwidth]{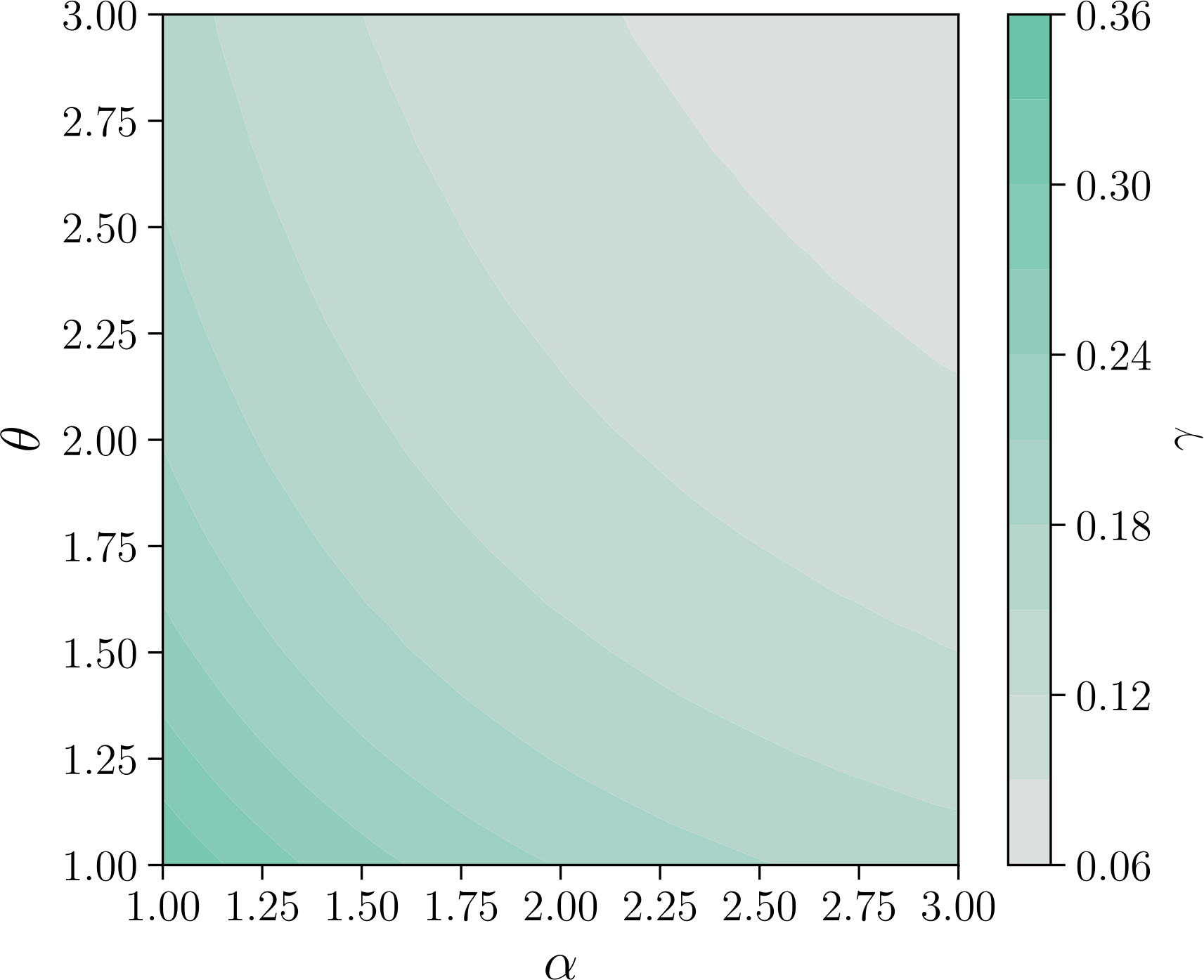}
    \end{minipage}
    \hfill
    \begin{minipage}{0.32\textwidth}
        \centering
        \includegraphics[width=\textwidth]{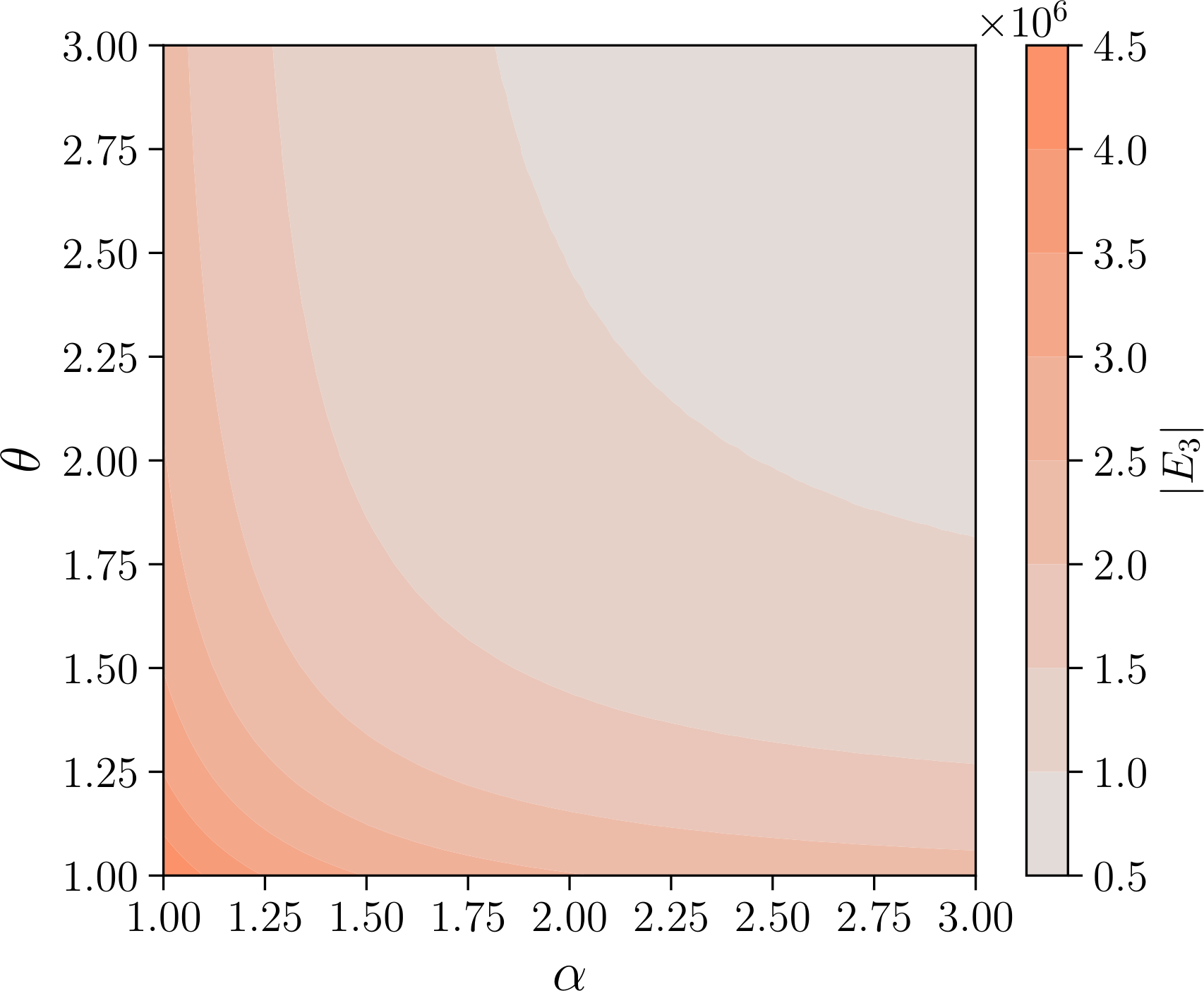}
    \end{minipage}
    \hfill
    \begin{minipage}{0.32\textwidth}
        \centering
        \includegraphics[width=\textwidth]{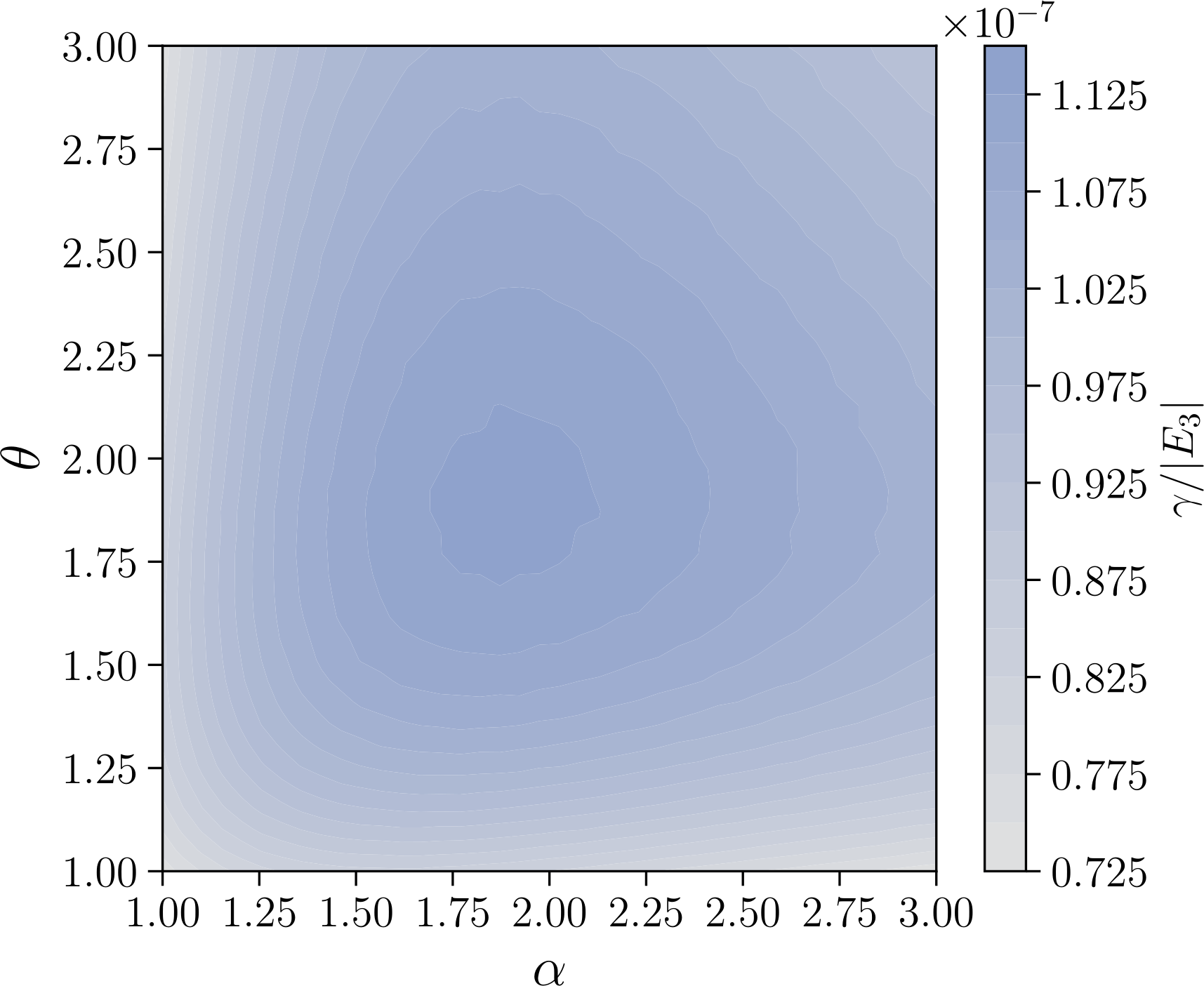}
    \end{minipage}
    \caption{Capacity, connectivity, and efficiency in TAM networks as functions of asymmetry.  
This figure explores how the asymmetry parameters \(\alpha\) and \(\theta\) shape the performance of TAM networks. The left panel shows that the maximum load \(\gamma_{max}\) is attained in the symmetric case \((\alpha, \theta) = (1,1)\), confirming that balanced architectures favor memory capacity. The central panel reports the number of synapses \(|E_3|\), which also peaks under symmetry due to maximal inter-layer connectivity. Interestingly, the right panel, showing the ratio \(\gamma_{max}/|E_3|\) as a measure of synaptic efficiency, reveals that the most efficient configuration occurs not in the fully symmetric regime, but at a slightly asymmetric point around \(\alpha = \theta \approx 1.92\). This indicates that a mild asymmetry can provide the best trade-off between capacity and synaptic cost.}
    \label{fig:TAM_Gamma_Synapses}
\end{figure}

For a \(k\)-partite network with partitions of cardinalities \(N_1, \dots, N_k\), the total number of synapses is given by
\begin{equation}
   |E_{k}| \;=\; \binom{\sum_{i=1}^k N_i}{2} \;-\; \sum_{i=1}^k \binom{N_i}{2}. 
\end{equation}
In our setting, we focus on the case \(k = 3\). The configuration illustrated in Fig.~\ref{fig:TAM_Gamma_Synapses} (left) demonstrates the computation of \(\gamma_{\textit{max}}\) for various choices of the asymmetry parameters \(\alpha\) and \(\theta\). Notably, the highest values of \(\gamma_{\textit{max}}\) occur when the network adopts a symmetric architecture—specifically, when \(\alpha = \theta = 1\) (i.e., \(N_1 = N_2 = N_3\)). Under these conditions, the number of synapses is maximized, thereby enhancing the network's load capacity.

Furthermore, the plot in Fig.~\ref{fig:TAM_Gamma_Synapses} (center), derived from the aforementioned expression, confirms that a more balanced network—one with partitions of comparable sizes—exhibits a larger number of synapses and, consequently, improved storage capabilities\footnote{Through Jensen's inequality, it is immediate to show that the number of synapses is maximized when the partitions of the network have the same cardinality.}.

\medskip
In addition to maximizing storage capacity, the efficiency of the TAM model is moreover highlighted by quantifying the synaptic savings relative to the classical Hopfield model. This is accomplished by evaluating the ratio ${|E_k|}/{\binom{\sum_{i=1}^k N_i}{2}},$
where the denominator represents the total number of synapses in a fully connected Hopfield network. Remarkably, in the symmetric regime, the TAM model utilizes only approximately \(55.6\%\) of the synapses required by the unipartite model, thereby demonstrating a significant saving in storing resources.

\subsection{Unsupervised setting}
For the unsupervised setting, using the same arguments to get \eqref{GaussApprox}, it is possible to prove the following Gaussian approximations:

    \begin{align}
        \dfrac{1}{M r_1r_2}\SOMMA{a=1}{M}\Xi_i^{\mu,a}\Theta_j^{\mu,a}\sim \sqrt{1+\rho_{12}}J_{ij}^\mu \;, \label{eq:approxGaussJ}\\ 
        \dfrac{1}{M r_1r_3}\SOMMA{a=1}{M}\Xi_i^{\mu,a}\Upsilon_k^{\mu,a}\sim \sqrt{1+\rho_{13}}Y_{ik}^\mu \;, \label{eq:approxGaussY}\\
        \dfrac{1}{M r_2r_3}\SOMMA{a=1}{M}\Theta_j^{\mu,a}\Upsilon_k^{\mu,a}\sim \sqrt{1+\rho_{23}}V_{jk}^\mu \;\label{eq:approxGaussV},
    \end{align}

\noindent
where $J_{ij}^\mu, Y_{ik}^\mu, V_{jk}^\mu \sim\mathcal{N}(0,1)$.

This allows us to rewrite the partition function as follows:

\begin{equation}\small
\label{eq:Z_supervised_aligned}
\begin{aligned}
\mathcal{Z}_{\bm N, M, \bm r}^{\bm g}(\beta, \bm J) = &
\sum_{\{\sigma\}, \{\tau\}, \{\phi\}} \exp\Bigg[\beta\Bigg(
    \a \sqrt{N_1 N_2 (1+\rho_1)(1+\rho_2)} \frac{1}{M} \sum_{a=1}^{M} n_{\Xi^{1,a}}^\sigma n_{\Theta^{1,a}}^\tau \\
& + \b \sqrt{N_1 N_3 (1+\rho_1)(1+\rho_3)} \frac{1}{M} \sum_{a=1}^{M} n_{\Xi^{1,a}}^\sigma n_{\Upsilon^{1,a}}^\phi \\
& + \c \sqrt{N_2 N_3 (1+\rho_2)(1+\rho_3)} \frac{1}{M} \sum_{a=1}^{M} n_{\Theta^{1,a}}^\tau n_{\Upsilon^{1,a}}^\phi
\Bigg)\Bigg] \times \\
& \times \exp\Bigg[
    \a\beta \sqrt{\frac{1+\rho_{12}}{N_1 N_2 (1+\rho_1)(1+\rho_2)}} \sum_{\mu>1} \sum_{i,j=1}^{N_1,N_2} J_{ij}^{\mu} \sigma_i \tau_j \\
& \;\qquad + \b\beta \sqrt{\frac{1+\rho_{13}}{N_1 N_3 (1+\rho_1)(1+\rho_3)}} \sum_{\mu>1} \sum_{i,j=1}^{N_1,N_3} Y_{ij}^{\mu} \sigma_i \phi_j \\
& \;\qquad + \c\beta \sqrt{\frac{1+\rho_{23}}{N_2 N_3 (1+\rho_2)(1+\rho_3)}} \sum_{\mu>1} \sum_{i,j=1}^{N_2,N_3} V_{ij}^{\mu} \tau_i \phi_j \\
& \;\qquad + J_1 \sum_{i=1}^{N_1} \xi_i^1 \sigma_i
      + J_2 \sum_{i=1}^{N_2} \eta_i^1 \tau_i
      + J_3 \sum_{i=1}^{N_3} \chi_i^1 \phi_i
\Bigg]
\end{aligned}
\end{equation}

\begin{remark}
In the unsupervised case, in the thermodynamic limit we can use the assumption, valid for $M$ sufficiently large

\begin{equation}\label{eq:UnsupStrongApprox}
\begin{split}
	\mathbb{P}( n_{\Xi^{1,a}}^\sigma) \xrightarrow[]{\bm N\to\infty} \delta ( n_{\Xi^{1,a}}^\sigma - \bar n_{\xi^{1}}^\sigma)\,, \\\\ \noalign{\vspace{-10pt}}
\mathbb{P}( n_{\Theta^{1,a}}^\tau) \xrightarrow[]{\bm N\to\infty} \delta ( n_{\Theta^{1,a}}^\tau - \bar n_{\eta^{1}}^\tau)\,, \\\\ \noalign{\vspace{-10pt}}
\mathbb{P}( n_{\Upsilon^{1,a}}^\phi) \xrightarrow[]{\bm N\to\infty} \delta ( n_{\Upsilon^{1,a}}^\phi - \bar n_{\chi^{1}}^\phi)\,.	
\end{split}
\end{equation}

Despite being rather strong, this assumption is useful to make the numerical simulation of self-consistency equations tractable. Otherwise, it would be necessary to simulate $3\cdot M + 6$ coupled self-consistency equations.
\end{remark}

Similarly to the previous Section, in the large $M$ limit, the following self-consistency equations hold:

\begin{multline}\small\label{eq:selfMagnSigma}
		\bar{m}^\sigma_{\xi^1} = \mathbb{E}\Bigg\{\xi^1\tanh\Bigg[\dfrac{\beta}{1+\rho_1}\xi^1\biggl(\tilde A_2  \bar{n}^\tau_{\eta^1}+ \tilde A_3 \bar{n}^\phi_{\chi^1}\biggr)+
		\\
		\qquad\qquad+z\beta\sqrt{\dfrac{\rho_1}{(1+\rho_1)^2}\biggl(\tilde A_2  \bar{n}^\tau_{\eta^1}+ \tilde A_3 \bar{n}^\phi_{\chi^1}\biggr)^2 +\a^2 D_1^2\qtau + \b^2 D_2^2\qphi}\Bigg]\Bigg\},  
\end{multline}
\begin{multline}\small\label{eq:selfMagnTau}
	\bar{m}^\tau_{\eta^1} = \mathbb{E}\Bigg\{\eta^1\tanh\Bigg[\dfrac{\beta}{1+\rho_2}\eta^1\theta\biggl(\tilde B_2 \bar n_{\xi^{1}}^\sigma+ \tilde B_3 \bar{n}^\phi_{\chi^1}\biggr) +
		\\ 
	\qquad\qquad+ z\beta\theta\sqrt{\dfrac{\rho_2}{(1+\rho_2)^2}\biggl(\tilde B_2 \bar n_{\xi^{1}}^\sigma+ \tilde B_3 \bar{n}^\phi_{\chi^1}\biggr)^2 + \a^2 D_1^2\qsigma + \c^2 D_3^2\qphi}\Bigg]\Bigg\},
\end{multline}
\begin{multline}\small\label{eq:selfMagnPhi}
	\bar{m}^\phi_{\chi^1} = \mathbb{E}\Bigg\{\chi^1\tanh\Bigg[\dfrac{\beta}{1+\rho_3}\chi^1\alpha\biggl(\tilde C_2 \bar n_{\xi^{1}}^\sigma+ \tilde C_3\bar{n}^\tau_{\eta^1}\biggr)+
	\\ 
	\qquad\qquad+ z\beta\alpha\sqrt{\dfrac{\rho_3}{(1+\rho_3)^2}\biggl(\tilde C_2 \bar n_{\xi^{1}}^\sigma+ \tilde C_3\bar{n}^\tau_{\eta^1}\biggr)^2 +\b^2 D_2^2\qsigma+\c^2 D_3^2\qtau}\Bigg]\Bigg\},
\end{multline}
\begin{multline}\small\label{eq:selfQSigma}
	\qsigma = \mathbb{E}\Bigg\{\tanh^2\Bigg[\dfrac{\beta}{1+\rho_1}\xi^1\biggl(\tilde A_2  \bar{n}^\tau_{\eta^1}+ \tilde A_3 \bar{n}^\phi_{\chi^1}\biggr) +
	\\
	\quad+	 z\beta\sqrt{\dfrac{\rho_1}{(1+\rho_1)^2}\biggl(\tilde A_2  \bar{n}^\tau_{\eta^1}+ \tilde A_3 \bar{n}^\phi_{\chi^1}\biggr)^2 +\a^2 D_2^2 + \b^2 D_4^2}\Bigg]\Bigg\},
\end{multline}
\begin{multline}\small\label{eq:selfQTau}
	\qtau = \mathbb{E}\Bigg\{\tanh^2\Bigg[\dfrac{\beta}{1+\rho_2}\eta^1\theta\biggl(\tilde B_2 \bar n_{\xi^{1}}^\sigma+ \tilde B_3 \bar{n}^\phi_{\chi^1}\biggr) +
	\\
	\quad+	 z\beta\theta\sqrt{\dfrac{\rho_2}{(1+\rho_2)^2}\biggl(\tilde B_2 \bar n_{\xi^{1}}^\sigma+ \tilde B_3 \bar{n}^\phi_{\chi^1}\biggr)^2 +\a^2 D_1^2 + \c^2 D_6^2}\Bigg]\Bigg\},
\end{multline}
\begin{multline}\small\label{eq:selfQPhi}
	\qphi =\mathbb{E}\Bigg\{\tanh^2\Bigg[\dfrac{\beta}{1+\rho_3}\chi^1\alpha\biggl(\tilde C_2 \bar n_{\xi^{1}}^\sigma+ \tilde C_3\bar{n}^\tau_{\eta^1}\biggr)+
	\\
	\quad+	 z\beta\alpha\sqrt{\dfrac{\rho_3}{(1+\rho_3)^2}\biggl(\tilde C_2 \bar n_{\xi^{1}}^\sigma+ \tilde C_3\bar{n}^\tau_{\eta^1}\biggr)^2 +\b^2 D_3^2+\c^2 D_5^2}\Bigg]\Bigg\},
\end{multline}
\noindent
where we define
\begin{equation*}
	\begin{array}{lllll}
		\tilde A_2=\dfrac{\a}{\theta} \sqrt{(1 + \rho_1) (1 + \rho_2)} ,&\tilde A_3=\dfrac{\b}{\alpha} \sqrt{(1 + \rho_1) (1 + \rho_3)} ,
		\\\\ \noalign{\vspace{-10pt}}
		\tilde B_2=\a \sqrt{(1 + \rho_1) (1 + \rho_2)}  ,&\tilde B_3=\dfrac{ \c}{\alpha} \sqrt{(1 + \rho_2) (1 + \rho_3)}  ,
		\\\\ \noalign{\vspace{-10pt}}
		\tilde C_2= \b \sqrt{(1 + \rho_1) (1 + \rho_3)}  ,&\tilde C_3=\dfrac{\c}{\theta} \sqrt{(1 + \rho_2) (1 + \rho_3)} ,
	\end{array}
\end{equation*}
\begin{equation*}
	\begin{array}{lllll}
		D_1=\sqrt{\gamma\dfrac{1+\rho_{12}}{(1+\rho_1)(1+\rho_2)}},\\\\ \noalign{\vspace{-10pt}}
		D_2=\sqrt{\gamma\dfrac{1+\rho_{13}}{(1+\rho_1)(1+\rho_3)}},\\\\ \noalign{\vspace{-10pt}}
		D_3=\sqrt{\gamma\dfrac{1+\rho_{23}}{(1+\rho_2)  (1+\rho_3)}}.
		
	\end{array}
\end{equation*}

\noindent
and the $\bar{n}$ variables are ruled by \eqref{eq:n_large_M}. These results are summarized in the phase diagrams shown  in Fig.~\ref{fig:TamUnsupSelfConPlanes}.

\begin{figure}[t]
    \centering
    \begin{minipage}{0.49\textwidth}
        \centering
        \includegraphics[width=\textwidth]{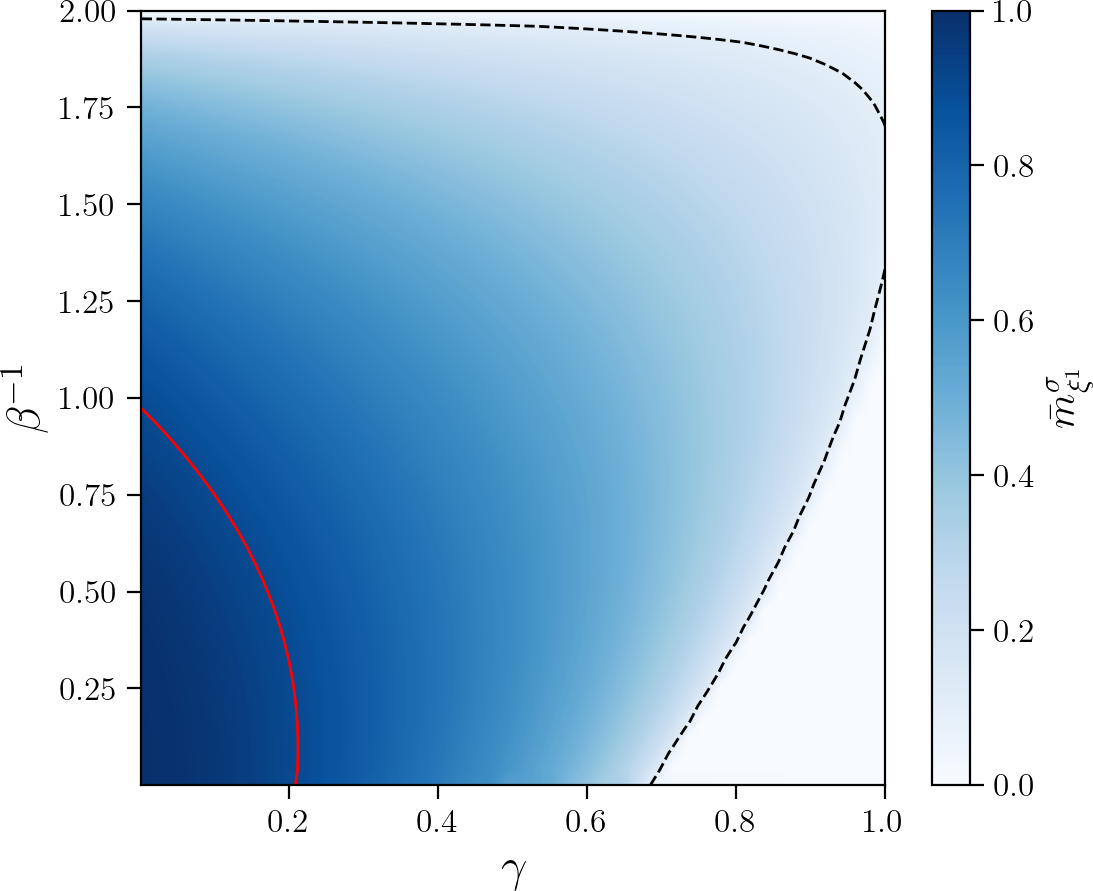}
    \end{minipage}
    \hfill
    \begin{minipage}{0.49\textwidth}
        \centering
        \includegraphics[width=\textwidth]{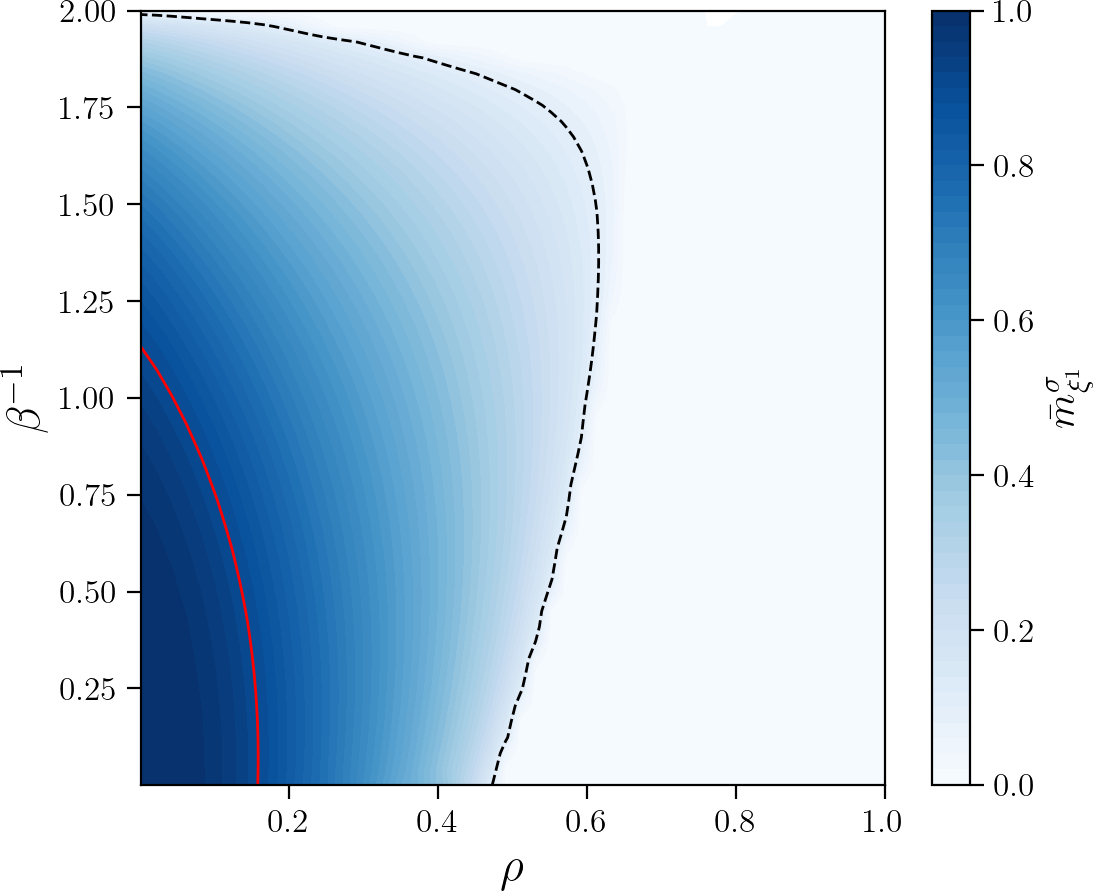}
    \end{minipage}
    \caption{Numerical solution of the self-consistency Eqs.~\eqref{eq:selfMagnSigma}-\eqref{eq:selfQPhi} for $\left(\alpha,\theta\right)=(1,1)$ in different control parameter planes. (left) The solution is represented in the $(\gamma,\beta^{-1})$ plane with datasets entropies $r_1=r_2=r_3=0.3$ and $M=100$, and $(\a,\b,\c)=(1,1,1)$. (right) The solution is shown in the $(\rho,\beta^{-1})$ plane with $\gamma=0.1$. In both cases, we depict only the $\bm\sigma$ layer due to network symmetry reasons. The black dashed line indicates the boundary between two regimes: one where the sigma layer magnetization is greater than $0.1$ and the other where it is less than $0.1$. The red line marks where $\bm\sigma$-layer magnetization exceeds $0.9$ and falls below $0.9$. Unlike the supervised case, where a clear phase transition can be observed, in the unsupervised setting — due to the approximations involved in solving the model (see App. \ref{appsec:proofGuerra}) — such interpretability is not preserved.}
    \label{fig:TamUnsupSelfConPlanes}
\end{figure}

\noindent
In both Fig.~\ref{fig:TamUnsupSelfConPlanes} (left) and Fig.~\ref{fig:TamUnsupSelfConPlanes} (right), two distinct retrieval regimes can be identified: a stable one, present at low temperatures and low values of \(\gamma\) (or, correspondingly, \(\rho\)), where the magnetization is high. In the diagram, this stable region is delineated as the area in which the magnetization exceeds a value of 0.9 by the red line. Conversely, there exists a metastable regime in which the magnetization is non-zero yet remains relatively small; this is represented on the diagram by the dashed black line. These diagrams differ from the phase diagrams typically obtained in storing processes (or those obtainable in supervised learning) due to the substantial approximations applied in Eqs.~\eqref{eq:approxGaussJ}-\eqref{eq:approxGaussV} and in Eq.~\eqref{eq:UnsupStrongApprox}. 

\subsection{Signal-to-noise} \label{subsec:montecarlo}
The fundamental concept behind this technique is quite straightforward \cite{Amit}: assuming that $(\bm\sigma, \bm\tau, \bm\phi)$ represents a stable configuration corresponding to a stored triplet of patterns, the stability requirement for this configuration is simply expressed as $h_i^\sigma\sigma_i \geq 0$ for all $i \in (1, \ldots, N_1)$ for one layer, $h_j^\tau\tau_j \geq 0$ for all $j \in (1, \ldots, N_2)$ for the other layer, and $h_k^\phi\phi_k \geq 0$ for all $k \in (1, \ldots, N_3)$ for the last one. Here, $h$ represents the post-synaptic fields acting on these neurons. Note that this approach does not provide information on the depth of the minima (nor those related to the archetypes nor the spurious ones), hence it provides a  first —zero fast noise— picture of the configuration landscape, later on to be deepened by Monte Carlo runs exploring the effect of the thermal noise.

First, we derive the conditions under which a given neuron $\sigma_i$ (and similarly for $\tau_j$ and $\phi_k$), influenced by the internal field $h_i^\sigma$, aligns with the stored pattern $\xi_i^1$ (and similarly for $h_j^\tau, h_k^\phi, \eta_j^1, \chi_k^1$) that we intend the network to retrieve\footnote{In the following, with no loss of generality, we focus on the choice $N=N_1=N_2=N_3$ for the sake of simplicity.}. We define the model's dynamics in the $\beta \rightarrow \infty$ limit as described by the following update equations: the objective is to assess the state of neuron $\bm\sigma$ at time $t = 1$, starting from the initial configuration $\bm\sigma^{(t=0)} = \bm\xi^1$, $\bm\tau^{(t=0)} = \bm\eta^1$, and $\bm\phi^{(t=0)} = \bm\chi^1$.

\begin{align}\label{eq:updateRule}
	\sigma_i^{(t+1)} &=\mathrm{sign}(h_{i}^{\sigma,(t)}\sigma_{i}^{(t)}), \nonumber \\
	\tau_j^{(t+1)}&=\mathrm{sign}({{{h}}}_{j}^{\tau,(t)}\tau_{j}^{(t)}), \\
	\phi_k^{(t+1)} &=\mathrm{sign}({h}_{k}^{\phi,(t)}\phi_{k}^{(t)}), \nonumber
\end{align}
where, for the supervised case
{\small
\begin{align}
h^{\sigma}_{i}(\bm\eta,\bm\chi) &= \SOMMA{\mu=1}{K}\left(\dfrac{\a\sqrt{(1+\rho_1)(1+\rho_2)}}{\sqrt{N_1 N_2}}\SOMMA{j=1}{N_2}\hat\eta_j^\mu\tau_j +\dfrac{\b\sqrt{(1+\rho_1)(1+\rho_3)}}{\sqrt{N_1 N_3}}\SOMMA{k=1}{N_3}\hat\chi_k^\mu\phi_k\right)\hat\xi_i^\mu  \,,
\end{align}
\begin{align}
h^{\tau}_{j}(\bm\xi,\bm\chi) &= \SOMMA{\mu=1}{K}\left(\dfrac{\a\sqrt{(1+\rho_1)(1+\rho_2)}}{\sqrt{N_1 N_2}}\SOMMA{i=1}{N_1}\hat\xi_i^\mu\sigma_i + \dfrac{\c\sqrt{(1+\rho_2)(1+\rho_3)}}{\sqrt{N_2 N_3}}\SOMMA{k=1}{N_3}\hat\chi_k^\mu\phi_k\right)\hat\eta_j^\mu \,,
\end{align}
\begin{align}
h^{\phi}_k(\bm\xi,\bm\eta) &= \SOMMA{\mu=1}{K}\left(\dfrac{\b\sqrt{(1+\rho_1)(1+\rho_3)}}{\sqrt{N_1 N_3}}\SOMMA{i=1}{N_1}\hat\xi_i^\mu\sigma_i+\dfrac{\c\sqrt{(1+\rho_2)(1+\rho_3)}}{\sqrt{N_2 N_3}}\SOMMA{j=1}{N_2}\hat\eta_j^\mu\tau_j\right)\hat\chi_k^\mu \,,
\end{align} 
}
whereas, in the unsupervised setting, we have the following post-synaptic fields

	\begin{align}
		h^{\sigma}_{i}(\bm \eta, \bm \chi) &= \SOMMA{\mu=1}{K}\Bigg(\dfrac{\a}{r_1r_2\sqrt{N_1 N_2(1+\rho_1)(1+\rho_2)}\,M}\SOMMA{j, a=1}{N_2,M}\Xi_i^{\mu,a}\Theta_j^{\mu,a}\tau_j+ \nonumber
		\\
		& \qquad\,\,\,+\dfrac{\b}{r_1r_3\sqrt{N_1 N_3(1+\rho_1)(1+\rho_3)}\,M}\SOMMA{k,a=1}{N_3,M}\Xi_i^{\mu,a}\Upsilon_k^{\mu,a}\phi_k\Bigg),
      \end{align}
      \begin{align}
    h^{\tau}_{j}(\bm \xi, \bm \chi) &= \SOMMA{\mu=1}{K}\Bigg(\dfrac{\a}{r_1r_2\sqrt{N_1 N_2(1+\rho_1)(1+\rho_2)}\,M}\SOMMA{i, a=1}{N_1,M}\Xi_i^{\mu,a}\Theta_j^{\mu,a}\sigma_i +\nonumber
		\\
		& \qquad\,\,\,+ \dfrac{\c}{r_2r_3\sqrt{N_2 N_3(1+\rho_2)(1+\rho_3)}\,M}\SOMMA{k,a=1}{N_3,M}\Theta_j^{\mu,a}\Upsilon_k^{\mu,a}\phi_k\Bigg),
	\end{align}
        \begin{align}
		h^{\phi}_k (\bm \xi, \bm \eta) &= \SOMMA{\mu=1}{K}\Bigg(\dfrac{\b}{r_1r_3\sqrt{N_1 N_3(1+\rho_1)(1+\rho_3)}\,M}\SOMMA{i,a=1}{N_1,M}\Xi_i^{\mu,a}\Upsilon_k^{\mu,a}\sigma_i+\nonumber
		\\
		& \qquad\,\,\,+\dfrac{\c}{r_2r_3\sqrt{N_2 N_3(1+\rho_2)(1+\rho_3)}\,M}\SOMMA{j,a=1}{N_2,M}\Theta_j^{\mu,a}\Upsilon_k^{\mu,a}\tau_j\Bigg).
     \end{align}

\begin{algorithm}\label{MCMC}
\caption{MCMC Sequential Glauber dynamics}
\KwIn{Parameters $\bm G = (\a,\b,\c)$; layer dimensions $\bm N=(N_1,N_2,N_3)$; number of patterns $K$; number of examples $\bm M = (M_1,M_2,M_3)$; dataset qualities $\bm r=(r_1,r_2,r_3)$; number of steps $N_s$; inverse temperature $\beta$; input state $\bm S^{(0)}=(\sigma^{(0)},\tau^{(0)},\phi^{(0)})$. }
\KwOut{Time series of magnetizations 
\[
\mathcal{M} = \Big\{ \big( m_{\xi^1}^{\sigma,\, (t=j)},m_{\eta^1}^{\tau,\, (t=j)},m_{\chi^1}^{\phi,\, (t=j)} \big) \Big\}_{j=0}^{N_s}\,.
\]}

\BlankLine
\For{$i \gets 1$ \KwTo $3$}{
Generate base patterns, examples, and test according to

  \(P_i \gets\) random binary matrix in \(\{-1,+1\}^{K \times N_i}\)\;
  \(E_i \gets \mathtt{generateExamples}(P_i, M_i, r_i)\)\;
}
\BlankLine

Record initial magnetizations:
\[
m_{\xi^1}^{\sigma,\, (t=0)} \gets \mathtt{magn}(\sigma^{(0)},P_1),\quad m_{\eta^1}^{\tau,\, (t=0)} \gets \mathtt{magn}(\tau^{(0)},P_2),\quad m_{\chi^1}^{\phi,\, (t=0)} \gets \mathtt{magn}(\phi^{(0)},P_3),
\]
and set \(\mathcal{M}[0] \gets \big( m_{\xi^1}^{\sigma,\, (t=0)},m_{\eta^1}^{\tau,\, (t=0)},m_{\chi^1}^{\phi,\, (t=0)} \big)\).

\BlankLine
\For{\(j \gets 1\) \KwTo \(N_s\)}{
    \For{\(i \gets 1\) \KwTo $3$}{
        Generate noise vector \(u_i \sim \mathcal{U}(-1,1)^{N_i}\)\;
    }
    Update states simultaneously:
    \begin{align*}
        \sigma^{(j)} &\gets \operatorname{sign}\Big(\tanh\big(h^\sigma(\bm S^{(j-1)},\bm E,\bm G, \bm N, \bm M, \bm \rho)\cdot\beta\big)+u_1\Big) \\
        \tau^{(j)} &\gets \operatorname{sign}\Big(\tanh\big(h^\tau(\bm S^{(j-1)},\bm E,\bm G, \bm N, \bm M, \bm \rho)\cdot\beta\big)+u_2\Big) \\
        \phi^{(j)} &\gets \operatorname{sign}\Big(\tanh\big(h^\phi(\bm S^{(j-1)},\bm E,\bm G, \bm N, \bm M, \bm \rho)\cdot\beta\big)+u_3\Big).
    \end{align*}
    Compute magnetizations: \\
    $\qquad m_{\xi^1}^{\sigma,\, (t=j)} \gets \mathtt{magn}(\sigma^{(j)},P_1),$ \\
    $\qquad m_{\eta^1}^{\tau,\, (t=j)}
     \gets \mathtt{magn}(\tau^{(j)},P_2),$ \\
    $\qquad m_{\chi^1}^{\phi,\, (t=j)}
     \gets \mathtt{magn}(\phi^{(j)},P_3); $
    
    Update \(\mathcal{M}[j] \gets \big( m_{\xi^1}^{\sigma,\, (t=j)},m_{\eta^1}^{\tau,\, (t=j)},m_{\chi^1}^{\phi,\, (t=j)} \big)\).
}
\BlankLine
\Return \(\mathcal{M}\).
\end{algorithm}

\vspace{3mm}
\noindent

\vspace{3mm}
\noindent
After deriving the state of the network at time $t=1$, $\bm \sigma^{(t=1)}$,$\bm\tau^{(t=1)}$, $\bm\phi^{(t=1)}$ the overlap between this state and the patterns ${\bm\xi}^1,{\bm\eta}^1,{\bm\chi}^1$ is given by:
\begin{equation}
\label{setofmagn}
    m_{x^1}^y = \dfrac{1}{N}\SOMMA{i=1}{N}x_i^1 y_i^{(t=1)},
\end{equation}
where $\bm x^1 \in \{\bm\xi^1,\bm\eta^1, \bm\chi^1\}$ and $\bm y \in \{\bm\sigma, \bm\tau, \bm\phi\}$.
\\
\noindent
Now, we explicitly calculate the matrix $3\times 3$ whose entries are
\begin{equation}\label{eq:TAMUnsupMagnMatrix}
	 \begin{bmatrix}
		m_{\xi^1}^\sigma  & m_{\eta^1}^\sigma  & m_{\chi^1}^\sigma \\
		m_{\xi^1}^\tau & m_{\eta^1}^\tau & m_{\chi^1}^\tau \\
		m_{\xi^1}^\phi & m_{\eta^1}^\phi & m_{\chi^1}^\phi \\
	\end{bmatrix}
\end{equation}
and determine whether it becomes diagonal, i.e., if each layer calls the corresponding pattern, with the network engaged in the task of \textit{generalized pattern recognition}, i.e., when the Cauchy datum is equal to $\bm\sigma^{(t=0)} = \bm\xi^1$, $\bm\tau^{(t=0)} = \bm\eta^1$, and $\bm\phi^{(t=0)} = \bm\chi^1$.

\vspace{3mm}

\noindent
We can rewrite the update rule for a given configuration (see Eqs.~\ref{eq:updateRule}) as follows
\begin{equation*}
	\sigma^{(t+1)}_{i} = \sign\left[h^{\sigma,(t)}_i(\sigma^{(t)}_i)\sigma_i^{(t)}\right]\sigma^{(t)}_i \,.
\end{equation*}
Now:
\begin{equation*}
	m_{\xi^1}^{\sigma}(t+1)=\dfrac{1}{N}\sum_{i=1}^{N}\xi^1_i\sigma^{(t+1)}_i = \dfrac{1}{N}\sum_{i=1}^{N}\xi^1_i\sigma^{(t)}_i\sign\left[h^{\sigma,(t)}_i(\sigma^{(t)}_i)\sigma_i^{(t)}\right] .
\end{equation*}
Assuming that at time $t$, the system is in the configuration
\begin{equation*}
	{\bm\sigma}^{(t)} = \bm\xi^1 
\end{equation*}
\begin{figure}[t]
    \centering
    \includegraphics[height=5.0cm]{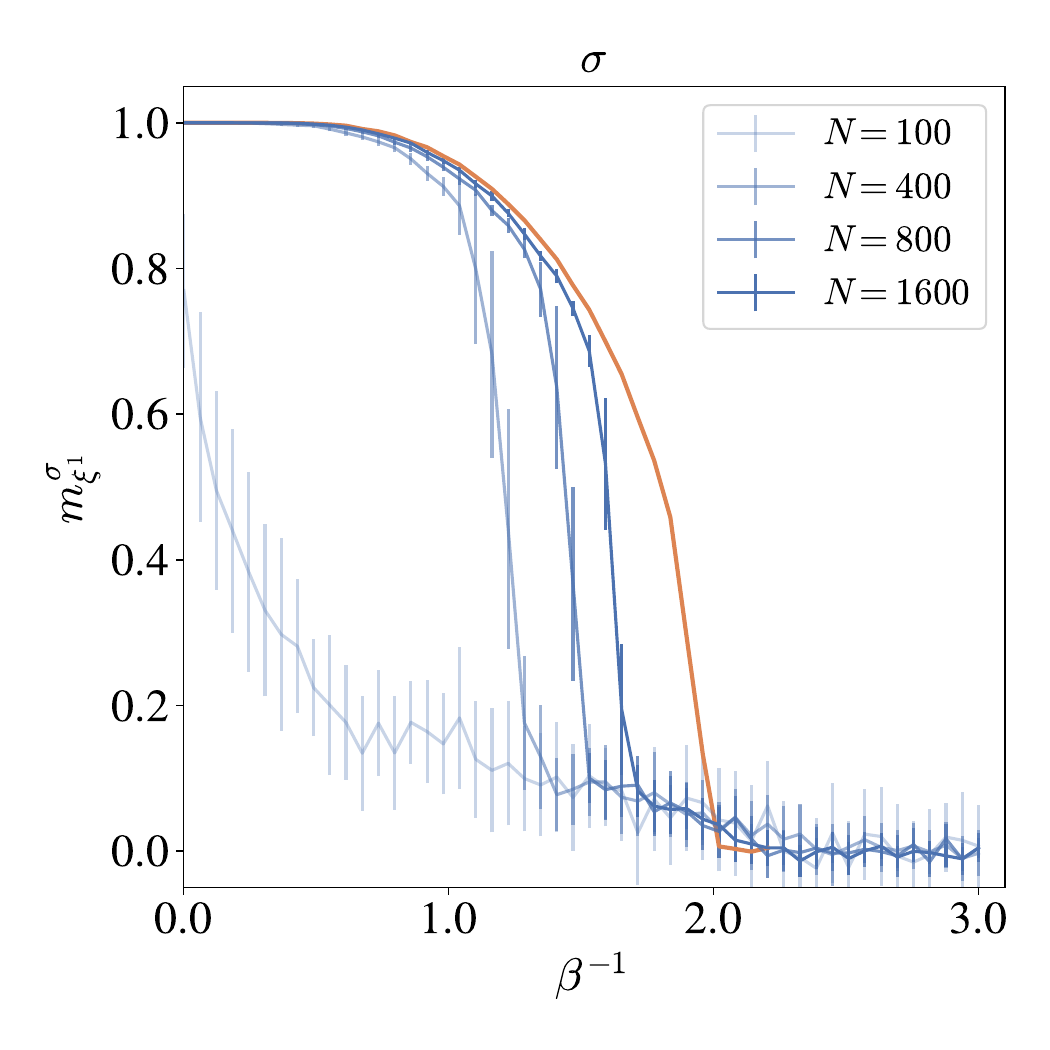}
    \label{0.9}
    \includegraphics[height=5.1cm]{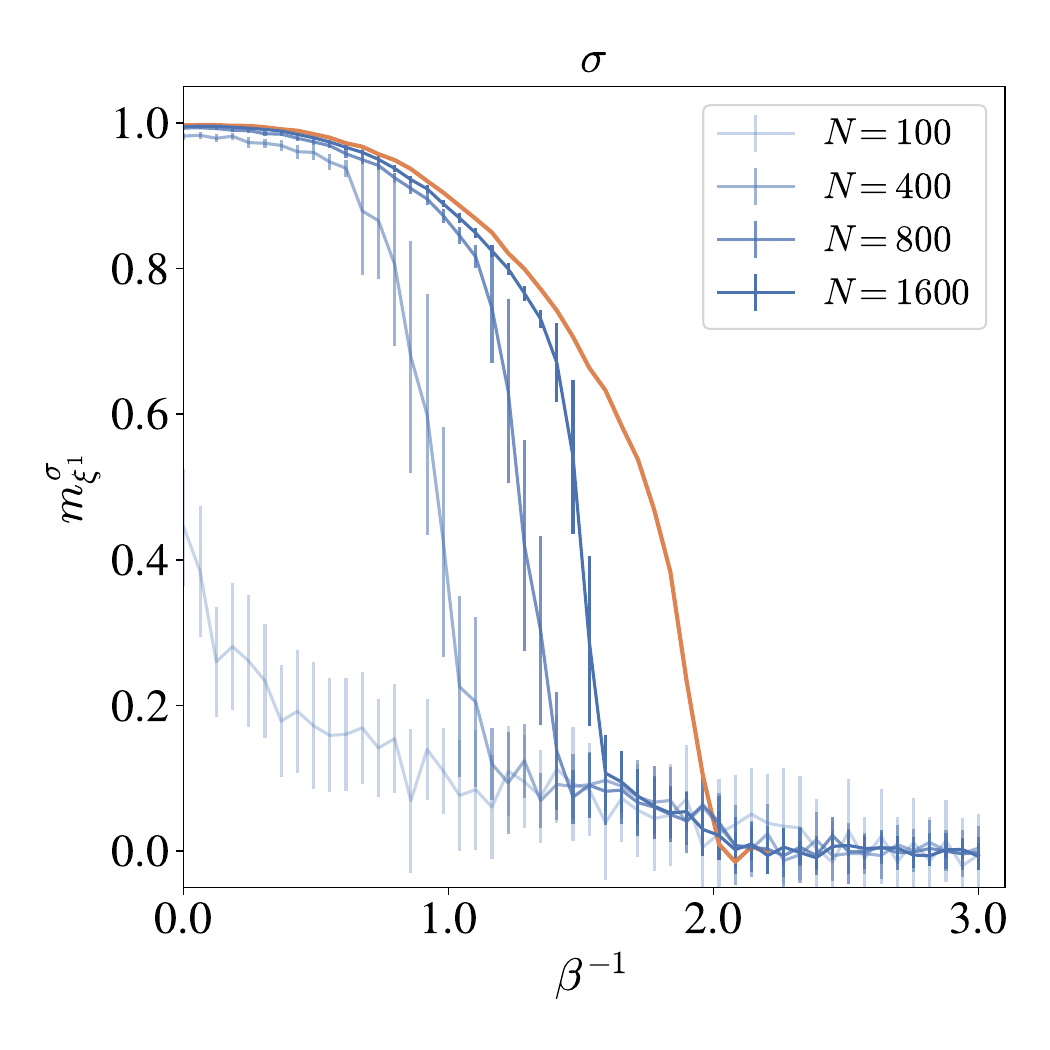}
    \label{0.6}
    \includegraphics[height=5.1cm]{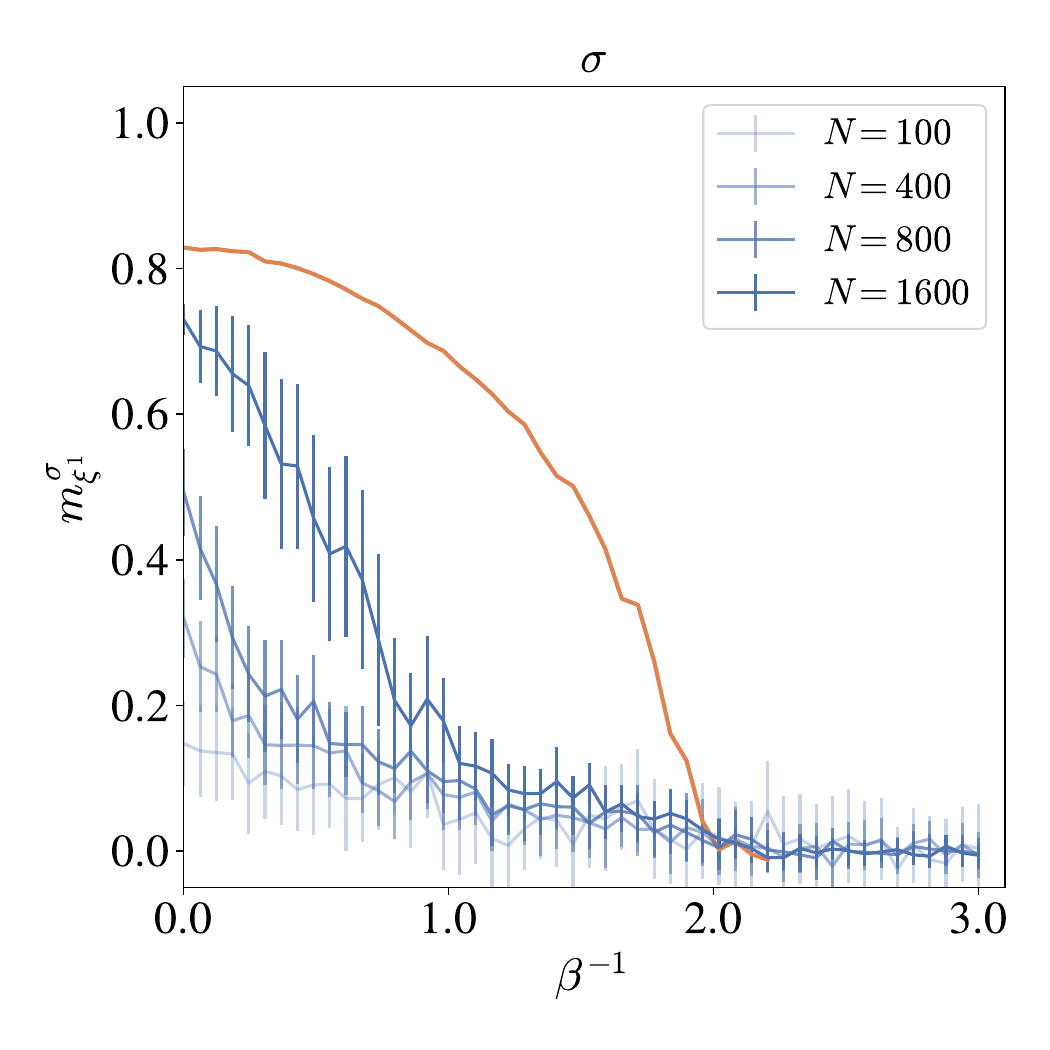}
    \label{0.3}
    \caption{
    Monte Carlo simulation of a symmetric $(\alpha =1, \ \theta=1)$ TAM network. The number of archetypes is fixed at $K= 40$ for each layer, with $M_1=M_2=M_3 =20$. Each curve in each plot is displayed in different shades of blue, corresponding to varying numbers of neurons (i.e., $N=N_1=N_2=N_3 \in \{100,400,800,1600\}$), while in orange, the same curves are shown in the limit $\gamma\to 0$. Each plot represents a distinct set of MC simulations that differ only in the quality of the dataset, denoted by $r$: specifically, from left to right, we set $\bm r=(0.9, 0.6, 0.3)$. The threshold value (critical temperature), beyond which magnetization drops sharply as the temperature increases, can be determined from the phase diagram in the thermodynamic limit $N_1 \to \infty$. Due to the symmetry of the network, we present the results for only one layer, $\bm\sigma$. The larger the bars, the greater the errors in the results, evaluated as the standard deviation over 100 independent runs.}
    \label{MagnVsTemp}
\end{figure}
then we rewrite the previous equation as:
\begin{equation*}
	m_{\xi^1}^{\sigma}(t+1) = \dfrac{1}{N}\sum_{i=1}^{N}\sign\left[h^{\sigma,(t)}_i(\xi^1)\xi^1\right].
\end{equation*}
We computed first and second moments of the field (see App.~\ref{appsec:evaluation} for the explicit calculations), so that we can use the Gaussian approximation:
\begin{equation*}
	m_{\xi^1}^{\sigma}(t+1) = \dfrac{1}{N}\sum_{i=1}^{N}\sign\left[\mu_1 + z_i\sqrt{\mu_2- \mu^2_1}\right]
\end{equation*}
where $z_i \sim\mathcal{N}(0,1) $.
For $N \gg 1$, by the law of large numbers, the empirical mean tends to the theoretical mean (i.e. $ N^{-1}\sum_{i=1}^{N}\to \mathbb{E}_z$), thus we obtain:
\begin{equation}
	\mathbb{E}_z\biggl(\sign\left[ \mu_1 + z\sqrt{\mu_2- \mu^2_1}\right]\biggr) .
\end{equation}
Now, noting that
\begin{equation}
	\dfrac{1}{\sqrt{2\pi}} \int e^{- \dfrac{z^2}{2}} \sign\left[\mu_1 + z\sqrt{\mu_2- \mu^2_1}\right] \, dz = \erf\left[\dfrac{1}{\sqrt{2}}\dfrac{\mu_1}{\sqrt{\mu_2-\mu_1^2}}\right]
\end{equation}
we obtain the thresholds in the supervised setting by substituting the expressions for \(\mu_1\) and \(\mu_2\) (\eqref{eq:m1sup} and \eqref{eq:m2sup}, respectively):
\begin{figure}[t]
    \centering
    \begin{subfigure}{0.49\textwidth}
        \centering
        \includegraphics[width=1\linewidth]{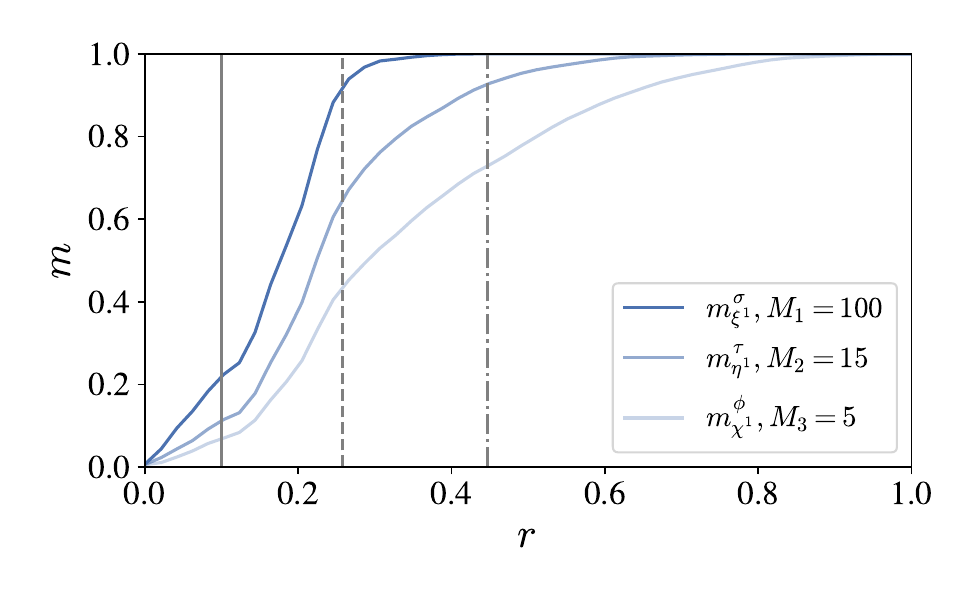}
    \end{subfigure}
    \hfill
    \begin{subfigure}{0.48\textwidth}
        \centering
        \includegraphics[width=1\linewidth]{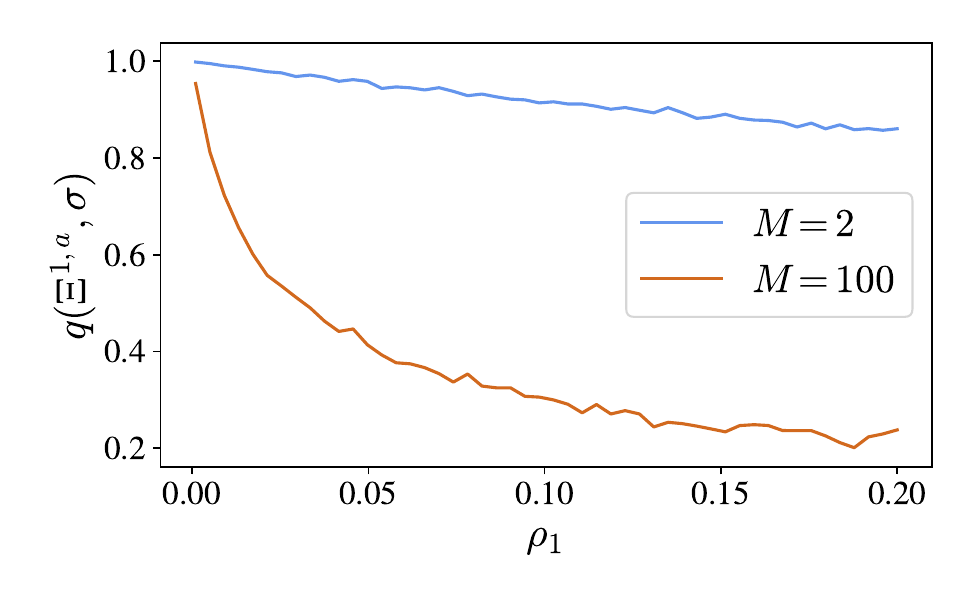}
    \end{subfigure}
    \caption{
    (Left) Monte Carlo simulations were conducted on a network comprising \(N_1=N_2=N_3=400\) neurons, employing \(M_1=100\), \(M_2=15\), and \(M_3=5\) examples per layer, under a storage load of \(\gamma=0.025\). Pattern retrieval is successfully achieved even with low-quality datasets as the number of examples increases, with the parameters set to \(\beta=10\) and the symmetry values \(\alpha=\theta=1\). The vertical lines indicate the theoretical predictions derived from statistical mechanics; specifically, from left to right, the first corresponds to \(M=100\), the second to \(M=15\), and the third to \(M=5\). 
    (Right) The performance in a generalized pattern reconstruction setting was evaluated using a TAM network with \(N=300\) neurons per layer.  
    Two datasets were compared: one with a small dataset of \(M=2\) examples (blue line) and one with a larger dataset of \(M=100\) examples (orange line).  
    Only the \(\bm{\sigma}\) layer is reported, since the \(\bm{\tau}\) and \(\bm{\phi}\) layers, initialized with a Rademacher distribution, exhibit zero overlap.  
The evaluation focused on the overlap between the network's final state and the examples after Monte Carlo simulations stabilized.  
As dataset entropy increases, the network with \(M=2\) retrieves the individual examples rather than the archetype.  
Conversely, the network with \(M=100\) is capable of retrieving the archetype; in other words, the network successfully generalizes the examples, with the archetypes corresponding to local minima \cite{agliariEmergence}.
}
    \label{fig:combined_variR}
\end{figure}

\begin{equation}
\label{magnsigntonoiseSUPsigma}
   m_{\xi^1}^\sigma = \mathrm{erf}\left[\dfrac{1}{\sqrt{2}}\dfrac{1}{\sqrt{\rho_1  +\gamma\dfrac{( \a^2 +\b^2)}{\mu_{1,\sigma}^2}}}\right],
\end{equation}
\begin{equation}
\label{magnsigntonoiseSUPtau}
  m_{\eta^1}^\tau = \mathrm{erf}\left[\dfrac{1}{\sqrt{2}}\dfrac{1}{\sqrt{\rho_2  +\gamma\theta^2\dfrac{( \a^2 +\c^2)}{\mu^2_{1,\tau}}}}\right],
  \end{equation}
  \begin{equation}
  \label{magnsigntonoiseSUPphi}
  m_{\chi^1}^\phi = \mathrm{erf}\left[\dfrac{1}{\sqrt{2}}\dfrac{1}{\sqrt{\rho_3  +\gamma\alpha^2\dfrac{( \b^2 +\c^2)}{\mu^2_{1,\phi}}}}\right].
\end{equation}

\noindent
Similarly, by substituting Eqs.~\eqref{eq:mom1Campo} and \eqref{eq:mom2Campo}, which contain the moments of the post-synaptic field for the unsupervised setting, we obtain the following thresholds:
\begin{equation}\label{magnsigntonoiseUNSUPsigma}
    m_{\xi^1}^\sigma =\mathrm{erf}\left[\dfrac{1}{\sqrt{2}}\dfrac{1}{\sqrt{\rho_1  +\dfrac{\gamma}{(1+\rho_1)\mu_{1,\sigma}^2}\left( \a^2\dfrac{1+\rho_{12}}{1+\rho_2} + \b^2\dfrac{1+\rho_{13}}{1+\rho_3}\right)}}\right],
\end{equation}
\begin{equation}\label{magnsigntonoiseUNSUPtau}
m_{\eta^1}^\tau = \mathrm{erf}\left[\dfrac{1}{\sqrt{2}}\dfrac{1}{\sqrt{\rho_2  +\dfrac{\gamma\theta^2}{(1+\rho_2)\mu_{1,\tau}^2}\left( \a^2\dfrac{1+\rho_{12}}{1+\rho_1} + \c^2\dfrac{1+\rho_{23}}{1+\rho_3}\right)}}\right], 
\end{equation}
\begin{equation}\label{magnsigntonoiseUNSUPphi}
m_{\chi^1}^\phi = \mathrm{erf}\left[\dfrac{1}{\sqrt{2}}\dfrac{1}{\sqrt{\rho_3  +\dfrac{\gamma\alpha^2}{(1+\rho_3)\mu_{1,\phi}^2}\left( \b^2\dfrac{1+\rho_{13}}{1+\rho_1} + \c^2\dfrac{1+\rho_{23}}{1+\rho_2}\right)}}\right]. 
\end{equation}

\noindent
In both scenarios, the remaining magnetizations in Eq.~\eqref{setofmagn} are all zero, so we conclude that the task of \textit{generalized pattern recognition} was successfully performed in a single step, as the network recovered the pattern triplet $(\bm\xi^1,\bm\eta^1,\bm\chi^1)$ with the corresponding non-zero magnetizations.



\subsection{Critical loads and thresholds for learning} \label{subsec:criticalLoad}

An inspection of  the explicit expressions  of the magnetization of the three layers (derived in the previous subsection, see Eqs.~\eqref{magnsigntonoiseSUPsigma}–\eqref{magnsigntonoiseSUPphi} and \eqref{magnsigntonoiseUNSUPsigma}–\eqref{magnsigntonoiseUNSUPphi} for the supervised and the unsupervised scenario, respectively) allows to assess the necessary amount of examples that have to be provided to the network in order to grant it the correct retrieval of the various archetypes hidden in these noisy representations. Indeed, when the network retrieves one of the archetypes, we require that the one-step magnetization of the pertaining archetype gets larger than $\erf(\Sigma)$, where $\Sigma \in \mathbb{R}^+$ defines tolerance level, 
thus obtaining the criteria for supervised and unsupervised settings (Eqs.~\eqref{eq:sogliasup} and \eqref{eq:sogliaunsup}, respectively, where we report only those related to layer $\bm\sigma$ for brevity):

\begin{equation}
\label{eq:sogliasup}
    \dfrac{1}{\sqrt{2}}\dfrac{1}{\sqrt{\rho_1  +\gamma\dfrac{( \a^2 +\b^2)}{\mu_{1,\sigma}^2}}} > \Sigma \,,
\end{equation}

\begin{equation}
\label{eq:sogliaunsup}
    \dfrac{1}{\sqrt{2}}\dfrac{1}{\sqrt{\rho_1  +\dfrac{\gamma}{(1+\rho_1)\mu_{1,\sigma}^2}\left( \a^2\dfrac{1+\rho_{12}}{1+\rho_2} + \b^2\dfrac{1+\rho_{13}}{1+\rho_3}\right)}} > \Sigma \,.
\end{equation}

\noindent
We set this tolerance level to $\Sigma = 1/\sqrt{2}$, which corresponds to the fairly standard condition
\begin{equation*}
    \mathbb{E}\left[h^{\xi^1}_{i}(\bm\eta^1,\bm\chi^1,t=0) \xi^1\right] > \sqrt{\text{Var}\left[h^{\xi^1}_{i}(\bm\eta^1,\bm\chi^1,t=0) \xi^1\right]},
\end{equation*}
and determines a lower bound for $M$ that guarantees the signal (l.h.s.) to be prevailing over the noise (r.h.s.).
\newline
Therefore, the stability condition in the supervised case  identifies the threshold for learning as the minimal amount of examples that have to be provided to the network such that, in future expositions to noisy examples,  its layers $\bm\sigma$, $\bm\tau$, and $\bm\phi$  retrieve their respective patterns $\bm\xi^1$, $\bm\eta^1$, and $\bm\chi^1$, namely the archetypes rather than the examples. These conditions are expressed as:
\begin{equation*}
\begin{array}{lll}
\rho_1  +\gamma\dfrac{( \a^2 +\b^2)}{\mu_{1,\sigma}^2} < 1,
\quad
\rho_2  +\gamma\theta^2\dfrac{( \a^2 +\c^2)}{\mu^2_{1,\tau}} < 1,
\quad
\rho_3  +\gamma\alpha^2\dfrac{( \b^2 +\c^2)}{\mu^2_{1,\phi}} < 1.
\end{array}
\end{equation*}
Note that the above knowledge  provides a priori, before starting to train the network, a quantitative condition necessary to guarantee the success of the training process and, thus, archetype storage and successive retrieval, across all the three layers.

In the simplest scenario of uniform dataset per layer —namely, where the three datasets provided to the layers consist of an equal number of examples generated with an identical noise level (thereby sharing the same entropy)— and assuming a network with symmetric layer sizes (hence setting the control parameters to values $\alpha=\theta=1$, $\a=\b=\c=1$, $r_1=r_2=r_3=r$, and $M_1=M_2=M_3=M$), the previously calculated learning thresholds simplify to
\begin{equation}\label{eq:thrSimpleSup}
\rho + \gamma \dfrac{(1 + \rho)^2}{2}<1,
\end{equation}
where $\rho = \frac{1-r^2}{M r^2}$ denotes the entropy of each dataset. 
In the low load regime, with $\gamma=0$, expression \eqref{eq:thrSimpleSup} reduces to
\begin{equation}
M_{\times} (r) =\dfrac{1-r^2}{ r^2},
\end{equation}
where $M_\times$ represents the minimum number of examples that satisfy inequality \eqref{eq:thrSimpleSup}. Notably, this reveals the power-law scaling $M_\times(r)\propto r^{-2}$ in the threshold for learning. This scaling behavior aligns with the findings of \cite{meir2020power}, who demonstrated that such power-law relationships indicate that the test error converges as a power-law with increasing dataset size, providing a theoretical benchmark for evaluating the training complexity.

\noindent
For the case where $0\neq\gamma \leq 1$\footnote{The assumption $\gamma\leq 1$ is well justified. In fact, for any network configuration—whether in a supervised or unsupervised setting—the maximum load-bearing capacity $\gamma_{\mathrm{max}}$, remains significantly below 1 (see Fig.~\ref{fig:TAM_Gamma_Synapses} (left)).}, we obtain:
\begin{equation*}
M_{\times}(r) = \dfrac{(1-r^2)(1 + \gamma)}{r^2(\gamma - 2)} + \sqrt{\dfrac{(1-r^2)^2(1 + 4\gamma)}{r^4(\gamma - 2)^2}}
\end{equation*}
which similarly preserves the power scaling $M_\times(r)\propto r^{-2}$ in the high-load regime. 
The persistence of this scaling relationship across different load conditions underscores the robustness of the power-law relation, further supporting the theoretical framework established by \cite{meir2020power}.

In the unsupervised setting, applying the same simplifying assumptions as in the supervised case yields the following learning threshold for all three layers:
\begin{equation}\label{thrSimple}
\rho + \frac{\gamma}{2}(1+\bar\rho)<1,
\end{equation}
where $\bar\rho = \frac{1-r^4}{M r^4}$. Following a similar analysis as in the supervised setting, the low-load regime gives:
\begin{equation*}
M_\times(r)= \frac{1-r^2}{r^2},
\end{equation*}
which maintains the power-law scaling $M \propto r^{-2}$. In the high-load regime (i.e., $0\neq\gamma \leq 1$), we find:
\begin{equation*}
M_\times(r,\gamma) = \frac{2r^4 - 2r^2 - \gamma + r^4\gamma}{r^4 \left(\gamma-2\right)},
\end{equation*}
which exhibits the power scaling $M \propto r^{-4}$.

\medskip
\noindent
These scaling behaviors are particularly significant when considered alongside the complementary research of \cite{uzan2019biological}, which provides evidence that similar scaling laws emerge in biological learning scenarios. Their work suggests that such power-law scaling not only characterizes artificial neural networks but also reflects fundamental properties of biological learning systems. Thus, the power-law scaling identified in our threshold analysis connects to a broader theoretical framework that links critical phenomena in statistical mechanics with learning dynamics across both artificial and biological neural networks.

\vspace{3mm}
\noindent To conclude this section, we present\footnote{We will now only present the plots for the supervised setting; for the unsupervised case we obtain similar results.} Monte Carlo simulations that involve the temporal evolution of the three fields $h^{\sigma}_{i}(\bm\eta,\bm\chi)$, $h^{\tau}_{j}(\bm\xi,\bm\chi)$ and $h^{\phi}_{i}(\bm\xi,\bm\eta)$. In particular, in Fig.~\ref{MagnVsTemp} we plot the values of the magnetization of the layer $\bm \sigma$ as a function of the temperature ($\beta^{-1}$) for different values of the dataset quality: the algorithm is reported in the pseudo-code {\em Algorithm 2}. In Fig.~\ref{fig:combined_variR} (left) we test the network’s ability to reconstruct patterns under different values of $M_1,M_2,M_3$, while in Fig.~\ref{fig:combined_variR} (right) we plot the overlap between the examples and the neurons states $\bm{\sigma}$ in function of the dataset entropy.

%% file: sections/nFindings.tex
In this section, we complement our theoretical analysis with numerical investigations of the performance of TAM model under both supervised and unsupervised learning paradigms. Through the implementation of Monte Carlo simulations and signal-to-noise analyses, we validate the analytical predictions derived in the previous section and explore the dynamic behavior of the network across different control parameters' regimes.

Our numerical approach focuses on assessing retrieval accuracy, critical load capacity and computational phase transitions by systematically varying  dataset quantity and quality as well as synaptic coupling strengths. Additionally, we examine the influence of thermal noise and asymmetry in network architecture on the model’s capacity to learn, store and retrieve patterns. The results obtained provide empirical support for our theoretical findings, shedding light on the fundamental differences between supervised and unsupervised learning in TAM networks.

By leveraging numerical simulations, we further investigate the emergence of cooperative behavior in the network, analyzing how information is distributed and reinforced across different layers. These insights are essential for understanding the scalability and robustness of hetero-associative memory models in practical applications.

\subsection{Pattern processing}\label{subsec:applications}

\medskip
Using Monte Carlo simulations, we present the network with various challenges; the first one is the \emph{generalized pattern reconstruction}. 
Unlike networks equipped with only one visible layer, here the presence of three visible layers allows for the simultaneous retrieval of three patterns, one per layer. 
Therefore, we initialize the first layer with a specific example, denoted as $\boldsymbol{\tilde\Xi}^{1}$, taken from a test dataset on which the network has not been trained, having the same distribution as the training dataset.
The remaining two layers are initialized randomly. In Fig.~\ref{fig:MC} (right) we can see that the network can successfully retrieve the triple pattern $({\bm\xi}^1, {\bm\eta}^1, {\bm\chi}^1)$. 
Furthermore, this result is in perfect agreement with the signal-to-noise analysis presented in Subsec.~\ref{subsec:montecarlo}.

\begin{figure}
    \centering
    \fbox{\includegraphics[width=0.44\linewidth]{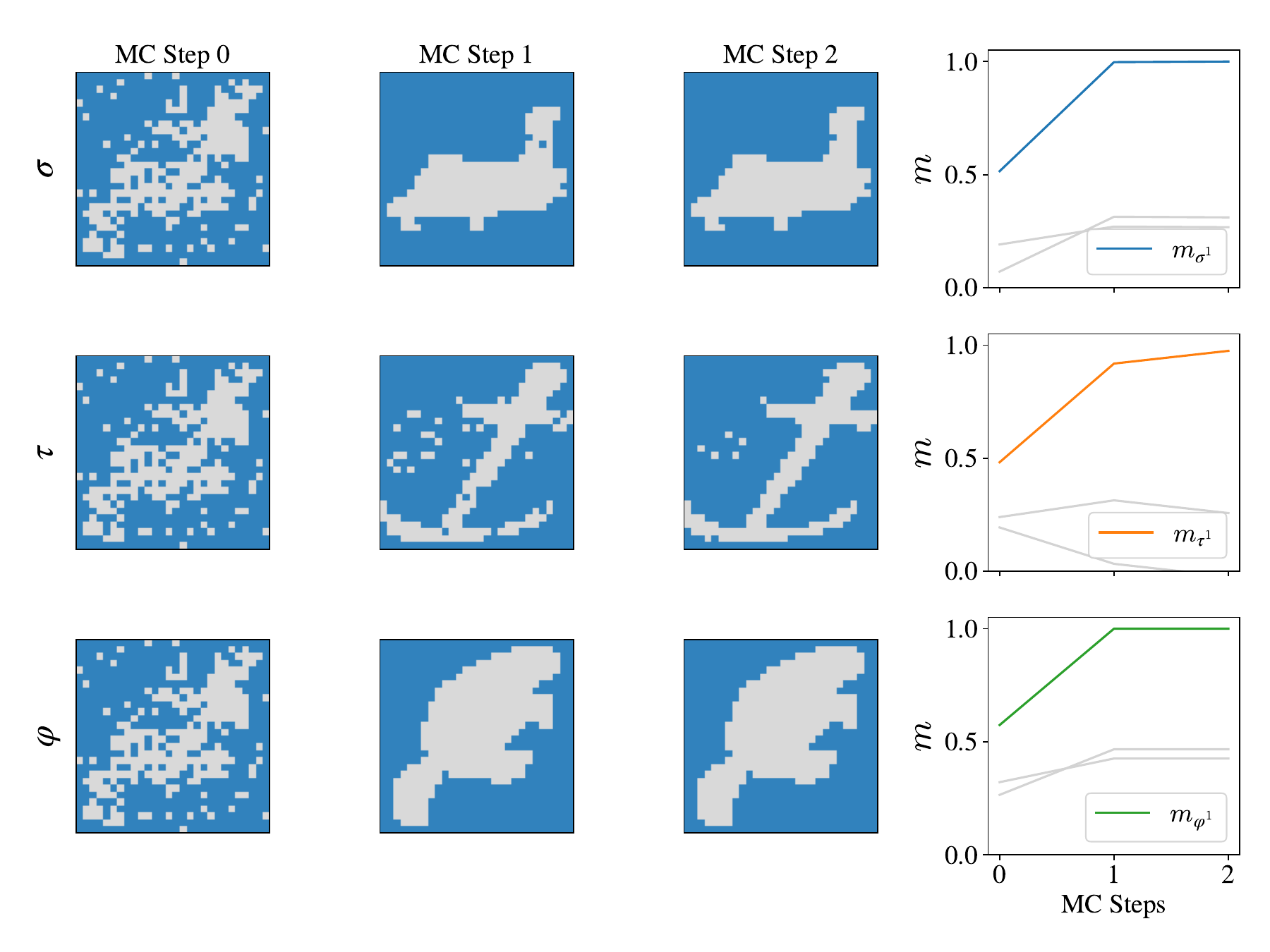}}
    \hfill
    \fbox{\includegraphics[width=0.44\linewidth]{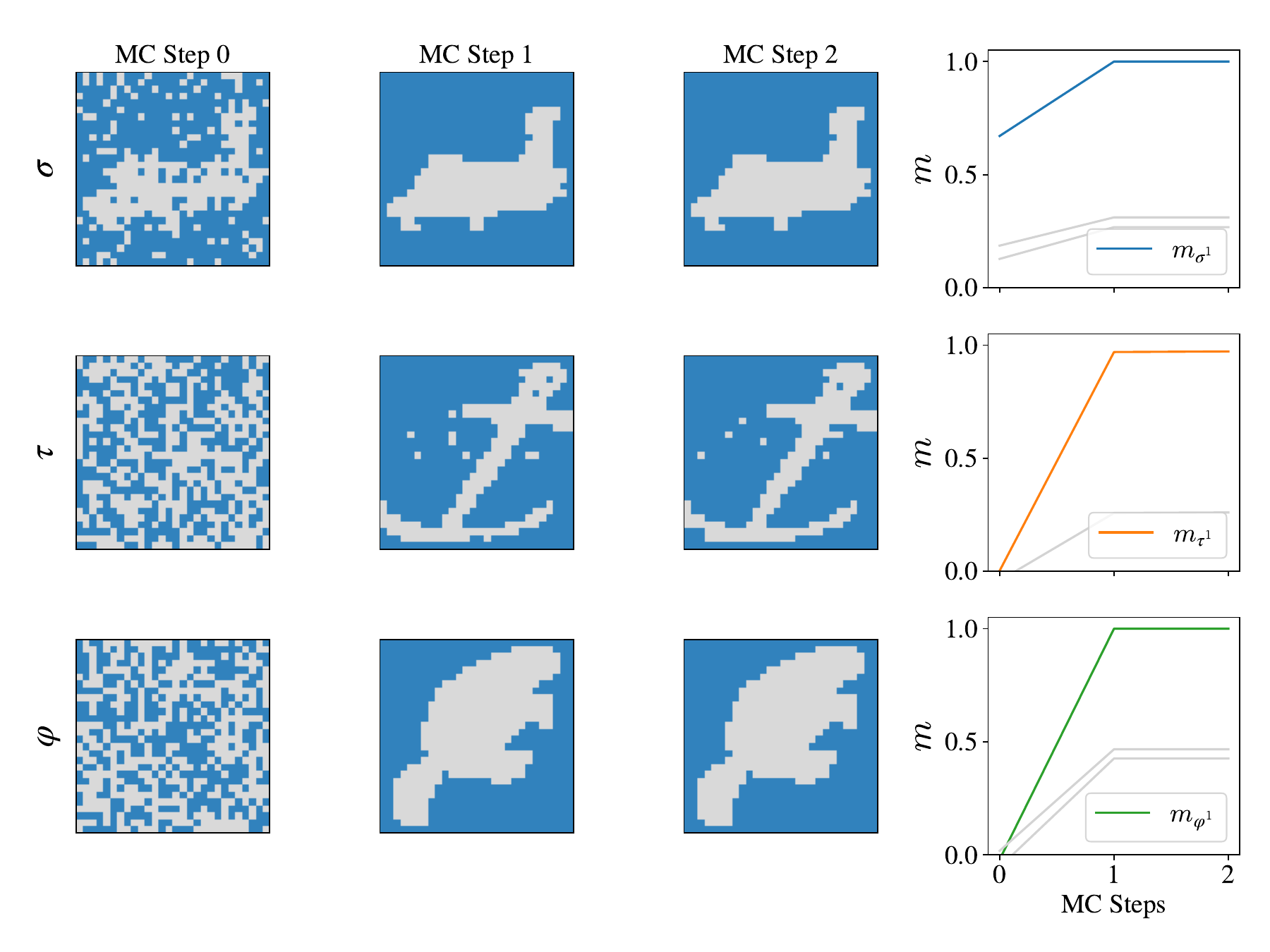}}
    \hfill
    \fbox{\includegraphics[width=0.0365\linewidth]{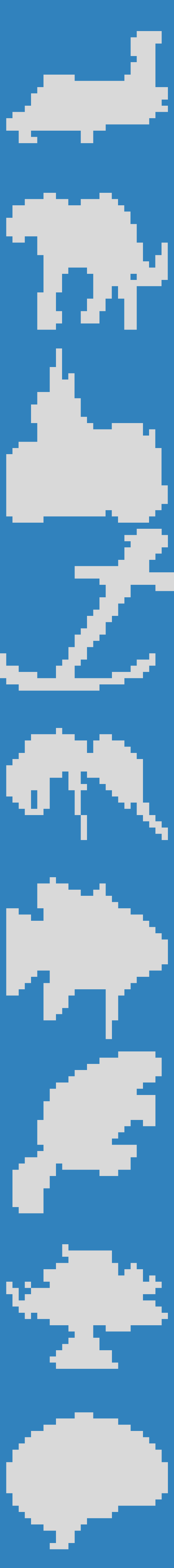}}
    \caption{Monte Carlo simulations illustrating two complementary tasks performed by the symmetric TAM. (left box) Disentanglement of spurious states: a mixture of examples (corresponding to the archetypes “airplane”, “anchor”, and “beaver”) extracted from a test dataset is presented simultaneously to the three visible layers; the evolution of the neuronal configurations (\(\sigma\), \(\tau\), \(\phi\)) and the corresponding Mattis magnetizations (distinguished by color) is tracked during the simulation with \(\beta=5\). 
    (right box) Generalized pattern reconstruction: the first layer is initialized with a specific test example \(\boldsymbol{\tilde\Xi}^{1}\) while the remaining two layers are randomly initialized, and the network is challenged to retrieve the complete triple pattern \(({\bm\xi}^1, {\bm\eta}^1, {\bm\chi}^1)\). In both tasks, simulations are conducted on images of size \(28\times 28\) pixels, with parameters \(N=784\), \(K=3\), \(\bm r=(0.7,0.7,0.7)\), and \(\bm M=(100,100,100)\). Additionally, the nine archetypes (three per layer) are also displayed on the far right of the diagram.}
    \label{fig:MC}
\end{figure}
\par\vspace{3mm}
\noindent As already seen in \cite{TAMstoring}, another interesting task, which turns out to be a special case of generalized pattern recognition, is the \emph{disentanglement} of spurious states.
In such a task, the three layers are given a mixture of examples extracted from a test dataset and the effective recall of the relevant archetype is checked for each of them.
From the plot in Fig.~\ref{fig:MC} (left), we observe that for suitable values of $\beta$, the network performs this task flawlessly.
In addition, we would like to point out that by taking an example that the network has never viewed, we have shown its actual generalization capabilities.

\medskip
Overall, these numerical investigations confirm that the TAM model not only supports a hippocampus-inspired division of labor—pattern separation in one layer and pattern completion in another—but also achieves a powerful synergy among the three interconnected layers. This synergy ensures accurate retrieval (generalized pattern reconstruction) and effective disentanglement of overlapping or corrupted patterns, thus providing a biologically plausible yet computationally efficient framework for multi-layer associative memories.

 \begin{figure}[t]     \centering
     \includegraphics[width=1.0\linewidth]{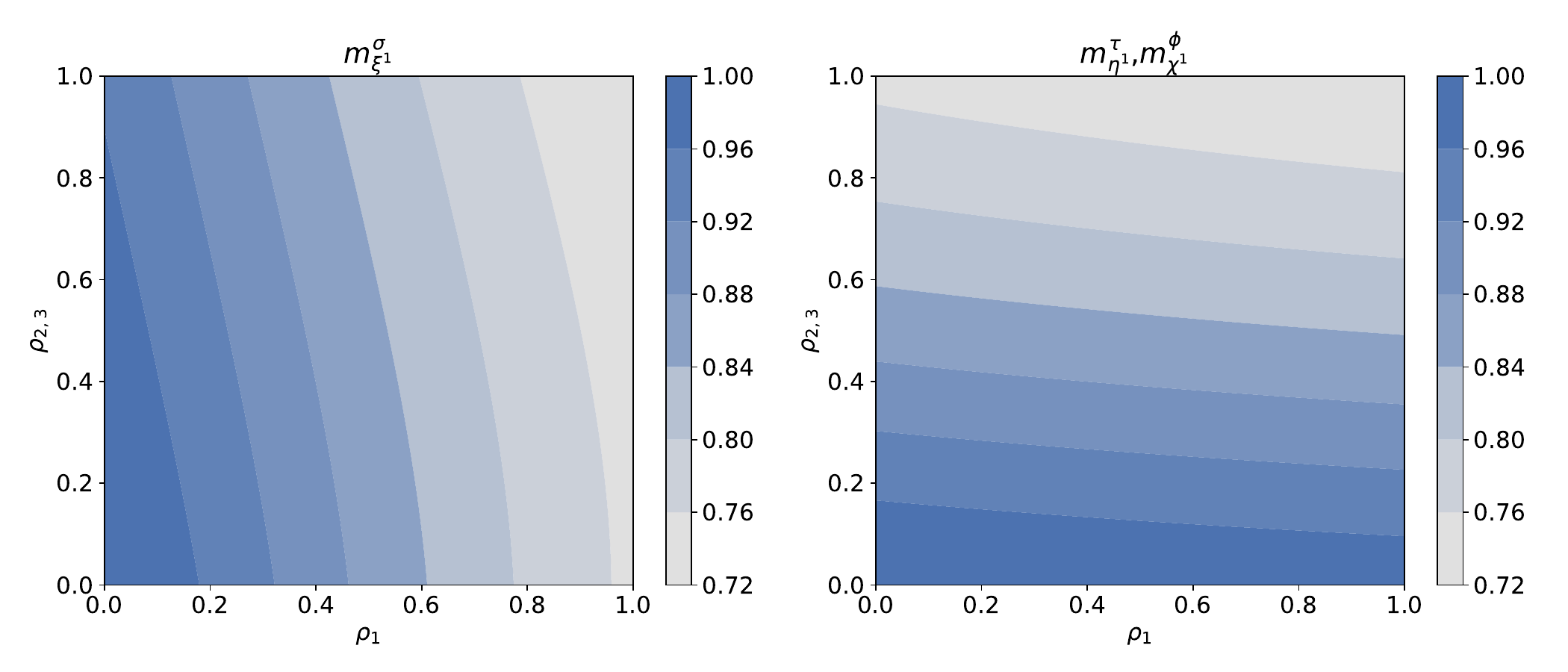}
     \caption{Contour plot of magnetization values as a function of entropy values, $\rho_1,\rho_2,\rho_3 \in [0,1]$, in the supervised setting. We also set $\gamma=0.25$. For simplicity we have set $\rho_2=\rho_3$ and
     consequently, the magnetization
     values corresponding to layer $\bm{\tau}$ are identical to those of the layer $\bm{\phi}$ (left). We observe that the increase in entropy $\rho_{2/3}$ marginally influences the decrease in the value of $m_{\xi^1}^{\sigma}$ whereas the decrease is more substantial as $\rho_1$ increases (right). An analogous situation occurs for magnetizations $m_{\eta^1}^{\tau}$ and $m_{\chi^1}^{\phi}$: here we observe that the dominant entropy is $\rho_{2/3}$.}
     \label{Magnerho}
 \end{figure}

\subsection{Cooperative behaviour in TAM model} \label{sec:cooperativeness}

The analysis of our TAM network reveals that the retrieval performance of each layer is not dictated solely by its local noise level but rather emerges from the intricate interplay among all layers. Under simplified approximations—such as the one-step magnetization approach based on the expressions in \eqref{magnsigntonoiseSUPsigma} for the symmetric case with \((\a,\b,\c)=(1,1,1)\) and fixed $\gamma$—each layer appears to operate in isolation, with its recovery performance predominantly controlled by its own entropy parameter. Indeed, as Fig.~\ref{Magnerho} illustrates, the magnetization values of each layer are only weakly influenced by the entropy levels of the other layers, reinforcing the self-centered nature of this local approximation in which inter-layer effects are neglected. However, when the full dynamics of the network are taken into account, a more realistic picture emerges from Monte Carlo simulations. In these simulations, the system evolves according to the Boltzmann–Gibbs measure, 
\begin{figure}[t]
     \centering
     \includegraphics[width=1.0\linewidth]{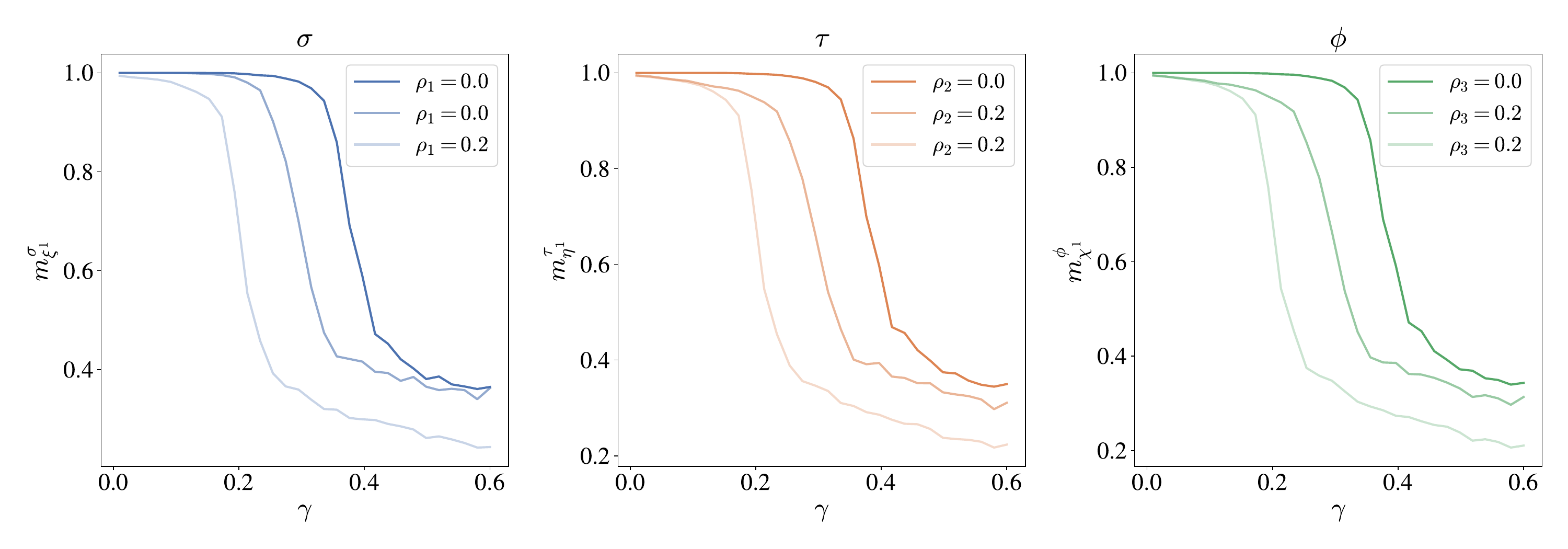}
     \caption{
Monte Carlo simulations in the supervised setting for a symmetric network with \(N_1 = N_2 = N_3 = 1000\), \(M_1 = M_2 = M_3 = 20\), and \(\beta = 1000\). 
The plots display the magnetization as a function of the loading capacity \(\gamma\) across the three layers of the network, under different entropy regimes. 
Curves with more saturated colors correspond to the low-entropy regime, whereas lighter curves represent the high-entropy regime. 
In the intermediate-entropy case, we set \(\rho_1 = 0.0\) for the first layer and \(\rho_2 = \rho_3 = 0.2\) for the second and third layers. 
The simulations reveal the emergence of a cooperative behavior among the layers, which becomes even more evident in the subsequent plots showing the phase diagram of the model.
}
     \label{montecarlolayersiaiutano}
 \end{figure}
\begin{equation}
    \mathbb{P}(\sigma,\tau,\phi\,|\,\Xi,\Theta,\Upsilon) \propto \exp\!\Big[-\beta\, \mathcal{H}(\sigma,\tau,\phi\,|\,\Xi,\Theta,\Upsilon)\Big],
\end{equation}

\noindent
where the Hamiltonian \(\mathcal{H}\) incorporates both intra- and inter-layer couplings. As shown in Fig.~\ref{montecarlolayersiaiutano}, the evolution of magnetizations as a function of the load \(\gamma\) is depicted for three configurations: a low-entropy regime (blue, with \(\rho=0.0\) in all layers), a high-entropy regime (red, with \(\rho=0.2\) uniformly), and an intermediate configuration (orange, where \(\rho_1=0.0\) and \(\rho_2=\rho_3=0.2\)). The results already suggest an emerging cooperative behavior among the layers, a phenomenon that will be further explored through the analysis of phase diagrams.

\begin{figure}[t]
    \centering
    \begin{subfigure}{0.49\textwidth}
        \centering
        \includegraphics[width=\linewidth]{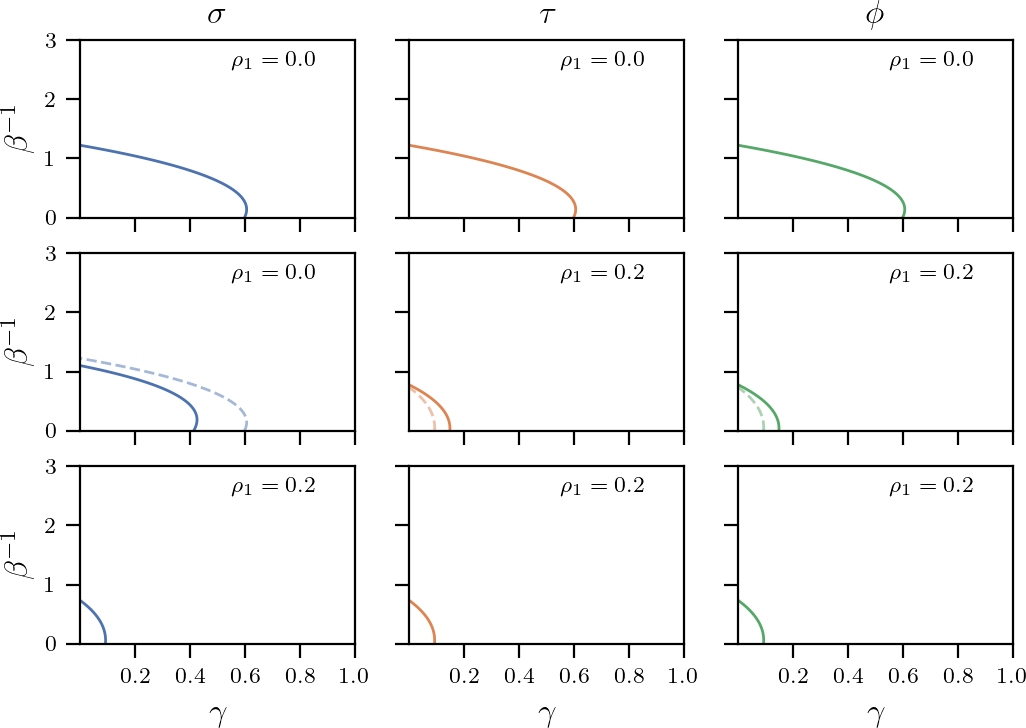}
    \end{subfigure}
    \hfill
    \begin{subfigure}{0.49\textwidth}
        \centering
        \includegraphics[width=\linewidth]{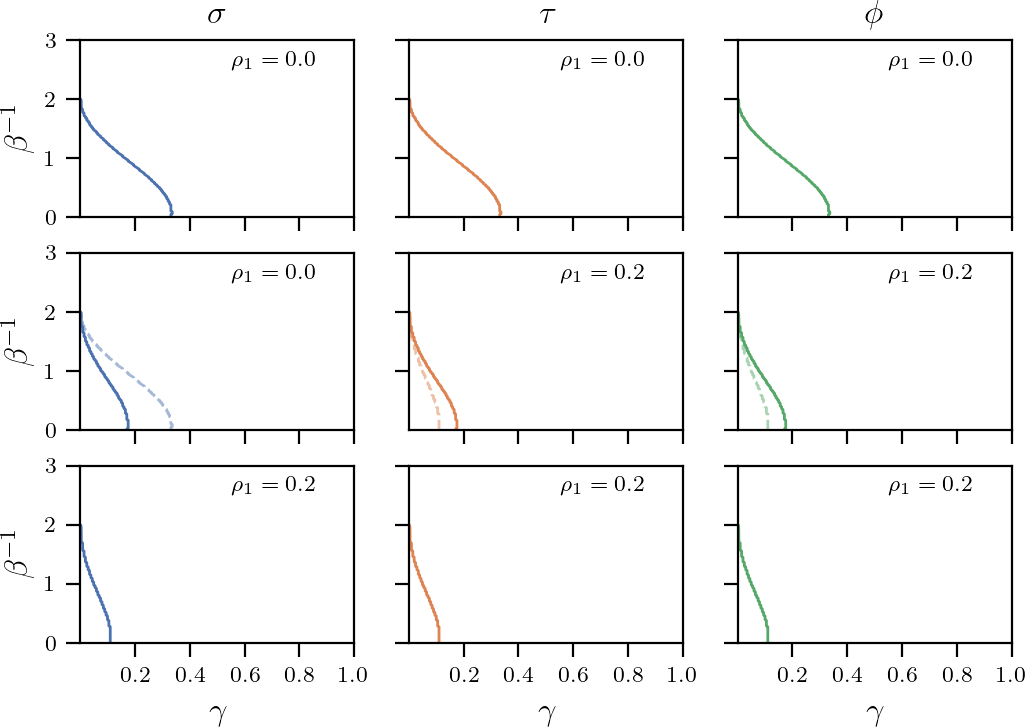}
    \end{subfigure}
    
    \caption{Phase diagrams comparing unsupervised (left) and supervised (right) learning settings across different noise configurations in the $(\gamma,\beta^{-1})$ plane. In both panels, the entropy parameters $(\rho_1,\rho_2,\rho_3)$ are fixed to the same values: top row $(0.0,0.0,0.0)$; middle row $(0.0,0.2,0.2)$; and bottom row $(0.2,0.2,0.2)$. The dashed lines delineate how the retrieval regions \emph{would appear} in the absence of the cooperative phenomenon among layers, highlighting the benefits of inter-layer cooperation. The diagrams illustrate how noise levels across the three layers ($\bm\sigma$, $\bm\tau$, and $\bm\phi$) influence retrieval performance and highlight the cooperative interplay among layers. For all cases, $\alpha=\theta=1$ and $(g_\sigma,g_\tau,g_\phi)=(1,1,1)$. 
    The retrieval regions represented in the unsupervised diagrams are obtained for magnetization values greater than $0.9$. The two models were solved using different techniques and approximations, which naturally lead to different outcomes.}
    \label{fig:combined-phase-diagrams}
\end{figure}

In fact, a complete understanding of the cooperative behavior is achieved through an analytical treatment based on the self-consistency equations, from which phase diagrams are derived to map the retrieval region in the \((\gamma,\beta^{-1})\) plane for each layer. As depicted in Fig.~\ref{fig:combined-phase-diagrams}, the ideal retrieval scenario is observed when there is no noise (\(\rho_1=\rho_2=\rho_3=0.0\)). However, when noise is introduced heterogeneously—for instance, with \(\rho_1=0.0\) and \(\rho_2=\rho_3=0.2\)—a redistribution effect occurs: the recovery region of the more structured layer contracts while those of the noisier layers expand, thus balancing the overall performance. In contrast, uniformly high noise (\(\rho_1=\rho_2=\rho_3=0.2\)) suppresses retrieval in all layers. These observations confirm that the cooperative mechanism—where mutual influences between layers enhance retrieval—is an intrinsic property of the TAM model \cite{alessandrelli2025beyond}. In summary, the progression from the local one-step magnetization approximation, through Monte Carlo simulations, to a full phase diagram analysis demonstrates that \emph{cooperativeness} in the TAM network is an emergent property resulting from complex inter-layer interactions, providing a robust theoretical foundation for the design of next-generation associative memory networks with enhanced robustness and balanced retrieval capabilities.

\begin{figure}
    \centering
    \begin{subfigure}{0.44\textwidth}
    \centering
    \includegraphics[width=\linewidth]{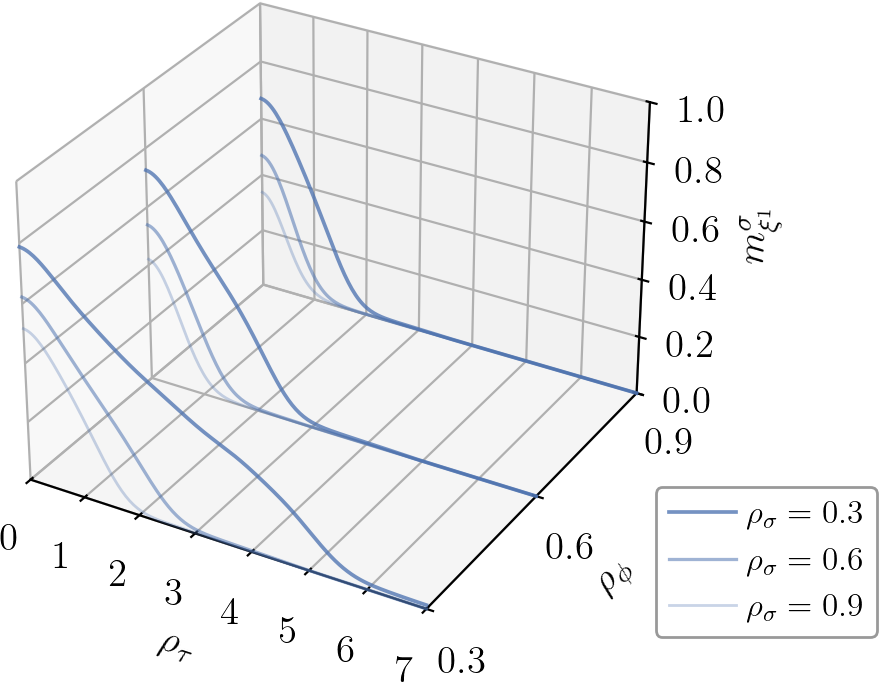}
    \end{subfigure}
    \hfill
    \begin{subfigure}{0.44\textwidth}
    \centering
    \includegraphics[width=\linewidth]{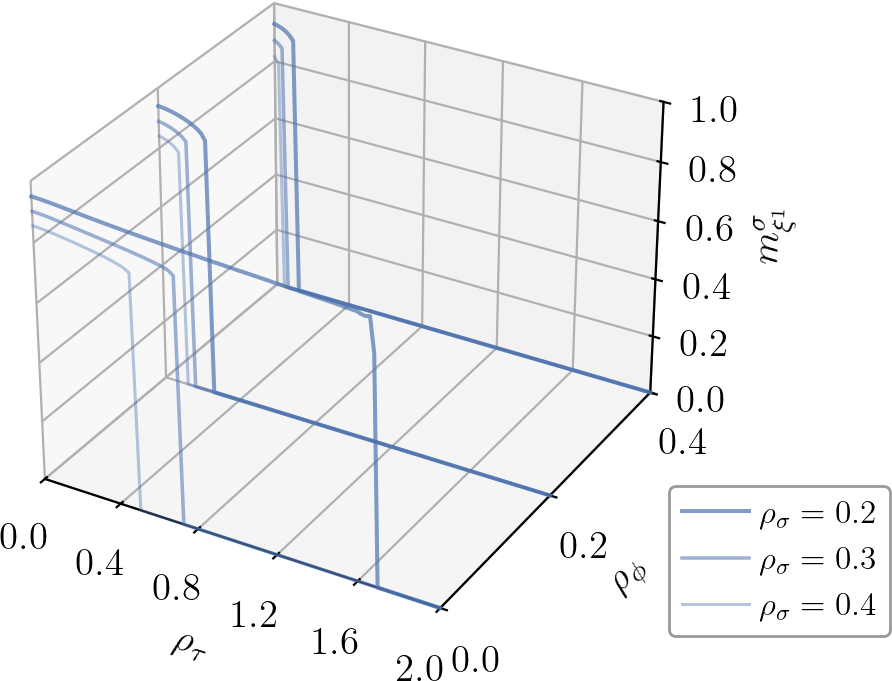}
        
    \end{subfigure}
    \caption{Numerical solutions of Eqs.~\ref{eq:selfMagnSigma} (Unsupervised scenario, on the left) and \ref{seldefinitive} (Supervised scenario, on the right) are presented. Both plots are configured as follows: the vertical axis (z-axis) represents the magnetization $m^{\sigma}_{\xi^1}$, while the horizontal axes correspond to $\rho_\phi$ (x-axis) and $\rho_\tau$ (y-axis). The parameters are fixed to $\beta = 1000$, $\gamma = 0.1$, and $\alpha = \theta = \a = \b = \c = 1$ in both settings.
    Three configurations are considered: for the Supervised case, we use $\rho_\sigma = 0.2$, $0.3$, and $0.4$; for the Unsupervised case, $\rho_\sigma = 0.3$, $0.6$, and $0.9$. The color gradient encodes the value of $\rho_\sigma$: darker blue lines correspond to lower values, while lighter blue lines indicate higher values. We report three slices of the resulting surface at selected values of $\rho_\phi$: for the Supervised case, $0.0$, $0.2$, and $0.4$, while for the Unsupervised case, the slices coincide with the corresponding values of $\rho_\sigma$.
    In both diagrams, the cooperative behavior of the network is clearly visible. Moreover, the entropy parameters $\rho_\tau$ and $\rho_\phi$ significantly influence the performance trend of the $\sigma$-layer.
    The key difference between the two plots lies in the emergence of a clear phase transition in the Supervised setting, which is absent in the Unsupervised scenario. This discrepancy is likely due to the strong approximations employed in the analytical treatment of the Unsupervised model, which may suppress such critical phenomena.}
\end{figure}

%% file: sections/conclusions.tex
\par
In this work, we have developed a comprehensive learning framework for Three-Directional Associative Memory (TAM) models by extending the classical Hebbian storage paradigm toward both supervised and unsupervised learning protocols. Our analysis —based on Guerra’s interpolation method— provides explicit self-consistency equations for the order parameters that quantify how dataset quality and quantity, synaptic density and storage capacity interact to shape the global performance of the network. The latter is a tripartite architecture where each layer communicates with all the others and it is exposed to noisy examples of an extensive number of distinct archetypes that it must infer, store and eventually retrieve. We have established a rigorous statistical-mechanical description of the network’s learning thresholds and retrieval performances that allow a careful calibration of the network’s control parameters (such as the relative layer sizes, the quality of the datasets, and the strength of the inter-layer couplings) in order to let it operate in the optimal regime. These results are in line with previous literature addressing related models, where the storage capacity and retrieval performance exhibit a power-law scaling with respect to the noise parameters \cite{agliariEmergence,alemanno2023supervised,centonze2024statistical}.
\par
Further, we highlighted a new emergent behavior, namely the   spontaneous cooperation  among layers when inputted with  datasets containing different amounts of information. Indeed, rather than operating in isolation, taking advantage of mutual interactions, the layers exhibit a form of reciprocal reinforcement whereby the retrieval capability of one layer (that, e.g., experienced more examples or less noisy ones) can compensate for the relative weakness of another layer (that, e.g., has been exposed to fewer or more corrupted examples). This \emph{cooperative effect} manifests itself in the phase diagrams, where the retrieval region is found to be determined not solely by the entropy of the local dataset but by the collective interplay of dataset entropies across the layers: the ultimate result of this cooperative effect shines in the amplitude of the retrieval region of the various layers that turns out to be the same, regardless of the information experienced by a given layer. 
On the computational side, MCMC simulations corroborate the analytical predictions, demonstrating robust retrieval performance even when the network is exposed to highly corrupted input patterns. The simulations further confirm the existence of computational phase transitions that demarcate regions of efficient learning —resulting in correct retrieval of archetypes during further exposition of examples to the network— from those where retrieval fails, due to a  training procedure that did not suffice or to an overload of patterns to be kept stored in memory. 
\par\vspace{3mm}
Looking ahead, several possible directions for future research stem from our study. First, while our analysis has been carried out under the replica-symmetric assumption, it remains an open question how the inclusion of  replica symmetry breaking  might further refine the information processing capabilities of the network, especially in the vicinity of the critical thresholds. 
Moreover, an in-depth analysis of the way  that layers interact to cooperate is still to be achieved both in terms of structural noise (e.g. by systemically providing the network with the standard datasets) and variations in dataset entropy (e.g. by letting the network deal with very heterogeneous -in their information content- datasets) as well as in terms of fine-tuned architecture (e.g. nor the role of the asymmetry among the layers neither the density of the synapses within the layers have been investigated to check their role in this cooperativity): the resulting computational analysis could offer guidelines for the design of optimal networks to handle very different datasets.
\par
Second, rather naturally within the statistical mechanical route, a dual representation of the TAM model emerges in terms of three Restricted Boltzmann Machines (RBMs) coupled via their inner layers that operate in the one-hot coding setting (i.e., they implement the grandmother cell scheme).  The correspondence between the hidden layers in the RBM formulation and the “\emph{grandmother}” cells observed in biological networks \cite{grandmother} provides a compelling narrative that aligns with both connectionist and localist theories of information processing and brings bio-inspired neural networks closer to statistical mechanics learning \footnote{For further reading on this duality and its implications, works by Krotov and Hopfield \cite{krotov2016dense,krotov2020large} and recent surveys by Ramsauer \cite{ramsauer2020hopfield} offer critical perspectives on energy-based models and their modern applications.}.
Another promising direction involves the application of the TAM framework to dynamic and temporally evolving datasets. In many real-world scenarios, the input data are not static but rather evolve according to complex temporal patterns \cite{piermarocchi1}. Our preliminary investigations into frequency modulation and temporal sequence retrieval indicate that TAM networks can successfully \emph{disentangle} mixtures of patterns occurring at different frequencies. This capability is of paramount importance in fields such as audio processing, video analysis, and even neuroscience, where the timing and sequencing of inputs carry critical information. 
\par
In addition, the interplay between supervised and unsupervised learning protocols in TAM networks warrants further exploration. Our results indicate that even in the absence of explicit supervisory signals, the network can infer underlying archetypes from noisy data through intrinsic statistical correlations. However, the supervised setting provides a clear advantage by guiding the network toward the correct partitioning of the data. Future work could aim to develop \emph{hybrid learning protocols} that dynamically adjust the balance between supervised and unsupervised learning based on the observed quality of the input data; furthermore, the exploration of semi-supervised schemes (where solely a fraction of data is labeled within the datasets) would be an interesting in-depth, eventually simple to obtain. 
\par
Finally, the broader implications of our findings for both artificial intelligence and neuroscience are significant. The emergence of cooperative dynamics in associative networks challenges traditional views on the independence of memory layers and suggests that the collective interplay of multiple subsystems can lead to enhanced cognitive capabilities. For the machine learning community, this insight offers a new paradigm for the design of deep neural networks that harness the power of distributed processing to achieve superior generalization. For neuroscientists, our results provide a quantitative framework for understanding how biological neural circuits might exploit cooperative interactions to maintain robust and efficient memory retrieval, even in the face of noise and uncertainty. As recommended by Agliari et al. \cite{agliariRedundantRepresentation,agliari2021transport} and further discussed in recent reviews \cite{albanese2022replica}, integrating these insights could foster interdisciplinary research that bridges the gap between artificial and biological systems.

%% file: appendices/proofGuerra.tex
In this section we exploit Guerra's interpolation \cite{Guerra2,Fachechi1} to express the quenched free energy explicitly as a function of the control and order parameters of the theory. The plan is to use one-parameter interpolation, $t \in [0,1]$ and, as standard in high-storage investigations \cite{Amit,CKS}, we assume that a finite number of patterns (actually, just one per layer without loss of generality) are retrieved, say the triplet $(\xi^1,\eta^1,\chi^1)$, and these patterns act as the {\em signal}, while all remaining ones (i.e., those with labels $\nu \neq 1$) act as {\em quenched noise} against the retrieving process.
\newline
As we deal with averages of observables, it is useful to define them as follows.
\begin{defin}
Given a function $F(\bm{\sigma},\bm{\tau},\bm{\phi})$, depending on the neuronal configuration $(\bm{\sigma},\bm{\tau},\bm{\phi})$, 
the Boltzmann average, namely the average over the distribution \eqref{BGmeasure}, is denoted by $\omega( F(\bm{\sigma},\bm{\tau},\bm{\phi}))$ and defined as
\begin{equation}
 \omega(F(\bm{\sigma},\bm{\tau},\bm{\phi}))= \frac{\sum_{\{\sigma \}}^{2^{N_1}} \sum_{\{\tau \}}^{2^{N_2}} \sum_{\{\phi \}}^{2^{N_3}} F(\bm{\sigma},\bm{\tau},\bm{\phi}) e^{-\beta \mathcal{H}_{\bm N}(\boldsymbol{s}|\boldsymbol{x})}}{\sum_{\{\sigma \}}^{2^{N_1}} \sum_{\{\tau \}}^{2^{N_2}} \sum_{\{\phi \}}^{2^{N_3}} e^{-\beta \mathcal{H}_{\bm N}(\boldsymbol{s}|\boldsymbol{x})}}.
\end{equation}
\end{defin}
\begin{defin}
Given a function of the examples $g(\bm{\Xi}, \bm{\Theta}, \bm{\Upsilon})$ depending on the realization of the $K\times \bm M$ triplets of examples, we introduce the quenched average, namely the average over the Rademacher distributions related to their generation, which is denoted as $\mathbb E [ g(\bm{\Xi}, \bm{\Theta}, \bm{\Upsilon})]$ or $\langle g(\bm{\Xi}, \bm{\Theta}, \bm{\Upsilon}) \rangle$ according to the context, and it is defined as
\begin{align}
\label{eq:mapsexpectation}
    \mathbb{E} [ g(\bm{\Xi}, \bm{\Theta}, \bm{\Upsilon}) ] \equiv \mathbb{E}_{\bm{\Xi}}\mathbb{E}_{\bm{\Theta}}\mathbb{E}_{\bm{\Upsilon}}  \P(\bm{\xi}, \bm{\eta}, \bm{\chi}) g(\bm{\Xi}, \bm{\Theta}, \bm{\Upsilon}).
\end{align}
\end{defin}

\noindent
This definition of the quenched average follows from the assumption that the patterns are statistically independent; therefore, the expectation factorizes over neurons, patterns, and examples.
\begin{defin}
We use the brackets $\langle \cdot \rangle$ to denote the average over both the Boltzmann-Gibbs distribution and the realization of the patterns, that is
\begin{equation}
    \langle F((\bm{\sigma},\bm{\tau},\bm{\phi})|(\bm{\Xi}, \bm{\Theta}, \bm{\Upsilon})) \rangle = \mathbb E [\omega(F((\bm{\sigma},\bm{\tau},\bm{\phi})|(\bm{\Xi}, \bm{\Theta}, \bm{\Upsilon})))].
\end{equation}

\end{defin}

\noindent
We use the usual replica-symmetric assumption, i.e., if $x$ is an order parameter, it does not fluctuate in the thermodynamic limit, that is:
\begin{equation}
    \lim_{\bm N \to \infty} \mathbb{P}(\langle x \rangle) = \delta(\langle x \rangle - \bar{x} ),
\end{equation}
where $\bar{x}$ is the corresponding equilibrium value.
\subsubsection*{Proof for supervised} \label{appsubsec:proofGuerraSup}
In this appendix, using Guerra's interpolation technique, we derive an explicit expression for the quenched free energy, in the Supervised case.  \\ Introducing an interpolating parameter $t \in [0,1]$ and by virtue of the fundamental theorem of calculus we can write the statistical pressure in the following way: 

\begin{eqnarray}
    \mathcal A^{\bm{g}}_{\alpha, \theta, \gamma}(\beta)&\doteq&\mathcal A(t=1) = \mathcal A(t=0) + \int_0^1 ds \left. \bigg[\partial_t \mathcal A(t)\biggr] \right \vert_{t=s} \,,\label{eq:pressioneInterp}
    \\ 
    \mathcal A(t) &=& \lim_{\bm N\to\infty} \frac{1}{L} \mathbb E  \ln \mathcal Z_{\bm{N}}(t) \,,\label{eq:pressioneLimite}
 \end{eqnarray}
where $\mathcal Z_{\bm{N}}(t)$ is the Guerra's interpolating partition function defined hereafter and we will indicate with $\mathcal{A}_{\bm{N}}(t)$ the statistical pressure in the finite limit.

Once introduced real valued auxiliary functions $\psi$ (that play as {\em ad hoc} fields to reproduce the neural interactions effectively) and the set of constants $C$ (that properly tune the variances of the distributions of the real-valued hidden grandmother neurons), the Guerra's interpolating partition function reads as

\begin{equation}\footnotesize
\label{Integralista}
    \begin{array}{lllll}
\mathcal{Z}_{\bm{N}}(t) \stackrel{.}{=} \mathcal{Z}^{\bm g}_{\bm N, \bm M, \bm r}(\beta,\bm J,t)=\SOMMA{\{\sigma\},\{\tau\},\{\phi\}}{2^{N_1},2^{N_2},2^{N_3}} \exp\biggr\{ a(t) 
         +b(t) + J(\bm{\sigma},\bm{\tau},\bm{\phi})\biggl\}\times
         \\\\
         \times\displaystyle\int\mathcal{D}(z\z s\s w\w)\exp\biggr\{ c(t) + d(t) + e(t)\biggl\},
\end{array}
\end{equation}

\noindent
where we have defined

\begin{equation}\footnotesize
    \begin{aligned}
    a(t) \doteq \beta t\biggl( & \a\sqrt{N_1N_2(1+\rho_1)(1+\rho_2)}n_{\xi^1}^\sigma n_{\eta^1}^\tau +\b\sqrt{N_1N_3(1+\rho_1)(1+\rho_3)}n_{\xi^1}^\sigma n_{\chi^1}^\phi+
    \\
     +&\c\sqrt{N_2N_3(1+\rho_2)(1+\rho_3) }n_{\chi^1}^\phi n_{\eta^1}^\tau\biggr),
    \end{aligned}
\end{equation}

\begin{equation}\footnotesize
    \begin{aligned}
      b(t) \doteq (1-t)\biggl[ & \sqrt{N_1N_2}(\psi^{(1)}_m n_{\xi^1}^\sigma + \psi^{(1)}_n n_{\eta^1}^\tau) + \sqrt{N_1N_3}(\psi^{(2)}_m n_{\xi^1}^\sigma + \psi^{(1)}_l n_{\chi^1}^\phi)+
    \\
         +& \sqrt{N_2N_3}(\psi^{(2)}_n n_{\eta^1}^\tau + \psi^{(2)}_l n_{\chi^1}^\phi)\biggr]  ,
    \end{aligned}
\end{equation}

\begin{equation}\footnotesize
    \begin{aligned}
J(\bm{\sigma},\bm{\tau},\bm{\phi}) \doteq & J_1 \sum_{i=1}^{N_1} \xi_i^1 \sigma_i + J_2 \sum_{i=1}^{N_2} \eta_i^1 \tau_i + J_3 \sum_{i=1}^{N_3} \chi_i^1 \phi_i,
\end{aligned}
\end{equation}

\begin{equation}\footnotesize
    \begin{aligned}
        c(t) \doteq &\sqrt{t} \left[ \sqrt{\dfrac{\beta}{2N_1}}\SOMMA{\mu>1}{K}\SOMMA{i=1}{N_1}J_i^\mu\sigma_i z_\mu+\sqrt{\dfrac{\beta}{2N_2}}\SOMMA{\mu>1}{K}\SOMMA{j=1}{N_2}Y_j^\mu\tau_j \s_\mu+ \right. \\
    +& \left. \sqrt{\dfrac{\beta}{2N_3}}\SOMMA{\mu>1}{K}\SOMMA{k=1}{N_3}V_k^\mu\phi_k \w_\mu \right] ,
    \end{aligned}
\end{equation}

\begin{equation}\footnotesize
    \begin{aligned}
    d(t) \doteq & \sqrt{1-t}\left(  A_{\sigma}\SOMMA{i=1}{N_1}X^{\sigma}_i\sigma_i + A_{\tau}\SOMMA{j=1}{N_2}X^{\tau}_j\tau_j + A_{\phi}\SOMMA{k=1}{N_3}X^{\phi}_k\phi_k\right)+
         \\ 
         + & \sqrt{1-t}\left(B_z\SOMMA{\mu>1}{K}D^{\mu}_z z_{\mu} + \Bar{B}_z\SOMMA{\mu>1}{K}\Bar{D}^{\mu}_z \z_{\mu} + B_s\SOMMA{\mu>1}{K}D^{\mu}_s s_{\mu}+ \right.
         \\
         &\left.\qquad\;\;+
         \Bar{B}_s\SOMMA{\mu>1}{K}\Bar{D}^{\mu}_s \s_{\mu} + B_w\SOMMA{\mu>1}{K}D^{\mu}_w w_{\mu} + \Bar{B}_w\SOMMA{\mu>1}{K}\Bar{D}^{\mu}_w \w_{\mu}\right),
    \end{aligned}
\end{equation}

\begin{equation}\footnotesize
    \begin{aligned}
        e(t) &\doteq \dfrac{(1-t)}{2}\left[\Bar{C}_z\SOMMA{\mu=1}{K}(\z_{\mu})^2 + C_z\SOMMA{\mu=1}{K}(z_{\mu})^2 + \Bar{C}_s\SOMMA{\mu=1}{K}(\s_{\mu})^2 + \right.
         \\
        & \left. + C_s\SOMMA{\mu=1}{K}(s_{\mu})^2 + \Bar{C}_w\SOMMA{\mu=1}{K}(\w_{\mu})^2 + C_w\SOMMA{\mu=1}{K}(w_{\mu})^2\right].
    \end{aligned}
\end{equation}

Note that, for simplicity, we drop all dependencies (except $\bm{N}$, to emphasize that we are working in the finite limit). Furthermore, $\mathcal Z_{\bm{N}}(t=1)$ does coincide with the partition function of the original TAM network, while $\mathcal Z_{\bm{N}}(t=0)$ is certainly analytically valuable as it is a linear combination of one-body models, whose underlying probabilistic measure is trivially factorized by definition.
\newline
In the following, if not otherwise specified, we refer to the generalized averages simply as $\langle \cdot \rangle$ in order to lighten the notation.\\
By glancing at \eqref{eq:pressioneInterp} we see that we have to evaluate the Cauchy condition  $\mathcal A_{\bm{N}}(t=0)$ (that is straightforward as in $t=0$ the Gibbs measure is factorized over the neural activities) and integrate the t-streaming, that is the derivative  $\partial_t\mathcal A_{\bm{N}}(t)$, over the interval $t \in (0,1)$. 
\\

The explicit calculation of the t-derivative of the interpolating free energy introduced in Eq.~\ref{eq:pressioneInterp}
is carried out by brute force; thus, we directly report our findings:

{
\begin{equation}\footnotesize
    \begin{array}{lllll}
    \partial_t \mathcal{A}_{\bm{N}}(t) &=\dfrac{1}{L}\Bigg[\beta\Big(\a\sqrt{N_1N_2(1+\rho_1)(1+\rho_2)}\l n_{\xi^1}^\sigma n_{\eta^1}^\tau\r+\b\sqrt{N_1N_3(1+\rho_1)(1+\rho_3)}\l n_{\xi^1}^\sigma n_{\chi^1}^\phi\r
    \\\\ \noalign{\vspace{-10pt}}
    &+\c\sqrt{N_2N_3(1+\rho_2)(1+\rho_3) }\l n_{\chi^1}^\phi n_{\eta^1}^\tau\r\Big)-\Big(\sqrt{N_1N_2}(\psi^{(1)}_m \l n_{\xi^1}^\sigma\r + \psi^{(1)}_n \l n_{\eta^1}^\tau\r)
    \\\\ \noalign{\vspace{-10pt}}
     &+ \sqrt{N_1N_3}(\psi^{(2)}_m \l n_{\xi^1}^\sigma\r + \psi^{(1)}_l \l n_{\chi^1}^\phi\r ) + \sqrt{N_2N_3}(\psi^{(2)}_n \l n_{\eta^1}^\tau\r + \psi^{(2)}_l \l n_{\chi^1}^\phi\r)\Big)
    \\\\ \noalign{\vspace{-10pt}}
        & +\dfrac{\beta}{4}K(\l p^z_{11}\r-\l \q^\sigma p^z_{12}\r)+\dfrac{\beta}{4}K(\l \ps_{11}\r-\l\q^\tau \ps_{12}\r)+\dfrac{\beta}{4}K(\l\pw_{11}\r-\l\q^\phi \pw_{12}\r)
    \\\\ \noalign{\vspace{-5pt}}
        & -\dfrac{K}{2} \left(\bar C_z \l\pz_{11}\r+C_z \l p^z_{11}\r+ \bar C_s \l\ps_{11}\r+ C_s \l p^s_{11}\r+ \bar C_w \l\pw_{11}\r+ C_w \l p_{11}^w\r\right)
    \\\\ \noalign{\vspace{-10pt}}
        & -\dfrac{K}{2}\Bigg( B_z^2 (\l p_{11}^z\r-\l p_{12}^z\r)+\bar B_z^2 (\l\pz_{11}\r-\l\pz_{12}\r)+ B_s^2 (\l p_{11}^s\r-\l p_{12}^s\r)
    \\\\ \noalign{\vspace{-10pt}}
        & +\bar B_s^2 (\l\ps_{11}\r-\l\ps_{12}\r)+ B_w^2 (\l p_{11}^w\r-\l p_{12}^w\r)+\bar B_w^2 (\l\pw_{11}\r-\l\pw_{12}\r)\Bigg)
  \\\\ \noalign{\vspace{-10pt}}
       & -\dfrac{1}{2}\left( N_1 A_\sigma^2(1-\l q^\sigma\r)+ N_2 A_\tau^2(1-\l q^\tau\r)+N_3 A_\phi^2(1-\l q^\phi\r)\right)\Bigg].
    \end{array}
\end{equation}
}
We see that, if we fix the values of the constants and the auxiliary field as follows
\begin{equation}
    \begin{array}{lllll}
         \bar C_z =-\bar B_z=0, &\quad C_s =-B_s=0 ,
         \\\\ \noalign{\vspace{-10pt}}
         C_w = -B_w=0

         , &\quad C_z = \dfrac{\beta}{2}(1-\qsigma) ,
         \\\\ \noalign{\vspace{-10pt}}
         \bar C_s = \dfrac{\beta}{2}(1-\qtau) , &\quad  \bar C_w = \dfrac{\beta}{2}(1-\qphi) ,
         \\\\ \noalign{\vspace{-10pt}}
    \end{array}
\end{equation}

\begin{equation}
    \begin{array}{lllll}
    B_z^2=\dfrac{\beta}{2}\qsigma, \quad & \bar B_s^2=\dfrac{\beta}{2}\qtau, \quad 
          
         \bar B_w^2=\dfrac{\beta}{2}\qphi,
    \end{array}
\end{equation}
         
\begin{equation}
    \begin{array}{lllll}
    A_\sigma^2=\dfrac{\beta}{2}\dfrac{K}{N_1}\bar p^z, \quad
    A_\tau^2=\dfrac{\beta}{2}\dfrac{K}{N_2}\pdags,     \quad A_\phi^2=\dfrac{\beta}{2}\dfrac{K}{N_3}\pdagw,
    \end{array} 
\end{equation}

\begin{equation}
    \begin{array}{lllll}  
         \psi_m^{(1)}=\a\beta \sqrt{(1+\rho_1)(1+\rho_2)}\bar{n}^\tau_{\eta^1} , &\quad  \psi_n^{(1)}=\a\beta \sqrt{(1+\rho_1)(1+\rho_2)}\bar{n}^\sigma_{\xi^1},
    \end{array}
\end{equation}
\begin{equation}
    \begin{array}{lllll}  
         \psi_l^{(1)}=\b\beta \sqrt{(1+\rho_1)(1+\rho_3)}\bar{n}^\sigma_{\xi^1} , &\quad 
         \psi_m^{(2)}=\b\beta \sqrt{(1+\rho_1)(1+\rho_3)}\bar{n}^\phi_{\chi^1} ,
    \end{array}
\end{equation}
\begin{equation}
    \begin{array}{lllll}  
         \psi_n^{(2)}=\c\beta\sqrt{(1+\rho_3)(1+\rho_2)} \bar{n}^\phi_{\chi^1} , &\quad  \psi_l^{(2)}=\c\beta\sqrt{(1+\rho_2)(1+\rho_3)} \bar{n}^\tau_{\eta^1},
    \end{array}
\end{equation}
and using the replica-symmetric ansatz, we get the final expression in the thermodynamic limit of the derivative of the interpolating free energy that reads as

\begin{equation}
\label{eq:straming_computed}
    \begin{array}{lllll}
    \partial_t \mathcal{A}(t)&=&-\beta\varsigma\left[\a\theta^{-1}\sqrt{(1 + \rho_1)(1 + \rho_2)}\bar{n}^\sigma_{\xi^1}\bar{n}^\tau_{\eta^1}+\b\alpha^{-1}\sqrt{(1 + \rho_1)(1 + \rho_3)}\bar{n}^\sigma_{\xi^1}\bar{n}^\phi_{\chi^1}\right.
    \\\\ \noalign{\vspace{-10pt}}
    &&+\left.\c(\alpha\theta)^{-1}\sqrt{(1 + \rho_2)(1 + \rho_3)}\bar{n}^\tau_{\eta^1}\bar{n}^\phi_{\chi^1}\right]
    \\\\ \noalign{\vspace{-10pt}}
    &&-\beta\dfrac{\varsigma}{4}\left[ \bar p^z(1-\qsigma)+ \pdags(1-\qtau)+\pdagw(1-\qphi)\right]
    \end{array}
\end{equation}
where we use $\varsigma = \dfrac{1}{3}(\alpha+\theta+\alpha\theta)$\,.

Now, in order to fulfill the sum rule stemmed by the fundamental theorem of calculus for the interpolating free energy —see Eq.~\ref{eq:pressioneInterp}— we have to evaluate the Cauchy condition $\mathcal A(t=0)$: by construction, as we are interpolating the original model with a linear combination of one-body contributions, such condition can be explicitly evaluated, again by brute force, and gives rise to
\begin{equation}
    \mathcal A(t=0) = \lim_{\bm N \to \infty} \frac{1}{L}\ln \mathbb{E}\mathcal{Z}_{\boldsymbol{N}}(t=0),
\end{equation}
    
\noindent
where the interpolating partition function, evaluated at $t=0$, reads as
\begin{equation}
    \begin{array}{lllll}
         \mathcal{Z}_{\boldsymbol{N}}(t=0)=\prod\limits_{i=1}^{N_1}2\cosh\left[\hat\xi_i^\mu\psi_m^{(1)}\sqrt{\dfrac{N_2}{N_1}}+\hat\xi_i^\mu\psi_m^{(2)}\sqrt{\dfrac{N_3}{N_1}}+A_\sigma X_i^\sigma \right]
         \\\\ \noalign{\vspace{-10pt}}
         \times\prod\limits_{k=1}^{N_2} 2\cosh\left[\hat\eta_k^\mu\psi_n^{(1)}\sqrt{\dfrac{N_1}{N_2}}+\hat\eta_k^\mu\psi_n^{(2)}\sqrt{\dfrac{N_3}{N_2}}+A_\tau X_k^\tau \right]\times
         \\\\ \noalign{\vspace{-10pt}}
         \times\prod\limits_{j=1}^{N_3}2\cosh\left[\hat\chi_j^\mu\psi_l^{(1)}\sqrt{\dfrac{N_1}{N_3}}+\hat\chi_j^\mu\psi_l^{(2)}\sqrt{\dfrac{N_2}{N_3}} +A_\phi X_j^\phi \right]\times
         \\\\ \noalign{\vspace{-10pt}}
         \times\prod\limits_{\mu>1}^{K}\displaystyle\int \dfrac{\beta^6\bm X_\mu }{(\sqrt{2\pi})^6}\exp\Bigg[-\dfrac{1}{2} \bm X_\mu^T \bm C \bm X_\mu + \bm X_\mu^T \bm B_\mu\Bigg],
    \end{array}
\end{equation}
where we have set
\begin{equation}
\begin{array}{llll}
     \bm X_\mu= \begin{pmatrix}
        Re(z_\mu)
        \\
       Im(z_\mu)
        \\ 
        Re(s_\mu)
        \\
        Im(s_\mu)
        \\ 
        Re(w_\mu)
        \\
        Im(w_\mu)
    \end{pmatrix}\,,\;&& \bm B_\mu= \begin{pmatrix}
        B_z J_\mu^z + \bar B_z \bar J_\mu^z
        \\
        i(B_z J_\mu^z - \bar B_z \bar J_\mu^z)
        \\ 
         B_s J_\mu^s + \bar B_s \bar J_\mu^s
        \\
        i(B_s J_\mu^s - \bar B_s \bar J_\mu^s)
        \\ 
         B_w J_\mu^w + \bar B_w \bar J_\mu^w
        \\
        i(B_w J_\mu^w - \bar B_w \bar J_\mu^w)
    \end{pmatrix},
\end{array}
\end{equation}
{
\begin{equation}\small \label{eq:A14}
    \bm C= \begin{pmatrix}
        1-\bar C_z -C_z & i(\bar C_z -C_z) & -\a/2&-i\a/2&-\b/2&-i\b/2
        \\
        i(\bar C_z -C_z) &1+\bar C_z +C_z & i\a/2&-\a/2&ib/2&-\b/2
        \\ 
        -\a/2&ia/2&1-\bar C_s -C_s & i(\bar C_z -C_z)&-\c/2&-i\c/2
        \\
        -i\a/2&-\a/2&i(\bar C_s -C_s) &1+\bar C_z +C_z&-i\c/2&\c/2
        \\ 
        -\b/2&ib/2&-\c/2&-i\c/2&1-\bar C_w -C_w & i(\bar C_w -C_w)
        \\
        -i\b/2&-\b/2&-i\c/2&\c/2&i(\bar C_w -C_w) &1+\bar C_w +C_w
    \end{pmatrix}.
\end{equation}
}
Finally, computing the Gaussian integrals, we get
\begin{equation}
    \begin{array}{lllll}
          \mathcal{Z}_{\boldsymbol{N}}(t=0)&&=\prod\limits_{i=1}^{N_1}2\cosh\left[\hat\xi_i^\mu\psi_m^{(1)}\sqrt{\dfrac{N_2}{N_1}}+\hat\xi_i^\mu\psi_m^{(2)}\sqrt{\dfrac{N_3}{N_1}}+A_\sigma X_i^\sigma \right]
         \\\\ \noalign{\vspace{-10pt}}
         &&\times\prod\limits_{k=1}^{N_2} 2\cosh\left[\hat\eta_k^\mu\psi_n^{(1)}\sqrt{\dfrac{N_1}{N_2}}+\hat\eta_k^\mu\psi_n^{(2)}\sqrt{\dfrac{N_3}{N_2}}+A_\tau X_k^\tau \right]\times
         \\\\ \noalign{\vspace{-10pt}}
         &&\times\prod\limits_{j=1}^{N_3}2\cosh\left[\hat\chi_j^\mu\psi_l^{(1)}\sqrt{\dfrac{N_1}{N_3}}+\hat\chi_j^\mu\psi_l^{(2)}\sqrt{\dfrac{N_2}{N_3}} +A_\phi X_j^\phi \right]\times
         \\\\ \noalign{\vspace{-10pt}}
         &&\times\left(\mathrm{det}\bm C\right)^{-K/2}\prod\limits_{\mu>1}^{K}\exp\Bigg[\dfrac{1}{2} \bm B_\mu^T \bm C^{-1} \bm B_\mu \Bigg],
    \end{array}
\end{equation}
thus the one-body contribution to the interpolating free energy, in the thermodynamic limit, is
\begin{equation}
\label{eq:one_body_computed}
    \begin{array}{lllll}
         \mathcal{A}(t=0)&=&\varsigma\,\mathbb E\left\{\log 2\cosh\left[\hat\xi^\mu\left(\psi_m^{(1)}\dfrac{1}{\theta}+\psi_m^{(2)}\dfrac{1}{\alpha}\right)+A_\sigma X^\sigma \right]\right\}
         \\\\ \noalign{\vspace{-10pt}}
         &&+\dfrac{\varsigma}{\theta^2}\mathbb E\left\{\log 2\cosh\left[\hat\eta^\mu\theta\left(\psi_n^{(1)}+\psi_n^{(2)}\dfrac{1}{\alpha}\right)+A_\tau X^\tau \right]\right\}
         \\\\ \noalign{\vspace{-10pt}}
         &&+\dfrac{\varsigma}{\alpha^2}\mathbb{E}\left\{\log 2\cosh\left[\hat\chi^\mu\alpha\left(\psi_l^{(1)}+\psi_l^{(2)}\dfrac{1}{\theta}\right) +A_\phi X^\phi \right]\right\}
         \\\\ \noalign{\vspace{-10pt}}
         &&-\dfrac{\gamma\varsigma}{2}\log\left[\mathrm{det}\bm C\right]+\dfrac{\gamma\varsigma}{2}\mathbb E\Bigg[ \bm B_\mu^T \bm C^{-1} \bm B_\mu \Bigg].
    \end{array}
\end{equation}

\vspace{3mm}
\noindent
Merging these two expressions, that is the integral of the $t$-streaming (Eq.~\eqref{eq:straming_computed}) and the Cauchy condition (Eq.~\eqref{eq:one_body_computed}), in Eq.\eqref{eq:pressioneInterp} under the replica symmetric assumption,  we get the explicit expression of the quenched free energy:
{
\begin{equation}\footnotesize\label{eq:fRS}
    \begin{array}{lllll}
&\mathcal{A}_{\alpha,\theta,\gamma}^{(\textit{sup})}(\beta)=
         \varsigma\,\mathbb E\left\{\log 2\cosh\left[\beta\hat\xi^1\left(\dfrac{\a}{\theta}\sqrt{(1 + \rho_1)(1 + \rho_2)} \bar{n}^\tau_{\eta^1}+  
          \dfrac{\b}{\alpha}\sqrt{(1+\rho_1)(1 + \rho_3)}\right)\bar{n}^\phi_{\chi^1}\right.\right.
          \\\\ \noalign{\vspace{-10pt}}
          
          &\left.\left.+X^\sigma \sqrt{\dfrac{\beta}{2}\gamma\bar p^z}\right]\right\}
         \\\\ \noalign{\vspace{-10pt}}
         &+\theta^{-2}\varsigma\,\mathbb E\left\{\log 2\cosh\left[\beta\hat\eta^1\theta\left( \a\sqrt{(1 + \rho_1)(1 + \rho_2)}\bar{n}^\sigma_{\xi^1}+ \dfrac{\c}{\alpha}\sqrt{(1 + \rho_2)(1 + \rho_3)}\bar{n}^\phi_{\chi^1}\right)\right.\right. 
         \\\\ \noalign{\vspace{-10pt}}
         &\left.\left.+X^\tau \sqrt{\dfrac{\beta}{2}\theta^2\gamma\pdags}\right]\right\}
         \\\\ \noalign{\vspace{-10pt}}
         &+\alpha^{-2}\varsigma\,\mathbb{E}\left\{\log 2\cosh\left[\beta\hat\chi^1\alpha\left( \b\sqrt{(1 + \rho_1)(1 + \rho_3)}\bar{n}^\sigma_{\xi^1}+ \dfrac{\c}{\theta}\sqrt{(1 + \rho_2)(1 + \rho_3)}\bar{n}^\tau_{\eta^1}\right)\right.\right. 
         \\\\ \noalign{\vspace{-10pt}}
         &\left.\left.+ X^\phi\sqrt{\dfrac{\beta}{2}\alpha^2\gamma\pdagw} \right]\right\} 
         \\\\ \noalign{\vspace{-10pt}}
         &-\dfrac{\gamma\varsigma}{2}\log\left[\mathrm{det}\bm C\right]
    \\\\ \noalign{\vspace{-10pt}}
    & +\gamma\varsigma\beta^3\a\b\c\dfrac{  \qsigma(1 - \qtau)(1 - \qphi)     +(1 - \qsigma) (1 - \qtau)\qphi    + (1 - \qsigma) \qtau   (1 - \qphi) }{\mathrm{det} \bm C}
    \\\\ \noalign{\vspace{-10pt}}
    &
    +\dfrac{\gamma\varsigma}{2}\dfrac{\beta^2 \a^2  (1 - \qsigma) \qtau + \beta^2 \a^2 \qsigma(1 - \qtau)  + 
 \beta^2 \c^2 \qtau(1 - \qphi )  }{\mathrm{det} \bm C}
 \\\\ \noalign{\vspace{-10pt}}
 & +\dfrac{\gamma\varsigma}{2}\dfrac{\beta^2 \c^2 (1 - \qtau) \qphi + 
 \beta^2 \b^2 (1 - \qphi) \qsigma + \beta^2 \b^2(1 - \qsigma) \qphi}{\mathrm{det} \bm C}
    \\\\ \noalign{\vspace{-10pt}}
        & -\beta\varsigma\left[\a\theta^{-1}\sqrt{(1 + \rho_1)(1 + \rho_2)}\bar{n}^\sigma_{\xi^1}\bar{n}^\tau_{\eta^1}+\b\alpha^{-1}\sqrt{(1 + \rho_1)(1 + \rho_3)}\bar{n}^\sigma_{\xi^1}\bar{n}^\phi_{\chi^1}\right.
        \\\\ \noalign{\vspace{-10pt}}
        &\left.+\c(\alpha\theta)^{-1}\sqrt{(1 + \rho_2)(1 + \rho_3)}\bar{n}^\tau_{\eta^1}\bar{n}^\phi_{\chi^1}\right]
    \\\\ \noalign{\vspace{-10pt}}
        &-\dfrac{\gamma\varsigma}{4}\beta\left[ \bar p^z(1-\qsigma)+ \pdags(1-\qtau)+\pdagw(1-\qphi)\right]
    \end{array}
\end{equation}}
where:
{\begin{equation} 
\footnotesize
\begin{aligned} 
  \det\mathbf{C} &= 1 - \big[\beta^{2} a^{2} (1-\bar{q}^{\tau}) (1-\bar{q}^{\sigma})+\beta^{2} b^{2} (1-\bar{q}^{\phi}) (1-\bar{q}^{\sigma}) +\beta^{2} c^{2} (1-\bar{q} ^{\phi}) (1-\bar{q}^{\tau}) \big] \\
  &\quad - 2 a b c \beta^{3} (1-\bar{q}^{\phi}) (1-\bar{q}^{ \sigma}) (1-\bar{q}^{\tau}).
\end{aligned}
\end{equation}}
and 
{\begin{equation} \small
    \varsigma = \dfrac{1}{3}(\alpha + \theta + \alpha\theta).
\end{equation}}

\noindent
The values of the order parameters at fixed control parameters can be obtained by evaluating
the saddle points of \eqref{eq:fRS}. The set of self-consistent equations to be fulfilled by
these values is shown below:

{\footnotesize\begin{equation} 
\label{eq:self_eq_sup_low_M}
\bar{n}^\sigma_{\xi^1} =\mathbb{E}\left\{\hat{\xi}^1\tanh\left[\beta\Hat{\xi}^1\biggl(\dfrac{\a}{\theta}\sqrt{(1+\rho_1)(1+\rho_2)}\bar{n}^\tau_{\eta^1} + \dfrac{\b}{\alpha}\sqrt{(1+\rho_1)(1+\rho_3)}\bar{n}^\phi_{\chi^1}\biggr)+X^{\sigma}\sqrt{\dfrac{\beta\gamma \bar{p}^z}{2}}\right]\right\},
\end{equation}}

{\footnotesize \begin{equation}   
\bar{n}^\tau_{\eta^1} = \mathbb{E}\left\{\hat{\eta}^1\tanh\left[\beta\Hat{\eta}^1\theta\biggl(\a\sqrt{(1+\rho_1)(1+\rho_2)}\bar{n}^\sigma_{\xi^1}+ \dfrac{\c}{\alpha}\sqrt{(1+\rho_2)(1+\rho_3)}\bar{n}^\phi_{\chi^1}\biggr)+X^{\tau}\theta\sqrt{\dfrac{\beta\gamma \bar{p}^{\s}}{2}}\right]\right\},
\end{equation}}

{\footnotesize\begin{equation}
\bar{n}^\phi_{\chi^1} = \mathbb{E}\left\{\hat{\chi}^1\tanh\left[\beta\Hat{\chi}^1\alpha\biggl(\b\sqrt{(1+\rho_1)(1+\rho_3)}  \bar{n}^\sigma_{\xi^1}+ \dfrac{\c}{\theta}\sqrt{(1+\rho_2)(1+\rho_3)}\bar{n}^\tau_{\eta^1}\biggr) +X^{\phi}\alpha\sqrt{\dfrac{\beta\gamma \bar{p}^{\w}}{2}}\right]\right\},
\end{equation}}

and the overlaps
{
\begin{equation}\footnotesize
\label{eq:overlap_sup_low_Q}
\begin{array}{lll}
        \qsigma = \mathbb{E}\left\{\tanh^2\left[\beta\Hat{\xi}^1\biggl(\dfrac{\a}{\theta}\sqrt{(1+\rho_1)(1+\rho_2)}\bar{n}^\tau_{\eta^1} + \dfrac{\b}{\alpha}\sqrt{(1+\rho_1)(1+\rho_3)}\bar{n}^\phi_{\chi^1}\biggr) + X^{\sigma}\sqrt{\dfrac{\beta\gamma \bar{p}^z}{2}}\right]\right\},
    \\\\ \noalign{\vspace{-10pt}}
        \qtau = \mathbb{E}\left\{\tanh^2\left[\beta\Hat{\eta}^1\theta\biggl(\a\sqrt{(1+\rho_1)(1+\rho_2)}\bar{n}^\sigma_{\xi^1}+ \dfrac{\c}{\alpha}\sqrt{(1+\rho_2)(1+\rho_3)}\bar{n}^\phi_{\chi^1}\biggr) + X^{\tau}\theta\sqrt{\dfrac{\beta\gamma \bar{p}^{\s}}{2}}\right]\right\},
    \\\\ \noalign{\vspace{-10pt}}
        \qphi = \mathbb{E}\left\{\tanh^2\left[\beta\Hat{\chi}^1\alpha\biggl(\b\sqrt{(1+\rho_1)(1+\rho_3)}  \bar{n}^\sigma_{\xi^1}+ \dfrac{\c}{\theta}\sqrt{(1+\rho_2)(1+\rho_3)}\bar{n}^\tau_{\eta^1}\biggr) + X^{\phi}\alpha\sqrt{\dfrac{\beta\gamma \bar{p}^{\w}}{2}}\right]\right\},
\end{array}
\end{equation}
}
where
{\begin{equation}\small
\label{eq:p_equations}
\begin{array}{lll}
\dfrac{\beta}{2}\gamma \bar{p}^z(\a,\b,\c,\beta,\qsigma,\qtau,\qphi)&= 2\dfrac{\mathcal{D}(\qsigma,\qtau,\qphi) }{[\det\mathbf{C}]^2}\partial_{\qsigma}(\det\mathbf{C})+2\dfrac{\gamma}{2}\dfrac{\partial_{\qsigma}(\det\mathbf{C})}{\det\mathbf{C}}
\\\\ \noalign{\vspace{-10pt}}
&-2\dfrac{\partial_{\qsigma}(\mathcal{D}(\qsigma,\qtau,\qphi)) }{\det\mathbf{C}},
\end{array}
\end{equation}}
and 
{\begin{equation}\small
\begin{array}{lll}
     \mathcal{D}(\qsigma,\qtau,\qphi) &=& \gamma\beta^3\a\b\c\left[  \qsigma(1 - \qtau)(1 - \qphi)     +(1 - \qsigma) (1 - \qtau)\qphi \right.  
     \\\\ \noalign{\vspace{-10pt}}
     &&\left.+ (1 - \qsigma) \qtau   (1 - \qphi) \right]
    \\\\ \noalign{\vspace{-10pt}}
    &&
    +\dfrac{\gamma}{2}\left[\beta^2 \a^2  (1 - \qsigma) \qtau + \beta^2 \a^2 \qsigma(1 - \qtau)  + 
 \beta^2 \c^2 \qtau(1 - \qphi )  \right]
 \\\\ \noalign{\vspace{-10pt}}
 && +\dfrac{\gamma}{2}\left[\beta^2 \c^2 (1 - \qtau) \qphi + 
 \beta^2 \b^2 (1 - \qphi) \qsigma + \beta^2 \b^2(1 - \qsigma) \qphi\right].
\end{array}
\end{equation}}
\vspace{3mm}
We stress that by a cyclically switch of the variables we obtain the self for $\bar{p}^{s^\dag}$ and $\bar{p}^{w^\dag}$, namely
\begin{equation}
\begin{array}{lll}
     &\bar{p}^{s^\dag}(\a,\b,\c,\beta,\qtau,\qsigma,\qphi)= \bar{p}^z(\a,\b,\c,\beta,\qsigma,\qtau,\qphi),\\\\

     &\bar{p}^{w^\dag}(\a,\b,\c,\beta,\qphi,\qtau,\qsigma)= \bar{p}^z(\a,\b,\c,\beta,\qsigma,\qtau,\qphi).
\end{array}
\end{equation}

\vspace{3mm}
\noindent   
Then exploiting the generator function, we get the self consistency equations for the three archetypes-Mattis magnetization
{
\begin{equation}\hspace{-0.5cm}\footnotesize
\label{mattisSelf}
\begin{array}{lllll}
    \bar{m}^\sigma_{\xi^1} = \mathbb{E}\left\{\xi^1\tanh\left[\beta\Hat{\xi}^1\biggl(\dfrac{\a}{\theta}\sqrt{(1+\rho_1)(1+\rho_2)}\bar{n}^\tau_{\eta^1} + \dfrac{\b}{\alpha}\sqrt{(1+\rho_1)(1+\rho_3)}\bar{n}^\phi_{\chi^1}\biggr) +X^{\sigma}\sqrt{\dfrac{\beta\gamma \bar{p}^z}{2}}\right]\right\},
    \\\\ \noalign{\vspace{-10pt}}
    \bar{m}^\tau_{\eta^1} = \mathbb{E}\left\{\eta^1\tanh\left[\beta\Hat{\eta}^1\theta\biggl(\a\sqrt{(1+\rho_1)(1+\rho_2)}\bar{n}^\sigma_{\xi^1}+ \dfrac{\c}{\alpha}\sqrt{(1+\rho_2)(1+\rho_3)}\bar{n}^\phi_{\chi^1}\biggr)+ X^{\tau}\theta\sqrt{\dfrac{\beta\gamma \bar{p}^{\s}}{2}}\right]\right\},
    \\\\ \noalign{\vspace{-10pt}}
    \bar{m}^\phi_{\chi^1} = \mathbb{E}\left\{\chi^1\tanh\left[\beta\Hat{\chi}^1\alpha\biggl(\b\sqrt{(1+\rho_1)(1+\rho_3)}  \bar{n}^\sigma_{\xi^1}+ \dfrac{\c}{\theta}\sqrt{(1+\rho_2)(1+\rho_3)}\bar{n}^\tau_{\eta^1}\biggr)+ X^{\phi}\alpha\sqrt{\dfrac{\beta\gamma \bar{p}^{\w}}{2}}\right]\right\}.
\end{array}
\end{equation}
}

\subsubsection*{Proof for unsupervised} \label{appsubsec:proofGuerraUnsup}
In this appendix, using Guerra's interpolation scheme, we are going to get an explicit expression for the quenched free energy, in the Unsupervised case.

\noindent
Guerra's interpolating partition function, in the unsupervised setting reads:

\begin{equation}\label{Integralista}\footnotesize
\begin{aligned}
\mathcal{Z}_{\bm N, M, \bm r}(\beta, t, \bm J) & \doteq \sum_{\{\sigma\},\{\tau\},\{\phi\}}^{2^{N_1},2^{N_2},2^{N_3}} \exp\Bigg\{ a(t) + b(t) + c(t) + d(t) + J(\bm{\sigma},\bm{\tau},\bm{\phi}) \Bigg\},
\end{aligned}    
\end{equation}

\noindent
where we have defined
\begin{equation}\footnotesize
    \begin{aligned}
a(t) \doteq t \, \beta \left[ \vphantom{\sum_{a=1}^{M}} \right. &
    \a \sqrt{N_1 N_2} \sqrt{1+\rho_1} \sqrt{1+\rho_2} \cdot \frac{1}{M} \sum_{a=1}^{M} n_{\Xi^{1,a}}^\sigma n_{\Theta^{1,a}}^\tau \\
+ & \b \sqrt{N_1 N_3} \sqrt{1+\rho_1} \sqrt{1+\rho_3} \cdot \frac{1}{M} \sum_{a=1}^{M} n_{\Xi^{1,a}}^\sigma n_{\Upsilon^{\mu,a}}^\phi \\
+ & \c \sqrt{N_2 N_3} \sqrt{1+\rho_2} \sqrt{1+\rho_3} \cdot \frac{1}{M} \sum_{a=1}^{M} n_{\Theta^{1,a}}^\tau n_{\Upsilon^{\mu,a}}^\phi
\left. \vphantom{\sum_{a=1}^{M}} \right],
\end{aligned}
\end{equation}

\begin{equation}\footnotesize
    \begin{aligned}
b(t) \doteq (1 - t) \left[ \vphantom{\sum_{a=1}^{M}} \right. &
    \sqrt{N_1 N_2} \left( \psi_\sigma^{(1)} \sum_{a=1}^{M} n_{\Xi^{1,a}}^\sigma + \psi_\tau^{(1)} \sum_{a=1}^{M} n_{\Theta^{1,a}}^\tau \right) \\
+ & \sqrt{N_1 N_3} \left( \psi_\sigma^{(2)} \sum_{a=1}^{M} n_{\Xi^{1,a}}^\sigma + \psi_\phi^{(1)} \sum_{a=1}^{M} n_{\Upsilon^{\mu,a}}^\phi \right) \\
+ & \sqrt{N_2 N_3} \left( \psi_\tau^{(2)} \sum_{a=1}^{M} n_{\Theta^{1,a}}^\tau + \psi_\phi^{(2)} \sum_{a=1}^{M} n_{\Upsilon^{\mu,a}}^\phi \right)
\left. \vphantom{\sum_{a=1}^{M}} \right],
\end{aligned}
\end{equation}

\begin{equation}\footnotesize
    \begin{aligned}
c(t) \doteq \sqrt{t} \, \beta \left[ \vphantom{\sum_{\mu>1}} \right. & 
    \a \sqrt{\left( \frac{1+\rho_{12}}{N_1 N_2 (1+\rho_1)(1+\rho_2)} \right)} 
    \sum_{\mu>1} \sum_{i,j=1}^{N_1,N_2} J_{ij}^\mu \sigma_i \tau_j \\
+ & \b \sqrt{\left( \frac{1+\rho_{13}}{N_1 N_3 (1+\rho_1)(1+\rho_3)} \right)} 
    \sum_{\mu>1} \sum_{i,k=1}^{N_1,N_3} Y_{ik}^\mu \sigma_i \phi_k \\
+ & \c \sqrt{\left( \frac{1+\rho_{23}}{N_2 N_3 (1+\rho_2)(1+\rho_3)} \right)} 
    \sum_{\mu>1} \sum_{j,k=1}^{N_2,N_3} V_{jk}^\mu \tau_j \phi_k 
\left. \vphantom{\sum_{\mu>1}} \right],
\end{aligned}
\end{equation}

\begin{equation}\footnotesize
    \begin{aligned}
d(t) \doteq \sqrt{1 - t}  \left[ 
    \vphantom{\sum_i} 
    \right. & 
    A_\sigma^\tau \sum_i \varphi_i^{\sigma,\tau} \sigma_i 
    + A_\sigma^\phi \sum_i \varphi_i^{\sigma,\phi} \sigma_i 
    + A_\tau^\sigma \sum_j \varphi_j^{\tau,\sigma} \tau_j \\
+ & A_\tau^\phi \sum_j \varphi_j^{\tau,\phi} \tau_j 
    + A_\phi^\sigma \sum_k \varphi_k^{\phi,\sigma} \phi_k 
    + A_\phi^\tau \sum_k \varphi_k^{\phi,\tau} \phi_k 
\left. \vphantom{\sum_i} \right],
\end{aligned}
\end{equation}

\begin{equation}\footnotesize
    \begin{aligned}
J(\bm{\sigma},\bm{\tau},\bm{\phi}) \doteq & J_1 \sum_{i=1}^{N_1} \xi_i^1 \sigma_i + J_2 \sum_{i=1}^{N_2} \eta_i^1 \tau_i + J_3 \sum_{i=1}^{N_3} \chi_i^1 \phi_i,
\end{aligned}
\end{equation}

\noindent
and $\varphi_i^{\sigma,\tau},\varphi_i^{\sigma,\phi},\varphi_j^{\tau,\sigma},\varphi_j^{\tau,\phi},\varphi_k^{\phi,\sigma},\varphi_k^{\phi,\tau}  \sim\mathcal{N}(0,1)$.

\vspace{3mm}
\noindent
Following the same procedure as in Subsec.~\ref{appsubsec:proofGuerraSup}, we can compute the derivative $\partial_t\mathcal{A}(t)$.

{
\begin{equation}\footnotesize
	\begin{array}{lllll}
    \partial_t\mathcal{A}(t)= \left[  \a\beta \dfrac{\sqrt{N_1 N_2} }{L} \sqrt{1+\rho_1}\sqrt{1+\rho_2} \dfrac{1}{M} \SOMMA{a=1}{M}\langle n_{\Xi^{1,a}}^\sigma n_{\Theta^{1,a}}^\tau\rangle \right.+   
		\\\\ \noalign{\vspace{-10pt}}
		\qquad\qquad+\b\beta \dfrac{\sqrt{N_1 N_3} }{L} \sqrt{1+\rho_1}\sqrt{1+\rho_3} \dfrac{1}{M} \SOMMA{a=1}{M}\langle n_{\Xi^{1,a}}^\sigma n_{\Upsilon^{\mu,a}}^\phi\rangle +
		\\\\ \noalign{\vspace{-10pt}}
		\left.\qquad\qquad+ \c\beta \dfrac{\sqrt{N_2 N_3} }{L} \sqrt{1+\rho_2}\sqrt{1+\rho_3} \dfrac{1}{M} \SOMMA{a=1}{M}\langle n_{\Theta^{1,a}}^\tau n_{\Upsilon^{\mu,a}}^\phi\rangle  \right]+
		\\\\ \noalign{\vspace{-10pt}}
		\qquad\qquad-\dfrac{\sqrt{N_1 N_2} }{L}\left[ \psi_\sigma^{(1)}\SOMMA{a=1}{M}\langle n_{\Xi^{1,a}}^\sigma\rangle + \psi_\tau^{(1)}\SOMMA{a=1}{M}\langle n_{\Theta^{1,a}}^\tau\rangle    \right] +
		\\\\ \noalign{\vspace{-10pt}}
		\qquad\qquad- \dfrac{\sqrt{N_1 N_3} }{L}\left[ 
		\psi_\sigma^{(2)}\SOMMA{a=1}{M}\langle n_{\Xi^{1,a}}^\sigma\rangle + \psi_\phi^{(1)}\SOMMA{a=1}{M}\langle n_{\Upsilon^{\mu,a}}^\phi\rangle\right]+
		\\\\ \noalign{\vspace{-10pt}}
		\qquad\qquad-\dfrac{\sqrt{N_2 N_3}}{L} \left[ \psi_\tau^{(2)}\SOMMA{a=1}{M}\langle n_{\Theta^{1,a}}^\tau\rangle + \psi_\phi^{(2)}\SOMMA{a=1}{M}\langle n_{\Upsilon^{\mu,a}}^\phi\rangle \right]+
		\\\\ \noalign{\vspace{-10pt}}
		\qquad\qquad+\dfrac{K-1}{L}\dfrac{\beta^2}{2}\a^2\dfrac{1+\rho_{12}}{(1+\rho_{1})(1+\rho_2)}\Big( 1-\langle q_{ab}^\sigma q_{ab}^\tau  \rangle\Big)+
		\\\\ \noalign{\vspace{-10pt}}
		\qquad\qquad+\dfrac{K-1}{L}\dfrac{\beta^2}{2}\b^2\dfrac{1+\rho_{13}}{(1+\rho_{1})(1+\rho_3)}\left( 1-\langle q_{ab}^\sigma q_{ab}^\phi  \rangle\right)+
		\\\\ \noalign{\vspace{-10pt}}
		\qquad\qquad+\dfrac{K-1}{L}\dfrac{\beta^2}{2}\c^2\dfrac{1+\rho_{23}}{(1+\rho_{2})(1+\rho_3)}\left( 1-\langle q_{ab}^\tau q_{ab}^\phi  \rangle\right)+
		\\\\ \noalign{\vspace{-10pt}}
		\qquad\qquad-\dfrac{A_\sigma^2}{2}N_1 \Big( 1-\langle q_{ab}^\sigma\rangle \Big) -\dfrac{A_\tau^2}{2}N_2 \Big( 1-\langle q_{ab}^\tau\rangle \Big) -\dfrac{A_\phi^2}{2}N_3 \Big( 1-\langle q_{ab}^\phi\rangle \Big).

	\end{array}
\end{equation}}

\noindent
We can fix the parameters as follows: 
{
\begin{equation}\footnotesize
\begin{array}{ll}
	\psi_\sigma^{(1)}=\dfrac{\a\beta}{M}\sqrt{(1+\rho_{1})(1+\rho_2)} \bar n_{\Theta^{1,a}}^\tau, & \psi_\sigma^{(2)}=\dfrac{\b\beta}{M}\sqrt{(1+\rho_{1})(1+\rho_3)} \bar n_{\Upsilon^{1,a}}^\phi, \\
	\psi_\tau^{(1)}=\dfrac{\a\beta}{M}\sqrt{(1+\rho_{1})(1+\rho_2)} \bar n_{\Xi^{1,a}}^\sigma, & \psi_\tau^{(2)}=\dfrac{\c\beta}{M}\sqrt{(1+\rho_{2})(1+\rho_3)} \bar n_{\Upsilon^{1,a}}^\phi, \\
	\psi_\phi^{(1)}=\dfrac{\b\beta}{M}\sqrt{(1+\rho_{1})(1+\rho_3)} \bar n_{\Xi^{1,a}}^\sigma, & \psi_\phi^{(2)}=\dfrac{\c\beta}{M}\sqrt{(1+\rho_{2})(1+\rho_3)} \bar n_{\Theta^{1,a}}^\tau, \\
    \end{array}
\end{equation}
\begin{equation}\footnotesize
    \begin{array}{lllll}  (A_\sigma^{(\tau)})^2=\a^2\beta^2\dfrac{K-1}{N_1}\dfrac{1+\rho_{12}}{(1+\rho_{1})(1+\rho_2)} \qtau, & (A_\sigma^{(\phi)})^2=\b^2\beta^2\dfrac{K-1}{N_1}\dfrac{1+\rho_{13}}{(1+\rho_{1})(1+\rho_3)} \qphi, \\
	(A_\tau^{(\sigma)})^2=\a^2\beta^2\dfrac{K-1}{N_2}\dfrac{1+\rho_{12}}{(1+\rho_{1})(1+\rho_2)} \qsigma, & (A_\tau^{(\phi)})^2=\c^2\beta^2\dfrac{K-1}{N_2}\dfrac{1+\rho_{23}}{(1+\rho_{2})(1+\rho_3)} \qphi, \\
	(A_\phi^{(\sigma)})^2=\b^2\beta^2\dfrac{K-1}{N_3}\dfrac{1+\rho_{13}}{(1+\rho_{1})(1+\rho_3)} \qsigma, & (A_\phi^{(\tau)})^2=\c^2\beta^2\dfrac{K-1}{N_3}\dfrac{1+\rho_{23}}{(1+\rho_{2})(1+\rho_3)} \qtau. \\
\end{array}
\end{equation}}

\noindent
We obtain the final expression for the derivative of the interpolating free energy in the thermodynamic limit, which reads

{
\begin{equation}\footnotesize\label{eq:TAMstraming_computed}
	\begin{array}{lllll}
		\partial_t\mathcal{A}(t) = - \dfrac{\beta}{LM}\SOMMA{a=1}{M}\left[ \a \sqrt{N_1 N_2} \sqrt{(1+\rho_{1})(1+\rho_2)}  \bar n_{\xi^{1}}^\sigma  \bar n_{\eta^{1}}^\tau \right.+ 
		\\\\ \noalign{\vspace{-10pt}}
		\qquad\qquad\qquad\qquad\quad+\b \sqrt{N_1 N_3} \sqrt{(1+\rho_{1})(1+\rho_3)}  \bar n_{\xi^{1}}^\sigma  \bar n_{\chi^{1}}^\phi \, +
		\\\\ \noalign{\vspace{-5pt}}
		\qquad\qquad\qquad\qquad\quad\left. + \c \sqrt{N_2 N_3} \sqrt{(1+\rho_{2})(1+\rho_3)}  \bar n_{\eta^{1}}^\tau  \bar n_{\chi^{1}}^\phi \right] +
		\\\\ \noalign{\vspace{-10pt}}
		\qquad\qquad-\a^2 \beta^2 \dfrac{K-1}{2L} \dfrac{1+\rho_{12}}{(1+\rho_{1})(1+\rho_2)}  \qtau \left( 1- \qsigma\right) 
		-\b^2 \beta^2 \dfrac{K-1}{2L} \dfrac{1+\rho_{13}}{(1+\rho_{1})(1+\rho_3)}  \qphi \left( 1- \qsigma\right)\,+
		\\\\ \noalign{\vspace{-10pt}}
		\qquad\qquad-\a^2 \beta^2 \dfrac{K-1}{2L} \dfrac{1+\rho_{12}}{(1+\rho_{1})(1+\rho_2)}  \qsigma \left( 1- \qtau\right) 
		-\c^2 \beta^2 \dfrac{K-1}{2L} \dfrac{1+\rho_{23}}{(1+\rho_{2})(1+\rho_3)}  \qphi \left( 1- \qtau\right)\,+
		\\\\ \noalign{\vspace{-10pt}}
		\qquad\qquad-\b^2 \beta^2 \dfrac{K-1}{2L} \dfrac{1+\rho_{13}}{(1+\rho_{1})(1+\rho_3)}  \qsigma \left( 1- \qphi\right) 
		-\c^2 \beta^2 \dfrac{K-1}{2L} \dfrac{1+\rho_{23}}{(1+\rho_{2})(1+\rho_3)}  \qtau \left( 1- \qphi\right)+
		\\\\ \noalign{\vspace{-10pt}}
		\qquad\qquad+\dfrac{K-1}{L}\dfrac{\beta^2}{2}\Bigg[ \a^2 \dfrac{1+\rho_{12}}{(1+\rho_{1})(1+\rho_2)} + \b^2 \dfrac{1+\rho_{13}}{(1+\rho_{1})(1+\rho_3)}+  \c^2 \dfrac{1+\rho_{23}}{(1+\rho_{2})(1+\rho_3)} \Bigg].
	\end{array}
\end{equation}}

\noindent
We are now ready to calculate the Cauchy condition $\mathcal{A}(t=0)$ as follows:
{\begin{equation}\small
    \mathcal A(t=0) = \lim_{\bm N \to \infty} \frac{1}{L}\ln \mathbb{E}\mathcal{Z}_{\bm N}(t=0),
\end{equation}}

\noindent
where the interpolating partition function, evaluated at $t=0$, reads as
{
\begin{equation}\hspace{-0.5cm}\footnotesize
	\begin{array}{lllll}
		\mathcal{Z}_{\bm N}(t=0) =
		\\\\ \noalign{\vspace{-10pt}}
		= \SOMMA{\{\sigma\},\{\tau\},\{\phi\}}{2^{N_1},2^{N_2},2^{N_3}} \exp \Bigg\{ \left( \sqrt{N_1 N_2}\psi_\sigma^{(1)} + \sqrt{N_1 N_3}\psi_\sigma^{(2)}  \right)\left[ \dfrac{1}{N_1 r_1 (1+\rho_1)} \SOMMA{a=1}{M}\SOMMA{i=1}{N_1} \Xi_i^{\mu,a}\sigma_i \right] + 
		\\\\ \noalign{\vspace{-10pt}}
		\qquad\qquad\qquad\qquad\quad+ A_\sigma^{(\tau)}\sum_i\varphi_i^{\sigma,\tau} + A_\sigma^{(\phi)}\sum_i\varphi_i^{\sigma,\phi}\Bigg\}\times
		\\\\ \noalign{\vspace{-10pt}}
		\qquad\qquad\quad\times \exp \Bigg\{ \left( \sqrt{N_1 N_2}\psi_\tau^{(1)} + \sqrt{N_2 N_3}\psi_\tau^{(2)}  \right)\left[ \dfrac{1}{N_2 r_2 (1+\rho_2)} \SOMMA{a=1}{M}\SOMMA{j=1}{N_2} \Theta_j^{\mu,a}\tau_j \right] +
		\\\\ \noalign{\vspace{-10pt}}
		\qquad\qquad\qquad\qquad\quad+ A_\tau^{(\sigma)}\sum_j\varphi_j^{\tau,\sigma} + A_\tau^{(\phi)}\sum_j\varphi_j^{\tau,\phi}\Bigg\}\times
		\\\\ \noalign{\vspace{-10pt}}
		\qquad\qquad\quad\times \exp \Bigg\{ \left( \sqrt{N_1 N_3}\psi_\phi^{(1)} + \sqrt{N_2 N_3}\psi_\phi^{(2)}  \right)\left[ \dfrac{1}{N_3 r_3 (1+\rho_3)} \SOMMA{a=1}{M}\SOMMA{k=1}{N_3} \Upsilon_k^{\mu,a}\phi_k \right] +
		\\\\ \noalign{\vspace{-10pt}}
		\qquad\qquad\qquad\qquad\quad+ A_\phi^{(\sigma)}\sum_k\varphi_k^{\phi,\sigma} + A_\phi^{(\tau)}\sum_k\varphi_k^{\phi,\tau}\Bigg\}.
	\end{array}
\end{equation}}

\vspace{3mm}
\noindent
So the One-body statistical pressure at finite size reads as:

\begin{equation}\label{eq:TAMone_body_computed}
    \mathcal{A}(t=0,\beta) = U + V + W,
\end{equation}
where we have defined

\begin{equation}
    \begin{aligned}
U \doteq \dfrac{N_1}{L} \, \mathbb{E} \ln 2 \cosh \Bigg\{ 
\bigg( & \dfrac{\a\beta}{M} \sqrt{N_1 N_2} \sqrt{(1+\rho_1)(1+\rho_2)} \, \bar n_{\eta^1}^\tau + \\
 + & \dfrac{\b\beta}{M} \sqrt{N_1 N_3} \sqrt{(1+\rho_1)(1+\rho_3)} \, \bar n_{\chi^1}^\phi \bigg) \cdot \\
 \cdot & \left( \dfrac{1}{N_1 r_1 (1+\rho_1)} \sum_{a=1}^M \Xi^{1,a} \right) \\
 + & \a\beta \sqrt{ \dfrac{K-1}{N_1} \cdot \dfrac{1+\rho_{12}}{(1+\rho_1)(1+\rho_2)} \cdot \qtau } \cdot \varphi^{\sigma,\tau} \\
 + & \b\beta \sqrt{ \dfrac{K-1}{N_1} \cdot \dfrac{1+\rho_{13}}{(1+\rho_1)(1+\rho_3)} \cdot \qphi } \cdot \varphi^{\sigma,\phi}
\Bigg\},
\end{aligned}
\end{equation}

\begin{equation}
    \begin{aligned}
V \doteq \dfrac{N_2}{L} \, \mathbb{E} \ln 2 \cosh \Bigg\{ 
\bigg( & \dfrac{\a\beta}{M} \sqrt{N_1 N_2} \sqrt{(1+\rho_1)(1+\rho_2)} \, \bar n_{\xi^1}^\sigma + \\
+ & \dfrac{\c\beta}{M} \sqrt{N_2 N_3} \sqrt{(1+\rho_2)(1+\rho_3)} \, \bar n_{\chi^1}^\phi \bigg) \cdot \\
 \cdot & \left( \dfrac{1}{N_2 r_2 (1+\rho_2)} \sum_{a=1}^M \Theta^{1,a} \right) \\
 + & \a\beta \sqrt{ \dfrac{K-1}{N_2} \cdot \dfrac{1+\rho_{12}}{(1+\rho_1)(1+\rho_2)} \cdot \qsigma } \cdot \varphi^{\tau,\sigma} \\
 + & \c\beta \sqrt{ \dfrac{K-1}{N_2} \cdot \dfrac{1+\rho_{23}}{(1+\rho_2)(1+\rho_3)} \cdot \qphi } \cdot \varphi^{\tau,\phi}
\Bigg\},
\end{aligned}
\end{equation}

\begin{equation}
    \begin{aligned}
W \doteq \dfrac{N_3}{L} \, \mathbb{E} \ln 2 \cosh \Bigg\{ 
\bigg( & \dfrac{\b\beta}{M} \sqrt{N_1 N_3} \sqrt{(1+\rho_1)(1+\rho_3)} \, \bar n_{\xi^1}^\sigma + \\
+ &
        \dfrac{\c\beta}{M} \sqrt{N_2 N_3} \sqrt{(1+\rho_2)(1+\rho_3)} \, \bar n_{\eta^1}^\tau \bigg) \cdot \\
 \cdot & \left( \dfrac{1}{N_3 r_3 (1+\rho_3)} \sum_{a=1}^M \Upsilon^{1,a} \right) \\
 + & \b\beta \sqrt{ \dfrac{K-1}{N_3} \cdot \dfrac{1+\rho_{13}}{(1+\rho_1)(1+\rho_3)} \cdot \qsigma } \cdot \varphi^{\phi,\sigma} \\
 + & \c\beta \sqrt{ \dfrac{K-1}{N_3} \cdot \dfrac{1+\rho_{23}}{(1+\rho_2)(1+\rho_3)} \cdot \qsigma } \cdot \varphi^{\phi,\sigma}
\Bigg\} \,.
\end{aligned}
\end{equation}

\noindent
Now, in the thermodynamic limit, fixing $\varsigma = \dfrac{1}{3}(\alpha + \theta + \alpha\theta)$, and merging these two expressions, that is the integral of the $t$-streaming (Eq.~\eqref{eq:TAMstraming_computed}) and the Cauchy condition (Eq.~\eqref{eq:TAMone_body_computed}), in Eq.\eqref{eq:pressioneInterp} under the replica symmetric assumption,  we get the explicit expression of the quenched free energy in terms of control and order parameters:
{
\begin{equation}\footnotesize
\label{eq:TAMUnsupPressMain}
	\begin{array}{lllll}
			\mathcal{A}^{(unsup)}_{\alpha,\theta,\gamma}(\beta)= -\beta\varsigma \left( \a\theta^{-1}\sqrt{(1+\rho_{1})(1+\rho_2)} \bar n_{\xi^{1}}^\sigma  \bar n_{\eta^{1}}^\tau +  \b\alpha^{-1}\sqrt{(1+\rho_{1})(1+\rho_3)} \bar n_{\xi^{1}}^\sigma  \bar n_{\chi^{1}}^\phi + \right.
			\\\\ \noalign{\vspace{-10pt}}
			\left. + \c(\alpha\theta)^{-1}\sqrt{(1+\rho_{2})(1+\rho_3)} \bar n_{\eta^{1}}^\tau  \bar n_{\chi^{1}}^\phi \right)+
			\\\\ \noalign{\vspace{-10pt}}
			+\dfrac{\gamma\varsigma}{2}\beta^2\left[ \a^2  \dfrac{1+\rho_{12}}{(1+\rho_{1})(1+\rho_2)} \left( 1- \qtau\right) \left( 1- \qsigma\right) 
			+\b^2  \dfrac{1+\rho_{13}}{(1+\rho_{1})(1+\rho_3)} \left( 1- \qphi\right) \left( 1- \qsigma\right)  \right. +
			\\\\ \noalign{\vspace{-10pt}}
			\left. +\c^2\dfrac{1+\rho_{23}}{(1+\rho_{2})(1+\rho_3)} \left(1-  \qphi \right) \left( 1-  \qtau \right)\right] +
			\\\\ \noalign{\vspace{-10pt}}
			+\varsigma\mathbb{E}\ln 2 \cosh \Bigg\{ \beta\left( \a \theta^{-1} \sqrt{(1+\rho_{1})(1+\rho_2)}  \bar n_{\eta^{1}}^\tau + \b\alpha^{-1} \sqrt{(1+\rho_{1})(1+\rho_3)}  \bar n_{\chi^{1}}^\phi  \right)\hat\xi^1 +
			\\\\ \noalign{\vspace{-10pt}}
			 +\a\beta \sqrt{\gamma \dfrac{1+\rho_{12}}{(1+\rho_{1})(1+\rho_2)} \qtau}\varphi^{\sigma,\tau}+ \b\beta\sqrt{\gamma \dfrac{1+\rho_{13}}{(1+\rho_{1})(1+\rho_3)} \qphi} \varphi^{\sigma,\phi} \Bigg\}+
			\\\\ \noalign{\vspace{-10pt}}
			+\varsigma\theta^{-2}\mathbb{E}\ln 2 \cosh \Bigg\{ \theta\beta\left( \a \sqrt{(1+\rho_{1})(1+\rho_2)}  \bar n_{\xi^{1}}^\sigma + \c\alpha^{-1} \sqrt{(1+\rho_{2})(1+\rho_3)}  \bar n_{\chi^{1}}^\phi  \right)\hat\eta^1 +
			\\\\ \noalign{\vspace{-10pt}}
			  +\a\beta\theta \sqrt{\gamma \dfrac{1+\rho_{12}}{(1+\rho_{1})(1+\rho_2)} \qsigma}\varphi^{\tau,\sigma}+ \c\beta\theta\sqrt{\gamma \dfrac{1+\rho_{23}}{(1+\rho_{2})(1+\rho_3)} \qphi} \varphi^{\tau,\phi} \Bigg\}+
			\\\\ \noalign{\vspace{-10pt}}
			+\varsigma\alpha^{-2}\mathbb{E}\ln 2 \cosh \Bigg\{ \alpha\beta\left( \b \sqrt{(1+\rho_{1})(1+\rho_3)}  \bar n_{\xi^{1}}^\sigma + \c\theta^{-1} \sqrt{(1+\rho_{2})(1+\rho_3)}  \bar n_{\eta^{1}}^\tau  \right)\hat\chi^1 +
			\\\\ \noalign{\vspace{-10pt}}
			 +\b\beta\alpha \sqrt{\gamma \dfrac{1+\rho_{13}}{(1+\rho_{1})(1+\rho_3)} \qsigma}\varphi^{\phi,\sigma}+ \c\beta\alpha\sqrt{\gamma \dfrac{1+\rho_{23}}{(1+\rho_{2})(1+\rho_3)} \qsigma} \varphi^{\phi,\sigma} \Bigg\}.
		\end{array}
	\end{equation}}

The values of the order parameters at fixed control parameters can be obtained by evaluating
the saddle points of \eqref{eq:TAMUnsupPressMain}. The set of self-consistent equations to be fulfilled by
these values is shown below:
{
\begin{equation}\footnotesize
	\begin{array}{lllll}
		\bar n_{\xi^{1}}^\sigma = \dfrac{1}{\sqrt{1+\rho_1}}\mathbb{E}\left\{ \hat \xi^1\tanh \Bigg[ \beta\left( \a \theta^{-1} \sqrt{(1+\rho_{1})(1+\rho_2)}  \bar n_{\eta^{1}}^\tau + \b\alpha^{-1} \sqrt{(1+\rho_{1})(1+\rho_3)}  \bar n_{\chi^{1}}^\phi  \right)\hat\xi^1+  \right. 
		\\\\ \noalign{\vspace{-10pt}}
		\left.\left. +\a\beta \sqrt{\gamma \dfrac{1+\rho_{12}}{(1+\rho_{1})(1+\rho_2)} \qtau}\varphi^{\sigma,\tau}+ \b\beta\sqrt{\gamma \dfrac{1+\rho_{13}}{(1+\rho_{1})(1+\rho_3)} \qphi} \varphi^{\sigma,\phi} \right]\right\},
	\end{array}
\end{equation}

\begin{equation}\footnotesize
	\begin{array}{lllll}
		\bar n_{\eta^{1}}^\tau = \dfrac{1}{\sqrt{1+\rho_2}}\mathbb{E}\left\{ \hat \eta^1\tanh \Bigg[ \theta\beta\left( \a \sqrt{(1+\rho_{1})(1+\rho_2)}  \bar n_{\xi^{1}}^\sigma + \c\alpha^{-1} \sqrt{(1+\rho_{2})(1+\rho_3)}  \bar n_{\chi^{1}}^\phi  \right)\hat\eta^1  +\right. 
		\\\\ \noalign{\vspace{-10pt}}
		\left.\left.  +\a\beta\theta \sqrt{\gamma \dfrac{1+\rho_{12}}{(1+\rho_{1})(1+\rho_2)} \qsigma}\varphi^{\tau,\sigma}+ \c\beta\theta\sqrt{\gamma \dfrac{1+\rho_{23}}{(1+\rho_{2})(1+\rho_3)} \qphi} \varphi^{\tau,\phi} \right]\right\},
	\end{array}
\end{equation}

\begin{equation}\footnotesize
	\begin{array}{lllll}
		\bar n_{\chi^{1}}^\phi = \dfrac{1}{\sqrt{1+\rho_3}}\mathbb{E}\left\{ \hat \chi^1\tanh \Bigg[ \alpha\beta\left( \b \sqrt{(1+\rho_{1})(1+\rho_3)}  \bar n_{\xi^{1}}^\sigma + \c\theta^{-1} \sqrt{(1+\rho_{2})(1+\rho_3)}  \bar n_{\eta^{1}}^\tau  \right)\hat\chi^1  +\right. 
		\\\\ \noalign{\vspace{-10pt}}
		\left.\left.  +\b\beta\alpha \sqrt{\gamma \dfrac{1+\rho_{13}}{(1+\rho_{1})(1+\rho_3)} \qsigma}\varphi^{\phi,\sigma}+ \c\beta\alpha\sqrt{\gamma \dfrac{1+\rho_{23}}{(1+\rho_{2})(1+\rho_3)} \qsigma} \varphi^{\phi,\sigma} \right]\right\},
	\end{array}
\end{equation}}

\vspace{3mm}
\noindent
and the overlaps:
{
\begin{equation}\footnotesize
	\begin{array}{lllll}
		\qsigma   = \mathbb{E}\left\{ \tanh^2 \Bigg[ \beta\left( \a \theta^{-1} \sqrt{(1+\rho_{1})(1+\rho_2)}  \bar n_{\eta^{1}}^\tau + \b\alpha^{-1} \sqrt{(1+\rho_{1})(1+\rho_3)}  \bar n_{\chi^{1}}^\phi  \right)\hat{\xi}^1 + \right. 
		\\\\ \noalign{\vspace{-10pt}}
		\left.\left.+\a\beta \sqrt{\gamma \dfrac{1+\rho_{12}}{(1+\rho_{1})(1+\rho_2)} \qtau}\varphi^{\sigma,\tau}+ \b\beta\sqrt{\gamma \dfrac{1+\rho_{13}}{(1+\rho_{1})(1+\rho_3)} \qphi} \varphi^{\sigma,\phi} \right]\right\},
	\end{array}
\end{equation}
\begin{equation}\footnotesize
	\begin{array}{lllll}
		\qtau = \mathbb{E}\left\{ \tanh^2 \Bigg[ \theta\beta\left( \a \sqrt{(1+\rho_{1})(1+\rho_2)}  \bar n_{\xi^{1}}^\sigma + \c\alpha^{-1} \sqrt{(1+\rho_{2})(1+\rho_3)}  \bar n_{\chi^{1}}^\phi  \right)\hat{\eta}^1 +\right. 
		\\\\ \noalign{\vspace{-10pt}}
		\left.\left. +\a\beta\theta \sqrt{\gamma \dfrac{1+\rho_{12}}{(1+\rho_{1})(1+\rho_2)} \qsigma}\varphi^{\tau,\sigma}+ \c\beta\theta\sqrt{\gamma \dfrac{1+\rho_{23}}{(1+\rho_{2})(1+\rho_3)} \qphi} \varphi^{\tau,\phi} \right]\right\},
	\end{array}
\end{equation}
\begin{equation}\footnotesize
	\begin{array}{lllll}
		\qphi  = \mathbb{E}\left\{ \tanh^2 \Bigg[ \alpha\beta\left( \b \sqrt{(1+\rho_{1})(1+\rho_3)}  \bar n_{\xi^{1}}^\sigma + \c\theta^{-1} \sqrt{(1+\rho_{2})(1+\rho_3)}  \bar n_{\eta^{1}}^\tau  \right)\hat{\chi}^1 +\right. 
		\\\\ \noalign{\vspace{-10pt}}
		\left.\left. +\b\beta\alpha \sqrt{\gamma \dfrac{1+\rho_{13}}{(1+\rho_{1})(1+\rho_3)} \qsigma}\varphi^{\phi,\sigma}+ \c\beta\alpha\sqrt{\gamma \dfrac{1+\rho_{23}}{(1+\rho_{2})(1+\rho_3)} \qsigma} \varphi^{\phi,\sigma} \right]\right\}.
	\end{array}
\end{equation}}

\vspace{3mm}
\noindent
Then exploiting the generator function, we get the self-consistency equations for the three archetypes - Mattis magnetization.

{
\begin{equation}\footnotesize
	\begin{array}{lllll}
		\bar{m}^\sigma_{\xi^1} = \dfrac{1}{\sqrt{1+\rho_1}}\mathbb{E}\left\{  \xi^1 \tanh \Bigg[ \beta\left( \a \theta^{-1} \sqrt{(1+\rho_{1})(1+\rho_2)}  \bar n_{\eta^{1}}^\tau + \b\alpha^{-1} \sqrt{(1+\rho_{1})(1+\rho_3)}  \bar n_{\chi^{1}}^\phi  \right)\xi^1 +  \right. 
		\\\\ \noalign{\vspace{-10pt}}
		\left.\left.+\a\beta \sqrt{\gamma \dfrac{1+\rho_{12}}{(1+\rho_{1})(1+\rho_2)} \qtau}\varphi^{\sigma,\tau}+ \b\beta\sqrt{\gamma \dfrac{1+\rho_{13}}{(1+\rho_{1})(1+\rho_3)} \qphi} \varphi^{\sigma,\phi} \right]\right\};
	\end{array}
\end{equation}

\begin{equation}\footnotesize
	\begin{array}{lllll}
		\bar{m}^\tau_{\eta^1} = \dfrac{1}{\sqrt{1+\rho_2}}\mathbb{E}\left\{ \eta^1\tanh \Bigg[ \theta\beta\left( \a \sqrt{(1+\rho_{1})(1+\rho_2)}  \bar n_{\xi^{1}}^\sigma + \c\alpha^{-1} \sqrt{(1+\rho_{2})(1+\rho_3)}  \bar n_{\chi^{1}}^\phi  \right)\eta^1+ \right. 
		\\\\ \noalign{\vspace{-10pt}}
		\left.\left.   +\a\beta\theta \sqrt{\gamma \dfrac{1+\rho_{12}}{(1+\rho_{1})(1+\rho_2)} \qsigma}\varphi^{\tau,\sigma}+ \c\beta\theta\sqrt{\gamma \dfrac{1+\rho_{23}}{(1+\rho_{2})(1+\rho_3)} \qphi} \varphi^{\tau,\phi} \right]\right\};
	\end{array}
\end{equation}

\begin{equation}\footnotesize
	\begin{array}{lllll}
		\bar{m}^\phi_{\chi^1} = \dfrac{1}{\sqrt{1+\rho_3}}\mathbb{E}\left\{ \chi^1\tanh \Bigg[ \alpha\beta\left( \b \sqrt{(1+\rho_{1})(1+\rho_3)}  \bar n_{\xi^{1}}^\sigma + \c\theta^{-1} \sqrt{(1+\rho_{2})(1+\rho_3)}  \bar n_{\eta^{1}}^\tau  \right)\chi^1+\right. 
		\\\\ \noalign{\vspace{-10pt}}
		\left.\left. +\b\beta\alpha \sqrt{\gamma \dfrac{1+\rho_{13}}{(1+\rho_{1})(1+\rho_3)} \qsigma}\varphi^{\phi,\sigma}+ \c\beta\alpha\sqrt{\gamma \dfrac{1+\rho_{23}}{(1+\rho_{2})(1+\rho_3)} \qsigma} \varphi^{\phi,\sigma} \right]\right\}.
	\end{array}
\end{equation}}

%% file: appendices/largeSample.tex
In this section we will explore the case of a large number of examples for each kind of dataset. Under this assumption we can introduce some approximation in the previous self consistency equations that allow us to recast them in a more compact and useful form.

\noindent
Both Supervised and Unsupervised settings will use the following results:

\begin{lemma}\label{lemmaLargeM}
Given the explicit expression of $\hat{\xi}^1$ \eqref{eq:rho_def} the following expectation
\begin{equation}
    \mathbb{E}_{(\bm{\Xi}|\bm{\xi})}[\hat\xi^1]= \mathbb{E}_{(\bm{\Xi}|\bm{\xi})}\left[\dfrac{1}{M_1 r_1 (1+\rho_1)}\SOMMA{a=1}{M_1}\Xi^{1,a}\right]
    \label{eq:gauss_approx_M_large}
\end{equation}
can be simply approximated using the CLT on the sum over $a$, under the assumption of $M_1\gg 1$, namely
\begin{equation}
      \mathbb{E}_{(\bm{\Xi}|\bm{\xi})}[\hat{\xi}^1]\simeq \dfrac{\xi^1}{1+\rho_1}\left(1+\mathcal{G}^1\sqrt{\rho_1}\right)\;\;\;\mathrm{with}\;\;\mathcal{G}^1\sim\mathcal{N}(0,1).
\end{equation}
\end{lemma}

\begin{proof}
Using the explicit expression of $\hat{\xi}^\mu$ \eqref{eq:rho_def} and exploiting the CLT on the sum over the examples, in the limit of a great number of examples (i.e. $M_1\gg 1$), we get 
\begin{equation}
    \hat{\xi}^\mu \sim \mu_1+\mathcal{G}^\mu \sqrt{\mu_2-\mu_1^2}
\end{equation}
where $\mathcal{G}^\mu\sim\mathcal{N}(0,1)$, and 
\begin{equation}
    \begin{array}{lll}
         \mu_1=\mathbb{E}_{(\bm\Xi|\bm\xi)}\left[\hat{\xi}^\mu\right]=\dfrac{\xi^1}{1+\rho_1},
         \\\\
         \mu_2=\mathbb{E}_{(\bm\Xi|\bm\xi)}\left[(\hat{\xi}^\mu)^2\right]=\dfrac{1}{1+\rho_1}.
    \end{array}
\end{equation}
rearranging all together, we complete the proof.
\end{proof}
\noindent
We have presented only the case for $\hat\xi^1$, however, the previous lemma also applies to the other two layers; we only need to replace $\rho_1\to \rho_2,\rho_3$.

\begin{lemma}
\label{lemmasumg}
    If $x_i$, $i=1,\dots,N$ are independent standard Gaussian variables, and $A$ and $B_i$, $i=1,\dots,N$ are constants, then for any sufficiently regular function $f$
    \begin{equation}
        \mathbb{E}_{x_1,\dots,x_N}\left[f\left(A+\SOMMA{i=1}{N}x_i B_i\right)\right]= \mathbb{E}_{z}\left[f\left(A+z\sqrt{\SOMMA{i=1}{N} B_i^2}\right)\right],
    \end{equation}
    with $z\sim \mathcal{N}(0,1).$
\end{lemma}

\subsubsection*{Supervised setting} \label{appsubsec:largeSampleSup}
We will now work in the supervised setting, while the unsupervised setting will follow the same path.

\noindent
Setting
 \begin{align}
     A &= \dfrac{\a}{\theta}\sqrt{(1+\rho_1)(1+\rho_2)}\bar{n}^\tau_{\eta^1} + \dfrac{\b}{\alpha}\sqrt{(1+\rho_1)(1+\rho_3)}\bar{n}^\phi_{\chi^1}\,,
     \\
     B &= \sqrt{\dfrac{\gamma\beta \bar{p}^z}{2}}
 \end{align}
and exploiting \eqref{eq:gauss_approx_M_large} in the first equation of \eqref{eq:self_eq_sup_low_M} we get
\begin{equation}
    \begin{array}{lllll}
     
     \bar{n}^\sigma_{\xi^1} &=& \mathbb{E}_{X^{\sigma},\xi,\mathcal{G}}\left\{\dfrac{1}{1+\rho_1}\xi^1(1 + \mathcal{G}^1\sqrt{\rho_1})\tanh{\left[\dfrac{\beta}{1+\rho_1}\xi^1(1 + \mathcal{G}^1\sqrt{\rho_1})A + X^{\sigma} B\right]}\right\}
     \\\\ \noalign{\vspace{-10pt}}
     &=& \dfrac{1}{1 + \rho_1}\mathbb{E}_{X^{\sigma},\xi,\mathcal{G}}\left\{\xi^1\tanh{\left[\dfrac{\beta}{1+\rho_1}\xi^1(1 + \mathcal{G}^1\sqrt{\rho_1})A + X^{\sigma} B\right]}\right\} +
     \\\\ \noalign{\vspace{-10pt}}
     &&+\dfrac{\sqrt{\rho_1}}{1 + \rho_1}\mathbb{E}_{X^{\sigma},\xi,\mathcal{G}}\left\{\xi^1\mathcal{G}^1\tanh{\left[\dfrac{\beta}{1+\rho_1}\xi^1A + \dfrac{\xi^1\mathcal{G}^1\sqrt{\rho_1}A\beta}{1+\rho_1} + X^{\sigma} B\right]}\right\}
     \\\\ \noalign{\vspace{-10pt}}
     &=& \dfrac{\bar{m}^\sigma_{\xi^1}}{1 + \rho_1}+\dfrac{\sqrt{\rho_1}}{1 + \rho_1}\mathbb{E}_{X^{\sigma},\xi,\mathcal{G}}\left\{\xi^1\mathcal{G}^1\tanh{\left[\dfrac{\beta}{1+\rho_1}\xi^1A + \dfrac{\xi^1\mathcal{G}^1\sqrt{\rho_1}A\beta}{1+\rho_1} + X^{\sigma} B\right]}\right\}
     \end{array}
 \end{equation}
 where the first term is clearly $m^{\sigma}_{\xi^1}$ of Eq.~\eqref{mattisSelf}.
\\Now exploiting on the r.h.s. of the last line of the previous equation the Wick's Theorem over $\mathcal{G}^1$ we get
\begin{equation}
    \begin{array}{lllll}
     
     \bar{n}^\sigma_{\xi^1} &=&  \dfrac{\bar{m}^\sigma_{\xi^1}}{1 + \rho_1}+\dfrac{\beta A \rho_1}{(1 + \rho_1)^2}\mathbb{E}_{X^{\sigma}, \xi}\left\{1 - \tanh^2{\left[\dfrac{\beta}{1+\rho_1}\xi^1A + \dfrac{\xi^1\mathcal{G}^1\sqrt{\rho_1}A\beta}{1+\rho_1} + X^{\sigma} B\right]}\right\}.
     \end{array}
 \end{equation}
As $\mathcal{G}^1$ and $X^\sigma$ are both i.i.d. Normal Gaussian variable we can further simplify the previous expression using Lemma~\ref{lemmasumg} getting
\begin{equation}
    \begin{array}{lllll}
     
     \bar{n}^\sigma_{\xi^1} &=&  \dfrac{\bar{m}^\sigma_{\xi^1}}{1 + \rho_1}+\dfrac{\beta A \rho_1}{(1 + \rho_1)^2}\mathbb{E}_{z, \xi}\left\{1 - \tanh^2{\left[\dfrac{\beta}{1+\rho_1}\xi^1A +z\xi^1\sqrt{\dfrac{\rho_1\beta^2A^2}{(1+\rho_1)^2} +  B^2} \right]}\right\}.
     \end{array}
 \end{equation}
Now, by comparison with \eqref{eq:self_eq_sup_low_M} the self equation for $\bar{n}^\sigma_{\xi^1}$ can be rewritten as
{\small
\begin{equation}
    \begin{array}{lllll}
     \label{eq:nsigma} \bar{n}^\sigma_{\xi^1} &=&  \dfrac{\bar{m}^\sigma_{\xi^1}}{1 + \rho_1}+\dfrac{\beta  \rho_1}{(1 + \rho_1)^2}\left(\dfrac{\a}{\theta}\sqrt{(1+\rho_1)(1+\rho_2)}\bar{n}^\tau_{\eta^1} + \dfrac{\b}{\alpha}\sqrt{(1+\rho_1)(1+\rho_3)}\bar{n}^\phi_{\chi^1}\right)\times\left(1 - \qsigma\right).
     \end{array}
 \end{equation}
 }
Following the same steps it is possible to rewrite also the other two example-Mattis magnetization self equations
{\small
\begin{equation}
    \begin{array}{lllll}
 \label{eq:ntau}    
\bar{n}^\tau_{\eta^1} =  \dfrac{\bar{m}^\tau_{\eta^1}}{1 + \rho_2}+\dfrac{\beta  \rho_2}{(1 + \rho_2)^2}\theta\left(\a\sqrt{(1+\rho_1)(1+\rho_2)}\bar{n}^\sigma_{\xi^1} + \dfrac{\c}{\alpha}\sqrt{(1+\rho_2)(1+\rho_3)}\bar{n}^\phi_{\chi^1}\right)\times\left(1 - \qtau\right),
     \end{array}
 \end{equation}
\begin{equation}
    \begin{array}{lllll}
    \label{eq:nphi}
\bar{n}^\phi_{\chi^1} =  \dfrac{\bar{m}^\phi_{\chi^1}}{1 + \rho_3}+\dfrac{\beta  \rho_3}{(1 + \rho_3)^2}\alpha\left(\b\sqrt{(1+\rho_1)(1+\rho_3)}\bar{n}^\sigma_{\xi^1} + \dfrac{\c}{\theta}\sqrt{(1+\rho_3)(1+\rho_2)}\bar{n}^\tau_{\eta^1}\right)\times\left(1 - \qphi\right).
     \end{array}
\end{equation}
}
Resolve \eqref{eq:nsigma}, \eqref{eq:ntau} and \eqref{eq:nphi} with respect to $(\bar{n}^\sigma_{\xi^1},\bar{n}^\tau_{\eta^1},\bar{n}^\phi_{\chi^1})$, we obtain
\begin{equation}
\label{eq:n_large_M}
    \begin{array}{lll}
         \bar{n}^\sigma_{\xi^1}=\dfrac{A_1 + A_2 B_1 + A_3 C_1 + A_2 B_3 C_1 + A_3 B_1 C_3 - A_1 B_3 C_3}{1 - A_2 B_2 - A_3 C_2 - A_2 B_3 C_2 - A_3 B_2 C_3 - B_3 C_3},
         \\\\ \noalign{\vspace{-10pt}}
         \bar{n}^\tau_{\eta^1}=\dfrac{B_1 + A_1 B_2 + A_3 B_2 C_1 + B_3 C_1 - A_3 B_1 C_2 + A_1 B_3 C_2}{1 - A_2 B_2 - A_3 C_2 - A_2 B_3 C_2 - A_3 B_2 C_3 - B_3 C_3},
         \\\\ \noalign{\vspace{-10pt}}
         \bar{n}^\phi_{\chi^1}=\dfrac{C_1 - A_2 B_2 C_1 + A_1 C_2 + A_2 B_1 C_2 + B_1 C_3 + A_1 B_2 C_3}{1 - A_2 B_2 - A_3 C_2 - A_2 B_3 C_2 - A_3 B_2 C_3 - B_3 C_3},
    \end{array}
\end{equation}
where we define
\begin{equation}
    \begin{array}{lllll}
         A_1=\dfrac{\bar{m}_{\xi^1}^\sigma}{(1 + \rho_1)},&A_2=\dfrac{\beta \rho_1 }{(1 + \rho_1)^2} \tilde A_2 (1 - \qsigma),&A_3=\dfrac{\beta \rho_1}{(1 + \rho_1)^2} \tilde A_3 (1 - \qsigma),
         \\\\ \noalign{\vspace{-10pt}}
         B_1=\dfrac{\bar{m}_{\eta^1}^\tau}{(1 + \rho_3)},&B_2=\dfrac{\beta \rho_2 \theta}{(1 + \rho_2)^2} \tilde B_2 (1 - \qtau),&B_3=\dfrac{\beta \rho_2 \theta}{(1 + \rho_2)^2} \tilde B_3 (1 - \qtau),
         \\\\ \noalign{\vspace{-10pt}}
         C_1=\dfrac{\bar{m}_{\chi^1}^\phi}{(1 + \rho_3)},&C_2=\dfrac{\beta \rho_3 \alpha}{(1 + \rho_3)^2} \tilde C_2 (1 - \qphi),&C_3=\dfrac{\beta \rho_3 \alpha}{(1 + \rho_3)^2} \tilde C_3 (1 - \qphi),
    \end{array}
\end{equation}
where
\begin{equation}
    \begin{array}{lllll}
         \tilde A_2=\dfrac{\a}{\theta} \sqrt{(1 + \rho_1) (1 + \rho_2)} ,&\tilde A_3=\dfrac{\b}{\alpha} \sqrt{(1 + \rho_1) (1 + \rho_3)} ,
         \\\\ \noalign{\vspace{-10pt}}
         \tilde B_2=\a \sqrt{(1 + \rho_1) (1 + \rho_2)}  ,&\tilde B_3=\dfrac{ \c}{\alpha} \sqrt{(1 + \rho_2) (1 + \rho_3)} , 
         \\\\ \noalign{\vspace{-10pt}}
         \tilde C_2=\b \sqrt{(1 + \rho_1) (1 + \rho_3)}  ,&\tilde C_3=\dfrac{\c}{\theta} \sqrt{(1 + \rho_2) (1 + \rho_3)}.
    \end{array}
\end{equation}
This transformation ensure us to express every self consistency equations as a function of only $(\bar{m}^\sigma_{\xi^1},\bar{m}^\tau_{\eta^1},\bar{m}^\phi_{\chi^1}, \qsigma,\qtau,\qphi)$, see \eqref{seldefinitive} and \eqref{seldefinitive1}.

\subsubsection*{Unsupervised setting} \label{appsubsec:largeSampleUnsup}
As already mentioned, we follow the same conceptual steps as for the supervised setting.
For simplicity, we call:

\begin{align}
	A &= \a \theta^{-1} \sqrt{(1+\rho_{1})(1+\rho_2)}  \bar n_{\eta^{1}}^\tau + \b\alpha^{-1} \sqrt{(1+\rho_{1})(1+\rho_3)}  \bar n_{\chi^{1}}^\phi,
	\\
	B &= \sqrt{\gamma \dfrac{1+\rho_{12}}{(1+\rho_{1})(1+\rho_2)} \qtau},
	\\
	C &= \sqrt{\gamma \dfrac{1+\rho_{13}}{(1+\rho_{1})(1+\rho_3)} \qphi}.
\end{align}
\noindent
So, using Lemma~\ref{lemmaLargeM}

\begin{equation}
	\begin{array}{lllll}
		
		\bar n_{\xi^{1}}^\sigma  &=  \dfrac{1}{\sqrt{1+\rho_1}}\mathbb{E}_{\varphi^{\sigma,\tau},\varphi^{\sigma,\phi},\xi,\mathcal{G}}\left\{\dfrac{1}{1+\rho_1}\xi^1(1 + \mathcal{G}^1\sqrt{\rho_1})\cdot\right.
		\\\\ \noalign{\vspace{-10pt}}
		& \quad\left.\cdot\tanh{\left[\dfrac{\beta}{1+\rho_1}\xi^1(1 + \mathcal{G}^1\sqrt{\rho_1})A + \a\beta B \varphi^{\sigma,\tau} + \b\beta C \varphi^{\sigma,\phi} \right]}\right\} =
		\\\\ \noalign{\vspace{-10pt}}
		&= \dfrac{1}{(1 + \rho_1)\sqrt{1+\rho_1}}\mathbb{E}_{\varphi^{\sigma,\tau},\varphi^{\sigma,\phi},\xi,\mathcal{G}}\left\{\xi^1\cdot\right.
		\\\\		 \noalign{\vspace{-10pt}}
		& \quad\left.\cdot\tanh{\left[\dfrac{\beta}{1+\rho_1}\xi^1(1 + \mathcal{G}^1\sqrt{\rho_1})A + \a\beta B \varphi^{\sigma,\tau} + \b\beta C \varphi^{\sigma,\phi}\right]}\right\} +
		\\\\ \noalign{\vspace{-10pt}}
		&+\dfrac{\sqrt{\rho_1}}{(1 + \rho_1)\sqrt{1+\rho_1}}\mathbb{E}_{\varphi^{\sigma,\tau},\varphi^{\sigma,\phi},\xi,\mathcal{G}}\left\{\xi^1\mathcal{G}^1\cdot\right.
		\\\\ \noalign{\vspace{-10pt}}
		& \quad\left.\cdot\tanh{\left[\dfrac{\beta}{1+\rho_1}\xi^1A + \dfrac{\xi^1\mathcal{G}^1\sqrt{\rho_1}A\beta}{1+\rho_1} + \a\beta B \varphi^{\sigma,\tau} + \b\beta C \varphi^{\sigma,\phi}\right]}\right\}.
	\end{array}
\end{equation}

\noindent
Regarding the second term of the equation, we can use Wick's theorem on the Gaussian variable $\mathcal{G}$.
\begin{equation}
	\begin{array}{lllll}
		\dfrac{\sqrt{\rho_1}}{(1 + \rho_1)\sqrt{1+\rho_1}}\mathbb{E}_{\varphi^{\sigma,\tau},\varphi^{\sigma,\phi},\xi,\mathcal{G}}\Bigg\{\xi^1\mathcal{G}^1\tanh\Bigg[\dfrac{\beta}{1+\rho_1}\xi^1A + \dfrac{\xi^1\mathcal{G}^1\sqrt{\rho_1}A\beta}{1+\rho_1} +
		\\\\ \noalign{\vspace{-10pt}}
		\qquad\qquad+ \a\beta B \varphi^{\sigma,\tau} + \b\beta C \varphi^{\sigma,\phi}\Bigg]\Bigg\}= 
		\\\\ \noalign{\vspace{-10pt}}
		=\dfrac{\rho_1}{(1 + \rho_1)^2\sqrt{1+\rho_1}}\beta A \left( 1- \qsigma\right).
	\end{array}    
\end{equation}

\noindent
Thus we can re-write the self equation for $ \bar n_{\Xi^{1,a}}^\sigma$ as follows:
\begin{equation}\label{eq:nsigma} 
		\begin{array}{lllll} 
	\bar n_{\xi^{1}}^\sigma = \dfrac{\bar{m}^\sigma_{\xi^1}}{1 + \rho_1}  + \dfrac{\beta\rho_1}{(1 + \rho_1)^2\sqrt{1+\rho_1}} \left( \a \theta^{-1} \sqrt{(1+\rho_{1})(1+\rho_2)}  \bar n_{\eta^{1}}^\tau \right. + 
	\\\\ \noalign{\vspace{-10pt}}
	\qquad\left.+\b\alpha^{-1} \sqrt{(1+\rho_{1})(1+\rho_3)}  \bar n_{\chi^{1}}^\phi\right) \left( 1- \qsigma\right).
	\end{array}
\end{equation}
\noindent
Following the same steps, it is also possible to rewrite the other two example-Mattis magnetization self-equations
\begin{equation}
	\begin{array}{lllll}
		\label{eq:ntau}    
		\bar{n}^\tau_{\eta^{1}} =  \dfrac{\bar{m}^\tau_{\eta^{1}}}{1 + \rho_2}+\dfrac{\beta  \rho_2}{(1 + \rho_2)^2\sqrt{1+\rho_2}}\theta\left(\a\sqrt{(1+\rho_1)(1+\rho_2)}\bar{n}^\sigma_{\Xi^{1}}  \right. + 
		\\\\ \noalign{\vspace{-10pt}}
		\qquad\left.+ \c \alpha^{-1} \sqrt{(1+\rho_2)(1+\rho_3)}\bar{n}^\phi_{\chi^{1}}\right)\left(1 - \qtau\right),
	\end{array}
\end{equation}

\begin{equation}
	\begin{array}{lllll}
		\label{eq:nphi}
		\bar{n}^\phi_{\chi^{1}} = \dfrac{\bar{m}^\phi_{\chi^1}}{1 + \rho_3}+\dfrac{\beta  \rho_3}{(1 + \rho_3)^2\sqrt{1+\rho_3}}\alpha\left(\b\sqrt{(1+\rho_1)(1+\rho_3)}\bar{n}^\sigma_{\Xi^{1}}  \right. + 
		\\\\ \noalign{\vspace{-10pt}}
		\qquad\left.+ \c\theta^{-1} \sqrt{(1+\rho_2)(1+\rho_3)}\bar{n}^\tau_{\eta^{1}}\right)\left(1 - \qphi\right).
	\end{array}
\end{equation}

\vspace{3mm}
\noindent
In accordance with the Lemma~\ref{lemmasumg} and using the oddness of the hyperbolic tangent and the evenness of the Gaussians, we can write the self-equation of $\bar{m}^\sigma_{\xi^1}$ as follows:
{\small
\begin{equation}
	\bar{m}^\sigma_{\xi^1} = \dfrac{1}{\sqrt{1+\rho_1}}\mathbb{E}_{w}\left\{\tanh{\left[   \dfrac{\beta}{1+\rho_1}A +\beta w \sqrt{ \dfrac{\rho_1}{(1+\rho_1)^2}A^2 + \a^2 B^2 + \b^2 C^2  } \right]}\right\} \,.
\end{equation}}

Resolve \eqref{eq:nsigma}, \eqref{eq:ntau} and \eqref{eq:nphi} with respect to $(\bar n_{\xi^{1}}^\sigma,\bar{n}^\tau_{\eta^1},\bar{n}^\phi_{\chi^1})$ we get
\begin{equation}
	\label{eq:n_large_M}
	\begin{array}{lll}
		\bar n_{\xi^{1}}^\sigma=\dfrac{A_1 + A_2 B_1 + A_3 C_1 + A_2 B_3 C_1 + A_3 B_1 C_3 - A_1 B_3 C_3}{1 - A_2 B_2 - A_3 C_2 - A_2 B_3 C_2 - A_3 B_2 C_3 - B_3 C_3},
		\\\\ \noalign{\vspace{-10pt}}
		\bar{n}^\tau_{\eta^1}=\dfrac{B_1 + A_1 B_2 + A_3 B_2 C_1 + B_3 C_1 - A_3 B_1 C_2 + A_1 B_3 C_2}{1 - A_2 B_2 - A_3 C_2 - A_2 B_3 C_2 - A_3 B_2 C_3 - B_3 C_3},
		\\\\ \noalign{\vspace{-10pt}}
		\bar{n}^\phi_{\chi^1}=\dfrac{C_1 - A_2 B_2 C_1 + A_1 C_2 + A_2 B_1 C_2 + B_1 C_3 + A_1 B_2 C_3}{1 - A_2 B_2 - A_3 C_2 - A_2 B_3 C_2 - A_3 B_2 C_3 - B_3 C_3},
	\end{array}
\end{equation}

where we use
{\small
\begin{equation}
	\begin{array}{lllll}
		A_1=\dfrac{\bar{m}_{\xi^1}^\sigma}{1 + \rho_1},&A_2=\dfrac{\beta \rho_1 }{(1 + \rho_1)^2\sqrt{1+\rho_1}} \tilde A_2 (1 - \qsigma),&A_3=\dfrac{\beta \rho_1}{(1 + \rho_1)^2\sqrt{1+\rho_1}} \tilde A_3 (1 - \qsigma),
		\\\\ \noalign{\vspace{-10pt}}
		B_1=\dfrac{\bar{m}_{\eta^1}^\tau}{1 + \rho_2},&B_2=\dfrac{\beta \rho_2 \theta}{(1 + \rho_2)^2\sqrt{1+\rho_2}} \tilde B_2 (1 - \qtau),&B_3=\dfrac{\beta \rho_2 \theta}{(1 + \rho_2)^2\sqrt{1+\rho_2}} \tilde B_3 (1 - \qtau),
		\\\\ \noalign{\vspace{-10pt}}
		C_1=\dfrac{\bar{m}_{\chi^1}^\phi}{1 + \rho_3},&C_2=\dfrac{\beta \rho_3 \alpha}{(1 + \rho_3)^2\sqrt{1+\rho_3}} \tilde C_2 (1 - \qphi),&C_3=\dfrac{\beta \rho_3 \alpha}{(1 + \rho_3)^2\sqrt{1+\rho_3}} \tilde C_3 (1 - \qphi),
	\end{array}
\end{equation}}

where
\begin{equation}
	\begin{array}{lllll}
		\tilde A_2=\dfrac{\a}{\theta} \sqrt{(1 + \rho_1) (1 + \rho_2)} ,&\tilde A_3=\dfrac{\b}{\alpha} \sqrt{(1 + \rho_1) (1 + \rho_3)} ,
		\\\\ \noalign{\vspace{-10pt}}
		\tilde B_2=\a \sqrt{(1 + \rho_1) (1 + \rho_2)}  ,&\tilde B_3=\dfrac{ \c}{\alpha} \sqrt{(1 + \rho_2) (1 + \rho_3)} , 
		\\\\ \noalign{\vspace{-10pt}}
		\tilde C_2=\b \sqrt{(1 + \rho_1) (1 + \rho_3)}  ,&\tilde C_3=\dfrac{\c}{\theta} \sqrt{(1 + \rho_2) (1 + \rho_3)} ,
	\end{array}
\end{equation}

and
\begin{equation}
	\begin{array}{lllll}
		D_1=\sqrt{\gamma\dfrac{1+\rho_{12}}{(1+\rho_1)(1+\rho_2)}},\\\\ \noalign{\vspace{-10pt}}
		D_2=\sqrt{\gamma\dfrac{1+\rho_{13}}{(1+\rho_1)(1+\rho_3)}},\\\\ \noalign{\vspace{-10pt}}
		D_3=\sqrt{\gamma\dfrac{1+\rho_{23}}{(1+\rho_2)  (1+\rho_3)}}.
		
	\end{array}
\end{equation}
This transformation ensures us to express every self-consistency equation as a function of only $(\bar{m}^\sigma_{\xi^1},\bar{m}^\tau_{\eta^1},\bar{m}^\phi_{\chi^1}, \qsigma,\qtau,\qphi)$. Putting all together, the set of 6 self-equations becomes Eqs.~\eqref{eq:selfMagnSigma}-\eqref{eq:selfQPhi}.

%% file: appendices/momentaEval.tex
In this appendix we calculate the first and second moments of Subsec.~\ref{subsec:montecarlo}, with respect to the expectation for the general pattern $\bm{x}$ where we define $\mathbb{E}=\mathbb{E}_{\bm x}\mathbb{E}_{\tilde{\bm x}|\bm x}$.

\subsection*{Supervised setting} \label{appsubsec:evaluationSupSett}

Let's begin the analysis for the Supervised scheme\footnote{
For the sake of brevity, we will present the calculations for the first field only.}, setting:
\begin{equation}
\begin{array}{lll}
    C_{12}= \dfrac{\a}{\sqrt{N_1 N_2}M_1 M_2 r_1 r_2\sqrt{(1+\rho_1)(1+\rho_2)}},
    \\\\ \noalign{\vspace{-10pt}} 
    C_{13}= \dfrac{\a}{\sqrt{N_1 N_3}M_1 M_3 r_1 r_3\sqrt{(1+\rho_1)(1+\rho_3)}}.
    \end{array}
\end{equation}
Using the relation:
\begin{equation}
    \mathbb{E}\Bigg[\Xi_i^{1,a_1} \xi^1_i\Bigg]= \mathbb{E}_{\bm\xi}\left\{\mathbb{E}_{\bm\Xi|\bm \xi}\Bigg[\Xi_i^{1,a_1}\Bigg] \xi^1_i\right\}= \mathbb{E}_{\bm\xi}\left\{r\xi_i^{1}\xi^1_i\right\} = r,
\end{equation}
we obtain
\begin{equation}
\label{eq:m1sup}
    \begin{array}{lll}
         \mu_1=\mathbb{E}\Big[h^{\xi}_{i}\xi^1\Big] &=& C_{12}\SOMMA{j=1}{N_2}\SOMMA{a_1,a_2=1}{M_1,M_2}\mathbb{E}\Bigg[\Theta_j^{1,a_2} \eta_j^1 \Xi_i^{1,a_1} \xi^1_i\Bigg]
         \\\\ \noalign{\vspace{-10pt}}
         &&+C_{13}\SOMMA{a_1,a_3=1}{M_1,M_3}\SOMMA{j=1}{N_3}\mathbb{E}\Bigg[\Upsilon_j^{1,a_3}\chi_j^1\Xi_i^{1,a_1} \xi^1_i \Bigg]
         \\\\  \noalign{\vspace{-10pt}}
         &=& 
         C_{12}M_1 M_2 r_1 r_2 N_2+C_{13}M_1 M_3 r_1 r_3 N_3
         \\\\ \noalign{\vspace{-10pt}}
         &=&
         \dfrac{\a\theta^{-1}}{\sqrt{(1 + \rho_1)(1 + \rho_2)}} + \dfrac{\a\alpha^{-1}}{\sqrt{(1 + \rho_1)(1 + \rho_3)}}.
    \end{array}
\end{equation}
Regarding the second moment, we obtain:
\begin{equation}
    \begin{array}{lll}
         \mu_2=\mathbb{E}\Big[(h^{\xi}_{i})^2\Big] &=& C_{12}^2\SOMMA{\mu=1}{K}\SOMMA{j,k=1}{N_2}\SOMMA{a_1^1,a_2^1=1}{M_1,M_2}\SOMMA{a_1^2,a_2^2=1}{M_1,M_2}\mathbb{E}\Bigg[\Theta_j^{\mu,a_2^1}\Theta_k^{\mu,a_2^2} \eta_j^1\eta_k^1 \Xi_i^{\mu,a_1^1}\Xi_i^{\mu,a_1^2} \Bigg]+
         \\\\ \noalign{\vspace{-10pt}}
         &&+C_{13}^2\SOMMA{\mu=1}{K}\SOMMA{a_1^1,a_3^1=1}{M_1,M_3}\SOMMA{a_1^2,a_3^2=1}{M_1,M_3}\SOMMA{j,k=1}{N_3}\mathbb{E}\Bigg[\Upsilon_j^{\mu,a_3^1}\Upsilon_k^{\mu,a_3^2}\chi_k^1\chi_j^1\Xi_i^{\mu,a_1^1} \Xi_i^{\mu,a_1^2} \Bigg]
         \\\\ \noalign{\vspace{-10pt}}
         &&+2 C_{12}C_{13}\SOMMA{\mu=1}{K}\SOMMA{j,k=1}{N_2, N_3}\SOMMA{a_1^1,a_2=1}{M_1,M_2}\SOMMA{a_1^2,a_3=1}{M_1,M_3}\mathbb{E}\Bigg[\Upsilon_j^{\mu,a_2} \eta_j^1 \Upsilon_k^{\mu,a_3}\chi_k^1\Xi_i^{\mu,a_1^1} \Xi_i^{\mu,a_1^2} \Bigg].
    \end{array}
\end{equation}
Let us start studying the case $\mu=1$:
\begin{equation}
    \begin{array}{lll}
         \mu_2=\mathbb{E}\Big[(h^{\xi}_{i})^2\Big] &=& C_{12}^2\SOMMA{j,k=1}{N_2}\SOMMA{a_1^1,a_2^1=1}{M_1,M_2}\SOMMA{a_1^2,a_2^2=1}{M_1,M_2}\mathbb{E}\Bigg[\Theta_j^{1,a_2^1}\Theta_k^{1,a_2^3} \eta_j^1\eta_k^1 \Xi_i^{1,a_1^1}\Xi_i^{1,a_1^2} \Bigg]
         \\\\ \noalign{\vspace{-10pt}}
         &&+C_{13}^2\SOMMA{a_1^1,a_3^1=1}{M_1,M_3}\SOMMA{a_1^2,a_3^2=1}{M_1,M_3}\SOMMA{j,k=1}{N_3}\mathbb{E}\Bigg[\Upsilon_j^{1,a_3^1}\Upsilon_k^{1,a_3^2}\chi_k^1\chi_j^1\Xi_i^{1,a_1^1} \tilde\xi_i^{1,a_1^2} \Bigg]
         \\\\ \noalign{\vspace{-10pt}}
         && +2 C_{12}C_{13}\SOMMA{j,k=1}{N_2, N_3}\SOMMA{a_1^1,a_2=1}{M_1,M_2}\SOMMA{a_1^2,a_3=1}{M_1,M_3}\mathbb{E}\Bigg[\Theta_j^{1,a_2} \eta_j^1 \Upsilon_k^{1,a_3}\chi_k^1\Xi_i^{1,a_1^1} \Xi_i^{1,a_1^2} \Bigg].
    \end{array}
\end{equation}
Now, noting that the following holds
\begin{equation}
\begin{array}{lll}
     \SOMMA{a_1^1=1}{M_1}\SOMMA{a_1^2=1}{M_1}\mathbb{E}\Bigg[ \Xi_i^{1,a_1^1}\Xi_i^{1,a_1^2} \Bigg]&=&\SOMMA{a_1^1=1}{M_1}\mathbb{E}\Bigg[ (\Xi_i^{1,a_1^1})^2 \Bigg]+\SOMMA{a_1^1=1}{M_1}\SOMMA{a_1^2=1, (a_1^2\neq a_1^1)}{M_1}\mathbb{E}\Bigg[ \Xi_i^{1,a_1^1}\Xi_i^{1,a_1^2} \Bigg]
     \\\\ \noalign{\vspace{-10pt}}
     &=&M_1+M_1(M_1-1)r_1^2 = M_1^2 r_1^2\left(1+\rho_1\right),
\end{array}
\end{equation}
\noindent
and
\vspace{3mm}

\begin{equation}
    \SOMMA{a_2^1=1}{M_2}\SOMMA{a_2^2=1}{M_2}\mathbb{E}\Bigg[ \Xi_j^{1,a_2^1}\Xi_k^{1,a_2^2} \Biggl]=\SOMMA{a_2^1=1}{M_2}\mathbb{E}\Bigg[ \Xi_j^{1,a_2^1} \Bigg]\SOMMA{a_2^2=1}{M_2}\mathbb{E}\Bigg[\Xi_k^{1,a_2^2} \Bigg] = M_2^2 r_2^2
\end{equation}
we can compute
\begin{equation}
    \begin{array}{lll}
         \mu_2=\mathbb{E}\Big[(h^{\xi}_{i})^2\Big] &=& C_{12}^2 N_2^2 r_1^2M_2^2r_2^2(1+\rho_1)
         \\\\ \noalign{\vspace{-10pt}}
         &&+C_{13}^2 N_3^2 r_1^2M_3^2r_3^2(1+\rho_1)
         \\\\ \noalign{\vspace{-10pt}}
         && +2 C_{12}C_{13}N_2 N_3 r_1^2 M_1^2 r_2 r_3 M_2 M_3 (1+\rho_1)
         \\\\ \noalign{\vspace{-10pt}}
         &=& (1+\rho_1) M_1^2 r_1^2 (C_{12}r_2 M_2+C_{13} r_3 M_3)^2 = (1+\rho_1) \mu_1^2.
    \end{array}
\end{equation}

\vspace{3mm}
\noindent
For the case $\mu>1$, we need to check $j=k$ (otherwise, the result is zero due to orthogonality):
\begin{equation}
\label{eq:m2sup}
    \begin{array}{lll}
         \mu_2=\mathbb{E}\Big[(h^{\xi}_{i})^2\Big] &=& C_{12}^2\SOMMA{\mu>1}{K}\SOMMA{j=1}{N_2}\SOMMA{a_1^1,a_2^1=1}{M_1,M_2}\SOMMA{a_1^2,a_2^2=1}{M_1,M_2}\mathbb{E}\Bigg[\Theta_j^{\mu,a_2^1}\Theta_j^{\mu,a_2^3} \eta_j^1\eta_j^1 \Xi_i^{\mu,a_1^1}\Xi_i^{\mu,a_1^2} \Bigg]
         \\\\ \noalign{\vspace{-10pt}}
         &&+C_{13}^2 \SOMMA{\mu>1}{K}\SOMMA{a_1^1,a_3^1=1}{M_1,M_3}\SOMMA{a_1^2,a_3^2=1}{M_1,M_3}\SOMMA{j=1}{N_3}\mathbb{E}\Bigg[\Upsilon_j^{\mu,a_3^1}\Upsilon_j^{\mu,a_3^2}\chi_j^1\chi_j^1\Upsilon_i^{\mu,a_1^1} \Xi_i^{\mu,a_1^2} \Bigg]
         \\\\ \noalign{\vspace{-10pt}}
         &=&M_1^2 r_1^2(1+\rho_1)\left[ C_{12}^2 K N_2  M_2^2  r_2^2(1+\rho_2)+C_{13}^2 K N_3  M_3^2  r_3^2(1+\rho_3)\right]
    \end{array}
\end{equation}
Finally, using the definitions of $C_{12}$ and $C_{13}$, we obtain:
\begin{equation}
    \mu_2-\mu_1^2 = \rho_1 \mu_1^2 +\dfrac{K}{N_1}\left[ \a^2 +\b^2\right].
\end{equation}

\subsection*{Unsupervised setting} \label{appsubsec:evaluationUnsupSett}
We begin by calculating the first two moments of the external field, which will then be useful to approximate the magnetization of the $\bm\sigma$ layer.
We start by computing $\mu_1$:
\begin{align}
\mu_1= \mathbb{E}\Bigg[\SOMMA{\mu=1}{K}\Bigg(\dfrac{\a}{r_1r_2\sqrt{N_1 N_2(1+\rho_1)(1+\rho_2)}\,M}\SOMMA{j, a=1}{N_2,M}\Xi_i^{\mu,a}\Theta_j^{\mu,a}\eta_j^1\xi^1+ \nonumber
\\ \noalign{\vspace{-0pt}}
\qquad+\dfrac{\b}{r_1r_3\sqrt{N_1 N_3(1+\rho_1)(1+\rho_3)}\,M}\SOMMA{k,a=1}{N_3,M}\Xi_i^{\mu,a}\Upsilon_k^{\mu,a}\chi_k^1\xi^1\Bigg)\Bigg]
\end{align}

\noindent
We can define
\begin{align}
	&C_{12}= \dfrac{\a}{\sqrt{N_1 N_2}M r_1 r_2\sqrt{(1+\rho_1)(1+\rho_2)}},\\ 
	&C_{13}= \dfrac{\b}{\sqrt{N_1 N_3}M r_1 r_3\sqrt{(1+\rho_1)(1+\rho_3)}}.
\end{align}

\noindent
The expression of $\mu_1$ can then be divided into four terms, the first two for $\mu = 1$ and the last two for $\mu > 1$:
\begin{align}
	\mu_1 &= \mathbb{E}\left[C_{12}\SOMMA{j, a=1}{N_2,M}\Xi_i^{1,a}\Theta_j^{1,a}\eta_j^1\xi^1\right] 
	+ \mathbb{E}\left[ C_{13}\SOMMA{k,a=1}{N_3,M}\Xi_i^{1,a}\Upsilon_k^{1,a}\chi_k^1\xi^1 \right]
	\\ 
	&+\mathbb{E}\left[\SOMMA{\mu>1}{K}\left(C_{12}\SOMMA{j, a=1}{N_2,M}\Xi_i^{\mu,a}\Theta_j^{\mu,a}\eta_j^1\xi^1\right)\right]
	+ \mathbb{E}\left[\SOMMA{\mu>1}{K}\left(  C_{13}\SOMMA{k,a=1}{N_3,M}\Xi_i^{\mu,a}\Upsilon_k^{\mu,a}\chi_k^1\xi^1 \right)\right]
\end{align}
where the last two term are zero given that $\mathbb{E}\left[{{\xi}}_{k}^{\mu}\eta_{k}^{\nu}\right]=\delta^{\mu\nu}$ and $\mathbb{E}\left[{{\xi}}_{k}^{\mu}\chi_{k}^{\nu}\right]=\delta^{\mu\nu}$.

\noindent
$\mu_1$ then reads as:
\begin{equation}\label{eq:mom1Campo}
	\mu_1 = \dfrac{\a\theta^{-1}}{\sqrt{(1+\rho_1)(1+\rho_2)}}+\dfrac{\b\alpha^{-1}}{\sqrt{(1+\rho_1)(1+\rho_3)}}
\end{equation}

\vspace{6mm}
\noindent
The following four relationships will be useful in calculating the second moment:
\begin{equation}
	\mathbb{E}\Bigg[\Xi_i^{1,a_1} \xi^1_i\Bigg]= \mathbb{E}_{\bm\xi}\left\{\mathbb{E}_{\tilde{\bm\xi}|\bm \xi}\Bigg[\Xi_i^{1,a_1}\Bigg] \xi^1_i\right\}= \mathbb{E}_{\bm\xi}\left\{r\xi_i^{1}\xi^1_i\right\} = r;
\end{equation}

\begin{equation}
	\begin{array}{lll}
		\SOMMA{a_1=1}{M}\SOMMA{a_2=1}{M}\mathbb{E}\Bigg[ \Xi_i^{1,a_1}\Xi_i^{1,a_2} \Bigg]&=&\SOMMA{a_1=1}{M}\mathbb{E}\Bigg[ (\Xi_i^{1,a_1})^2 \Bigg]+\SOMMA{a_1=1}{M}\SOMMA{a_2=1, (a_1\neq a_2)}{M}\mathbb{E}\Bigg[ \Xi_i^{1,a_1}\Xi_i^{1,a_2} \Bigg]
		\\\\ \noalign{\vspace{-10pt}}
		&=&M+M(M-1)r_1^2 = M^2 r_1^2\left(1+\rho_1\right);
	\end{array}
\end{equation}
\begin{equation}
    \SOMMA{a_1=1}{M}\SOMMA{a_2=1}{M}\mathbb{E}\Bigg[ \Xi_j^{1,a_1}\Xi_k^{1,a_2} \Bigg]=\SOMMA{a_1=1}{M}\mathbb{E}\Bigg[ \Xi_j^{1,a_1} \Bigg]\SOMMA{a_2=1}{M}\mathbb{E}\Bigg[\Xi_k^{1,a_2} \Bigg] = M^2 r_2^2;
\end{equation}

\begin{equation}
	\mathbb{E}\left[ \left( \SOMMA{a_1,a_2=1}{M,M} \Xi_i^{\mu,a_1} \Theta_j^{\mu,a_2} \right)^2\right] = M^2r_1^2r_2^2(1+\rho_{12}).
\end{equation}

\noindent
So, we can now calculate the second moment, obtaining:

\begin{equation}
	\begin{array}{lll}
		\mu_2=\mathbb{E}\Big[(h^{\sigma}_{i})^2\Big] &= C_{12}^2\SOMMA{\mu=1}{K}\SOMMA{j,k=1}{N_2,N_2}\SOMMA{a_1,a_2=1}{M,M}\mathbb{E}\Bigg[\Theta_j^{\mu,a_1}\Theta_k^{\mu,a_2} \eta_j^1\eta_k^1 \Xi_i^{\mu,a_1}\Xi_i^{\mu,a_2} \Bigg]
		\\\\ \noalign{\vspace{-10pt}} 
		&\quad+C_{13}^2\SOMMA{\mu=1}{K}\SOMMA{j,k=1}{N_3,N_3}\SOMMA{a_1,a_3=1}{M,M}\mathbb{E}\Bigg[\Upsilon_j^{\mu,a_1}\Upsilon_k^{\mu,a_3}\chi_k^1\chi_j^1\Xi_i^{\mu,a_1} \Xi_i^{\mu,a_3} \Bigg]
		\\\\ \noalign{\vspace{-10pt}}
		&\quad+2 C_{12}C_{13}\SOMMA{\mu=1}{K}\SOMMA{j,k=1}{N_2, N_3}\SOMMA{a_2,a_3=1}{M,M}\mathbb{E}\Bigg[\Theta_j^{\mu,a_2} \eta_j^1 \Upsilon_k^{\mu,a_3}\chi_k^1\Xi_i^{\mu,a_2} \Xi_i^{\mu,a_3} \Bigg] .
	\end{array}
\end{equation}

\noindent
Let's start studying the case $\mu=1$:
\begin{equation}
	\begin{array}{lll}
		\mu_2=\mathbb{E}\Big[(h^{\sigma}_{i})^2\Big] &= C_{12}^2\SOMMA{j,k=1}{N_2,N_2}\SOMMA{a_1,a_2=1}{M,M}\mathbb{E}\Bigg[\Theta_j^{1,a_1}\Theta_k^{1,a_2} \eta_j^1\eta_k^1 \Xi_i^{1,a_1}\Xi_i^{1,a_2} \Bigg]
		\\\\ \noalign{\vspace{-10pt}}
		&\quad+C_{13}^2\SOMMA{j,k=1}{N_3,N_3}\SOMMA{a_1,a_3=1}{M,M}\mathbb{E}\Bigg[\Upsilon_j^{1,a_1}\Upsilon_k^{1,a_3}\chi_k^1\chi_j^1\Xi_i^{1,a_1} \Xi_i^{1,a_3} \Bigg]
		\\\\ \noalign{\vspace{-10pt}}
		&\quad+2 C_{12}C_{13}\SOMMA{j,k=1}{N_2, N_3}\SOMMA{a_2,a_3=1}{M,M}\mathbb{E}\Bigg[\Theta_j^{1,a_2} \eta_j^1 \Upsilon_k^{1,a_3}\chi_k^1\Xi_i^{1,a_2} \Xi_i^{1,a_3} \Bigg]
        \\\\ \noalign{\vspace{-10pt}}
        &= C_{12}^2 N_2^2 r_2^2M^2r_1^2(1+\rho_1)
		\\\\ \noalign{\vspace{-10pt}}
		&\quad+C_{13}^2 N_3^2 r_3^2 M^2r_1^2(1+\rho_1)
		\\\\ \noalign{\vspace{-10pt}}
		&\quad +2 C_{12}C_{13}N_2 N_3 r_1^2 M^4 r_2 r_3  (1+\rho_1)
		\\\\ \noalign{\vspace{-10pt}}
		&= (1+\rho_1) M^2 r_1^2 (C_{12}r_2 M+C_{13} r_3 M)^2 = (1+\rho_1) \mu_1^2,
        
	\end{array}
\end{equation}
while for $\mu>1$ it must be $j=k$ (in other cases, we would obtain zero due to orthogonality):
\begin{equation}\label{eq:mom2Campo}
	\begin{array}{lll}
		\mu_2=\mathbb{E}\Big[(h^{\sigma}_{i})^2\Big] &= C_{12}^2\SOMMA{\mu>1}{K}\SOMMA{j=1}{N_2}\SOMMA{a_1,a_2=1}{M,M}\mathbb{E}\Bigg[\Theta_j^{\mu,a_1}\Theta_j^{\mu,a_2} \eta_j^1\eta_j^1 \Xi_i^{\mu,a_1}\Xi_i^{\mu,a_2} \Bigg]
		\\\\ \noalign{\vspace{-10pt}}
		&+C_{13}^2\SOMMA{\mu>1}{K}\SOMMA{j=1}{N_3}\SOMMA{a_1,a_3=1}{M,M}\mathbb{E}\Bigg[\Upsilon_j^{\mu,a_1}\Upsilon_j^{\mu,a_3}\chi_j^1\chi_j^1\Xi_i^{\mu,a_1} \Xi_i^{\mu,a_3} \Bigg]
		\\\\ \noalign{\vspace{-10pt}}
		&=C_{12}^2 K N_2 M^2r_1^2 r_2^2(1-\rho_{12})+C_{13}^2 K N_3 M^2r_1^2 r_3^2(1-\rho_{13})
		\\\\ \noalign{\vspace{-10pt}}
		&= K M^2 r_1^2 \left[ C_{12}^2N_2r_2^2(1-\rho_{12}) + C_{13}^2N_3r_3^2(1-\rho_{13}) \right].
	\end{array}
\end{equation}

\noindent
Putting it all together and using the definitions of $C_{12}$ and $C_{13}$
\begin{equation}
	\mu_2-\mu_1^2 = \rho_1 \mu_1^2 +\dfrac{K}{N_1(1+\rho_1)}\left[ \a^2\dfrac{1+\rho_{12}}{1+\rho_2} + \b^2\dfrac{1+\rho_{13}}{1+\rho_3}\right].
\end{equation}